\RequirePackage{fix-cm}
\documentclass[smallextended]{svjour3}       
\smartqed  

 \usepackage{mathptmx}      
\usepackage[font={small},labelfont={small,bf}]{caption}
\usepackage{latexsym}
\usepackage{forest}
\usepackage{graphicx}
\usepackage{braket}
\usepackage{multirow}
\usepackage{todonotes}
\usepackage[noadjust]{cite}
\usepackage{booktabs}          
\usepackage{prftree}
\usepackage{fancybox,framed}
\usepackage[linesnumbered,ruled,vlined]{algorithm2e}

\usepackage{amsthm}
\usepackage{mathpartir}

\usepackage{mdframed}

\usepackage{amssymb}
\usepackage{apptools}
\usepackage{mathtools}
\usepackage{chngcntr}
\usepackage{enumitem}
\usepackage{makecell}
\usepackage{dblfloatfix}

\usepackage[cal = txupr]{mathalpha}

\usepackage[colorlinks = true,
            linkcolor = blue,
            urlcolor  = blue,
            citecolor = blue]{hyperref}

\usepackage{pifont}
\newcommand{\cmark}{\ding{51}}%
\newcommand{\xmark}{\ding{55}}%

\usepackage[mathlines,displaymath,pagewise]{lineno}

\newtheorem{thm}{Theorem}
\newtheorem{defi}[thm]{Definition}
\newtheorem{lem}[thm]{Lemma}
\newtheorem{col}[thm]{Corollary}
\newtheorem{rem}[thm]{Remark}

\theoremstyle{theorem}
\newtheorem{ques}{Question}

\DeclareMathOperator{\SelectNC}{SelectNC}

\DeclareMathOperator{\TRes}{TRes}

\DeclareMathOperator{\Max}{Max}
\DeclareMathOperator{\Dep}{dep}
\DeclareMathOperator{\Var}{var}
\DeclareMathOperator{\ComT}{ComT}
\DeclareMathOperator{\Sep}{Sep}
\DeclareMathOperator{\Factor}{Factor}
\DeclareMathOperator{\T-Res}{T-Res}
\DeclareMathOperator{\TTrans}{T-Trans}
\DeclareMathOperator{\Trans}{Trans}

\DeclareMathOperator{\mgu}{mgu}

\DeclareMathOperator{\Simplify}{Smp}
\DeclareMathOperator{\Pick}{Pick}

\begin{document}

\title{Saturation-based Boolean conjunctive query answering and rewriting for the guarded quantification fragments 
}

\titlerunning{Saturation-based methods for querying the guarded quantification fragments}        

\author{Sen Zheng         \and
        Renate A.~Schmidt 
}


\institute{
Sen Zheng \and Renate A.~Schmidt \at
Department of Computer Science \\ 
University of Manchester, United Kingdom 
}

\date{Received: date / Accepted: date}

\maketitle
\begin{abstract}
Query answering is an important problem in AI, database and knowledge representation. In this paper, we develop saturation-based Boolean conjunctive query answering and rewriting procedures for the guarded, the loosely guarded and the clique-guarded fragments. Our query answering procedure improves existing resolution-based decision procedures for the guarded and the loosely guarded fragments and this procedure solves Boolean conjunctive query answering problems for the guarded, the loosely guarded and the clique-guarded fragments. Based on this query answering procedure, we also introduce a novel saturation-based query rewriting procedure for these guarded fragments. Unlike mainstream query answering and rewriting methods, our procedures derive a compact and reusable saturation, namely a closure of formulas, to handle the challenge of querying for distributed datasets. This paper lays the theoretical foundations for the first automated deduction decision procedures for Boolean conjunctive query answering and the first saturation-based Boolean conjunctive query rewriting in the guarded, the loosely guarded and the clique-guarded fragments.


\end{abstract}

\section{Introduction}
\label{sec:intro}

The problem of answering \emph{conjunctive queries} \cite{AHV95,U89} over logical constraints is at the heart of knowledge representation and database research. This problem can be reduced to that of \emph{Boolean conjunctive query} (\emph{\textsf{BCQ}}) \emph{answering} by instantiating free variables in conjunctive queries with facts from databases. Problems in many fields of computer science such as constraint satisfaction problems \cite{FV93,KV00}, homomorphism problems~\cite{CMP77} and query evaluation and containment problems~\cite{CMP77} can be recast as Boolean conjunctive query answering problems \cite{V00}. Our interest in this paper is to develop practical methods and inference systems that can provide the basis for the following problems:
\begin{itemize}
	\item answering \textsf{BCQ}s for the \emph{guarded}, the \emph{loosely guarded} and the \emph{clique-guarded fragments}, and
	\item saturation-based rewriting of \textsf{BCQ}s for these guarded fragments.
\end{itemize}

The \emph{guarded fragment} (\emph{\textsf{GF}}) and the \emph{loosely guarded fragment} (\emph{\textsf{LGF}}) are introduced in \cite{vB97,ANvB98} as generalised \emph{modal fragments} of \emph{first-order logic} (\emph{\textsf{FOL}}). In a \emph{guarded formula} the free variables of quantified formulas are `guarded' by an atom. Strictly extended from \textsf{GF}, the loosely guarded fragment \textsf{LGF}, which is also known as the \emph{pairwise guarded fragment} \cite{vB97,AMdNdR99}, pairwise `guards' the free variables of quantified formulas using a conjunction of atoms. This conjunction is called a \emph{loose guard} where the variables form a `clique'. Further \textsf{LGF} has been extended to the \emph{clique-guarded fragment} (\emph{\textsf{CGF}})~\cite{G99b}, in which the `cliques' are extended with branches. In \cite{Ho02,M07} \textsf{CGF} is called the \emph{packed fragment}. A common characteristic of \textsf{GF}, \textsf{LGF} and \textsf{CGF} is that the free variables of quantified formulas need to be guarded; therefore we collectively refer to these fragments as the \emph{guarded quantification fragments}. These fragments are decidable and have well-behaved computational properties \cite{vB97,ANvB98,M07,G99a,G99b,Ho02,HM02,DaL15}. \textbf{Figure~\ref{fig:gf}} shows the relationship between the guarded quantification fragments, (negated) \textsf{BCQ}s targeted in this paper and \textsf{FOL}.

\begin{figure}[t]
\center
\includegraphics[width=.85\textwidth]{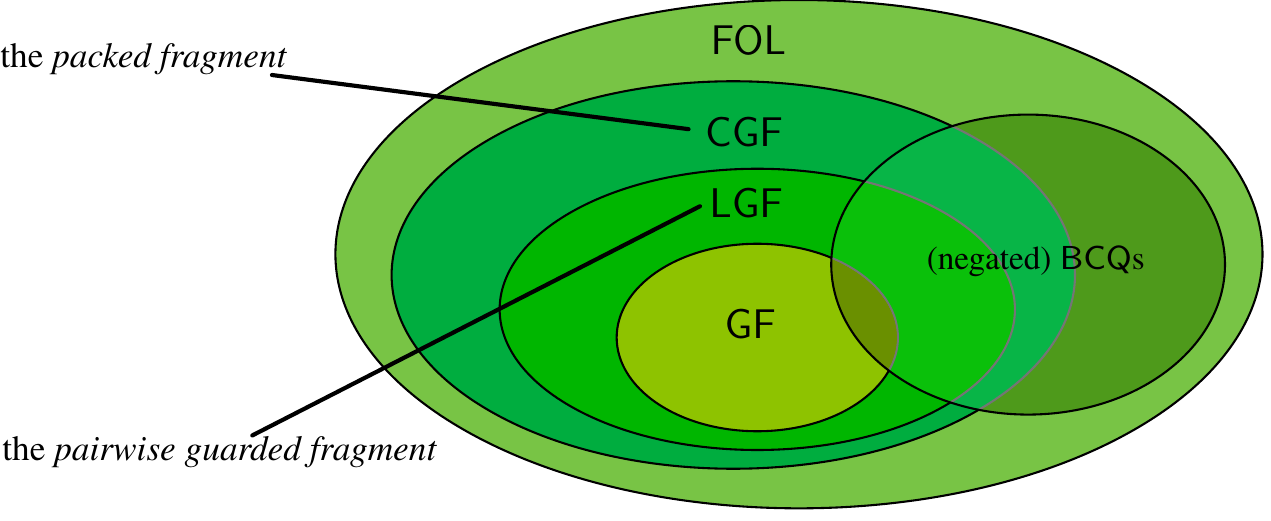}
\caption{The relationship of the guarded quantification fragments, (negated) \textsf{BCQ}s and first-order logic}
\label{fig:gf}
\end{figure} 


The computational complexity of the \textsf{BCQ} answering problem for \textsf{GF} is \textsc{2ExpTime}-complete~\cite{BGO14}. For \textsf{LGF} and \textsf{CGF} the complexity of the \textsf{BCQ} answering problem is also \textsc{2ExpTime}-complete, as in both cases the problem is reducible to the satisfiability checking problem of the \emph{clique-guarded negation fragment} \cite{BtCS15}.\footnote{This paper does not consider the clique-guarded negation fragment.} \textbf{Figure~\ref{fig:result}} lists important known properties of the guarded quantification fragments where \cmark \ and~\xmark \ respectively denote positive and negative results. In the Satisfiability checking column of \textbf{Figure~\ref{fig:result}}, we assume that the fragments have a \emph{fixed signature}.


\begin{figure}[t]
\center
\normalsize
\begin{tabular}{ |c|c|c|c|c|c|c| }
\hline
& \thead{Decidability} & {\thead{Satisfiability \\ checking}} & {\thead{Tree-like \\ model}} & \thead{Finite \\ model} & \thead{\textsf{BCQ} \\ answering} & \thead{\textsf{FO} rewritable \\ (with \textsf{BCQ}s)}\\
\hline

\thead{\textsf{GF}} & {\thead{\cmark \ \cite{vB97,ANvB98}}} & {\thead{\textsc{ExpTime} \\ \cite{G99a}}} & {\thead{\cmark \ \cite{G99a}}} & {\thead{\cmark \ \cite{G99a}}} & \thead{\textsc{2ExpTime} \\ \cite{BGO14}} & \multirow{5.5}{*}{\thead{\xmark \ \cite{BBLP18,BBGP21}}}\\

\cline{1-6}
\thead{\textsf{LGF}} & {\thead{\cmark \ \cite{vB97}}} & {\thead{\textsc{ExpTime} \\  \cite{G99a}}} & {\thead{\cmark \ \cite{G99b}}} & {\thead{\cmark \ \cite{Ho02}}} & \multirow{3.5}{*}{\thead{\thead{\textsc{2ExpTime} \\  \cite{BtCS15}}}} & \\
\cline{1-5}
\thead{\textsf{CGF}} & {\thead{\cmark \ \cite{G99a,M07}}} & {\thead{\textsc{ExpTime} \\  \cite{G99a,M07}}} & {\thead{\cmark \ \cite{G99b}}} & {\thead{\cmark \ \cite{M07,Ho02}}} &  & \\
\hline
\end{tabular}
\caption{Known properties of the guarded quantification fragments}
\label{fig:result}
\end{figure}

Resolution-based procedures have been devised for deciding satisfiability in \textsf{GF} in \cite{dNdR03,GdN99} and for \textsf{LGF} in \cite{dNdR03,GdN99,ZS20a}. Tableau-based procedures have been devised for deciding satisfiability in \textsf{GF} \cite{H02} and \textsf{CGF}~\cite{HT01}. However, querying poses a major problem, since neither \textsf{BCQ} nor its negation belongs to the guarded quantification fragments~(see \textbf{Figure~\ref{fig:gf}}). Indeed, so far it appears that there has been no effort to extend these methods to solving the \textsf{BCQ} answering problems for any of the guarded quantification fragments, even if the aforementioned complexity results mean that in theory, these querying problems are decidable. Introducing new techniques, this paper develops decision procedures to answer  \textsf{BCQ}s for all the guarded quantification fragments. Our initial work for solving the \textsf{BCQ} answering problem for Horn \textsf{LGF} was published in~\cite{ZS20a} and for \textsf{GF} was published in~\cite{ZS20b}. 

\textbf{Figure~\ref{fig:q_ans}} illustrates the idea of our query answering procedure. Given a set~$\Sigma$~of rules, a dataset $D$ and a \textsf{BCQ} $q$, checking whether $\Sigma \cup D \models q$ is equivalent to checking unsatisfiability of $\{\lnot q\} \cup \Sigma \cup D$. To decide $\{\lnot q\} \cup \Sigma \cup D$, we transform it into a clausal form. In particular,~$\Sigma$ and $D$ are mapped to \emph{loosely guarded clauses} and $\lnot q$ to \emph{query clauses}. To perform the saturation process we develop a novel \emph{top-variable inference system}. This system ensures termination when we perform resolution inferences on loosely guarded clauses and query clauses.


\begin{figure}[b]
\center
\includegraphics[width=.9\textwidth]{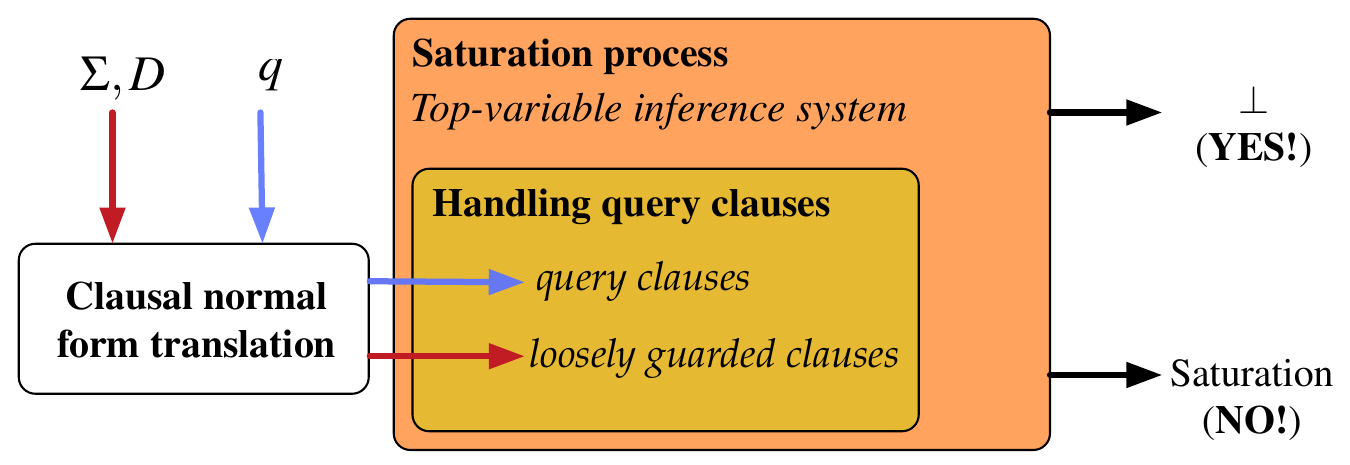}
\caption{Saturation-based \textsf{BCQ} answering processing of a set of guarded quantification formulas $\Sigma$, a dataset $D$ and a \textsf{BCQ} $q$}
\label{fig:q_ans}
\end{figure}


Conventional \textsf{BCQ} rewriting tasks aim to reduce a \emph{\textsf{BCQ} entailment problem} to a \emph{model checking problem}: one first compiles a \textsf{BCQ}~$q$ and a set~$\Sigma$ of formulas into a~(function-free) first-order formula $\Sigma_q$, and then applies a model checking algorithm to~$\Sigma_q$ over datasets~\cite{CGL07,EOS12,GMSGH13}. If this reduction is possible, then $q$ and~$\Sigma$ are called \emph{first-order \textnormal{(}\textsf{FO}\textnormal{)} rewritable}. Counter-examples in~\cite{BBLP18,BBGP21} imply that this property does not hold for \textsf{BCQ}s for any of the guarded quantification fragments. To address this problem, we introduce a new setting of \emph{saturation-based query rewriting}. This rewriting reduces the query answering problem $\Sigma \cup D \models q$ to the entailment problem $D \models \Sigma_q$, where $\Sigma_q$ is a first-order formula. Our query rewriting method pre-saturates the clausal form of $\{\lnot q\} \cup \Sigma$ and does it in such a way that this pre-saturation can be restored to a first-order formula $\Sigma_q$. Using our method, any dataset $D$ can be tested against the pre-saturation, but it is also possible to use other reasoning methods such as the chase algorithm~\cite{ABU79,MMS79} to solve the entailment checking of $D \models \Sigma_q$. \textbf{Figure~\ref{fig:rew}} outlines our saturation-based query rewriting procedure, which applies the saturation process to the rules and the query but not the dataset, and back-translates the saturation to a first-order formula $\Sigma_q$.

\begin{figure}[t]
\center
\includegraphics[width=.9\textwidth]{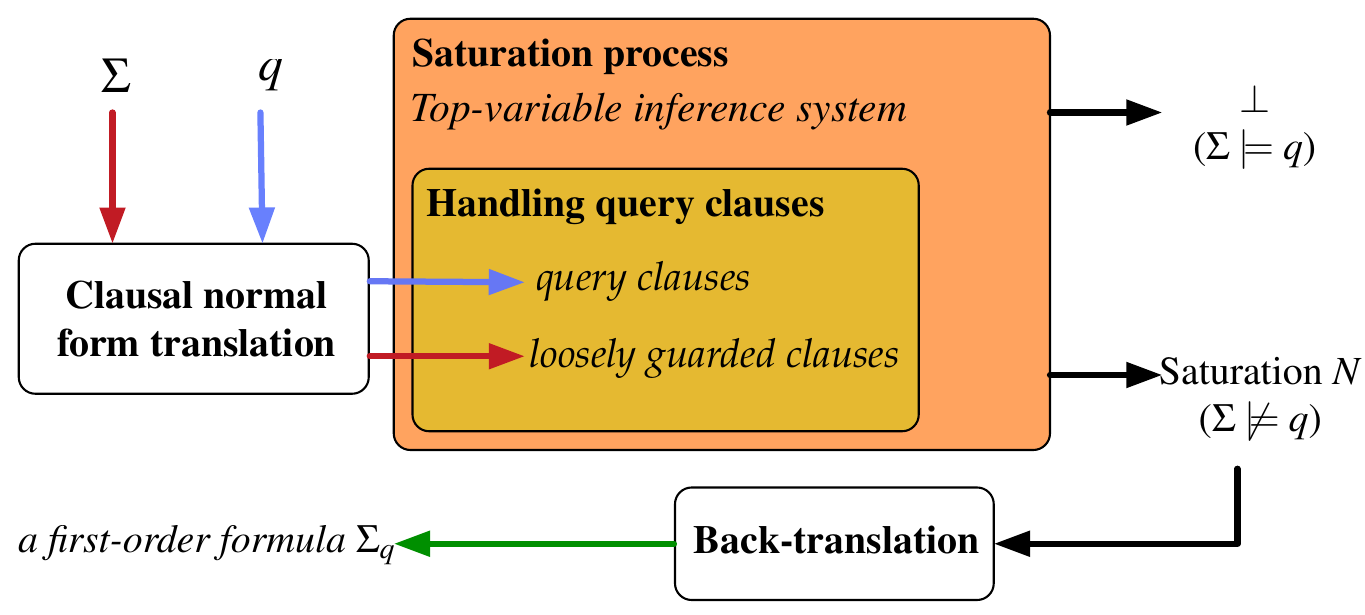}
\caption{Saturation-based \textsf{BCQ} rewriting processing of a set of guarded quantification formulas $\Sigma$ and a \textsf{BCQ} $q$}
\label{fig:rew}
\end{figure}

This result is of independent interest to automated reasoning, as back-translating a clausal set that includes inferred conclusions, to a first-order formula typically fails, as in general this problem is undecidable \cite{GSS08}. Using results established in \cite{E96} that a clausal set can be back-translated into a first-order formula if the clausal set satisfies certain properties, we devise a query rewriting procedure that ensures a successful back-translation. To distinguish our query rewriting setting from the traditional ones, we refer to our approach as \emph{saturation-based query rewriting}. 

To provide a basis for implementation, our query answering and rewriting approaches are devised in line with the \emph{resolution framework} of \cite{BG01}, which provides the basis for powerful saturation-based theorem provers such as E~\cite{S13}, SPASS~\cite{WDFKSW09}, Vampire~\cite{RV01b} and Zipperposition \cite{C15} and a lot of work in automated reasoning~\cite{WTRB20,EP20,SBT20,SM16,RV19}.

In a nutshell, the contributions of this paper are:
\begin{itemize}[noitemsep]
\item Inference systems for deciding \textsf{BCQ} answering for \textsf{GF}, \textsf{LGF} and \textsf{CGF}, dedicated to provide the basis for practical decision procedures.
\item A novel saturation-based \textsf{BCQ} rewriting approach for \textsf{GF}, \textsf{LGF} and \textsf{CGF}.
\item Improvements on existing resolution-based decision approaches for \textsf{GF} and \textsf{LGF}, and the first resolution-based approach for deciding \textsf{CGF}.
\item Novel saturation-based resolution inference systems, namely a \emph{partial selection-based resolution system} and a \emph{top-variable resolution system}, with formal soundness and refutational completeness proofs for first-order clausal logic.
\item Our procedures are applicable to answer and rewrite \textsf{BCQ}s for real-world ontological languages such as guarded, loosely guarded, and clique-guarded Datalog$^\pm$ and the description logic $\mathcal{ALCHOI}$.
\item Novel aspects of our approach include but are not limited to: the \emph{separation rules}, the \emph{partial selection-based} and \emph{top-variable resolution rules}, the \emph{clausification processes} and the \emph{back-translation procedure}. These techniques may allow decision and querying problems for other fragments to be solved in the future.
\end{itemize}

The remainder of this paper is organised as follows. \textbf{Section \ref{sec:pre}} formally defines basic notions of first-order logic, the guarded quantification fragments and the research questions. \textbf{Section~\ref{sec:trans}} reduces the \textsf{BCQ} answering problem for the targeted guarded fragments to an unsatisfiability checking problem of loosely guarded clauses and query clauses. \textbf{Section \ref{sec:tinf}} presents the partial selection-based resolution system and the top-variable resolution system. \textbf{Section~\ref{sec:lgc}} then proves that the top-variable system decides satisfiability of the class of loosely guarded clauses. \textbf{Section \ref{sec:qc}} tackles query clauses by introducing the separation rules and formula renaming. Combining the results from the previous sections, \textbf{Section~\ref{sec:qans}} devises a \textsf{BCQ} answering procedure for the guarded quantification fragments. \textbf{Section \ref{sec:qrew}} develops a saturation-based \textsf{BCQ} rewriting procedure for these guarded fragments. \textbf{Sections~\ref{sec:relate}} and \textbf{\ref{sec:conclu}} discuss related work and conclude the paper, respectively.



\section{Basic notions, guarded fragments and the querying problems of interest}
\label{sec:pre}
\subsection*{\textbf{Basic notions}}
\label{sec:notion}
Let $\mathtt{C}$, $\mathtt{F}$ and $\mathtt{P}$ be countably infinite sets that are pairwise disjoint. The elements in~$\mathtt{C}$,~$\mathtt{F}$ and $\mathtt{P}$ are \emph{constant symbols} (or \emph{constants}), \emph{function symbols} and \emph{predicate symbols}. A predicate symbol of arity zero is a \emph{propositional variable}. We refer the triple $(\mathtt{C}, \mathtt{F}, \mathtt{P})$ as a \emph{signature}. A \emph{term} is either a constant, or a variable, or it has the form of $f(t_1, \ldots, t_n)$ if i) $f$ is a function symbol of arity $n$ and ii) $t_1, \ldots, t_n$ are terms. A \emph{compound term} is a term that is neither a constant nor a variable. An \emph{atom} is an expression $P(t_1, \ldots ,t_n)$, where $P$ is a $n$-ary predicate symbol distinct from $\approx$ and $t_1, \ldots , t_n$ are terms. A \emph{literal} is an atom $A$ or a negated atom~$\lnot A$. Given two terms~(or atoms) $E_1 = A(\ldots, t, \ldots)$ and $E_2 = B(\ldots, s, \ldots)$, we say $t$ \emph{pairs} $s$ if the argument position of $t$ in $E_1$ is the same as that of $s$ in $E_2$. If a signature allows the special predicate symbols $\approx$ and $\not  \approx$, then the setting is \emph{first-order logic with equality}. We use \emph{infix notation} for positive and negative equational atoms: $s \approx t$ and $s \not \approx t$. 

In a quantified first-order formula $\forall x F$ or $\exists x F$, $x$ is the \emph{quantified variable} and~$F$~is \emph{the scope of} the quantified variable $x$. An occurrence of a variable $x$ in a first-order formula $F$ is a \emph{free variable} of $F$ if and only if $x$ is not within the scope of quantified variables. A \emph{sentence} (or \emph{closed formula}) is a first-order formula without free variables. A \emph{first-order clause} (or \emph{clause}) is a multiset of literals, interpreted as a disjunction of literals. A \emph{positive} (\emph{negative}) clause is a clause that contains only positive (negative) literals. An \emph{expression} can be a term, an atom, a literal, or a clause. The set of variables that occur in an expression $E$ is denoted as $\Var(E)$. A variable-free expression is a \emph{ground expression}. A clause is \emph{decomposable} if it can be partitioned into two variable-disjoint subclauses, otherwise, the clause is \emph{indecomposable}.

The \emph{depth of a term} $t$ is denoted $\Dep(t)$ and defined as: i) if $t$ is a variable or a constant, then $\Dep(t) = 0$, ii) if $t$ is a compound term $f(t_1, \ldots, t_n)$, then $\Dep(t) = 1 + max(\{\Dep(t_i) \ | \ 1 \leq i \leq n \})$. The \emph{depth of an expression} $E$ is the depth of the deepest term in $E$, denoted as $\Dep(E)$. The \emph{width of an expression}~$E$ is the number of distinct variables in $E$. If an expression $E$ does not contain any term, then $\Dep(E) = 0$ and the width of $E$ is 0.



A \emph{substitution} of terms for variables is a set $\{x_1 \mapsto t_1, \ldots, x_n \mapsto t_n\}$ where each~$x_i$ is a distinct variable and each $t_i$ is a term, which is not identical to the respective variable $x_i$. We use lower-case Greek letters $\sigma, \theta, \eta$ to denote substitutions. We use~$E\sigma$~to denote the result of the \emph{application of a substitution $\sigma$ to the expression}~$E$. It is also said to be an \emph{instance} of $E$. A \emph{variable renaming} is a substitution $\sigma$ such that $\sigma = \{x_1 \mapsto y_1, \ldots, x_n \mapsto y_n\}$ where $x_1, \ldots, x_n, y_1, \ldots, y_n$ are variables and $\sigma$ is bijective. An expression $E_1$ is a \emph{variant} of an expression $E$ if there exists a variable renaming~$\sigma$ such that $E_1 = E\sigma$. We consider two clauses $C_1$ and $C_2$ to be identical if~$C_1$ is a variant of $C_2$. Given substitutions $\sigma$ and $\theta$, the \emph{composition} $\sigma\theta$ denotes that for each variable~$x$, $x\sigma\theta = (x\sigma)\theta$. A substitution $\sigma$ is a \emph{unifier} of a set $\{E_1, \ldots, E_n\}$ of expressions if and only if $E_1\sigma = \ldots = E_n\sigma$. The set $\{E_1, \ldots, E_n\}$ is said to be \emph{unifiable} if there is a unifier for it. A unifier $\sigma$ of a set $\{E_1, \ldots, E_n\}$ of expressions is a \emph{most general unifier}~(\emph{mgu}) if and only if for each unifier $\theta$ of the set, there is a substitution $\eta$ such that $\sigma = \theta\eta$. A unifier $\sigma$ is a \emph{simultaneous mgu} of two sequences $E_1, \ldots, E_n$ and $E_1^\prime, \ldots, E_n^\prime$ of expressions where $n > 1$,  if $\sigma$ is an mgu for each pair~$E_i$~and $E_i^\prime$. By $\sigma = \mgu(E \doteq E^\prime)$, we mean that $\sigma$ is an mgu of expressions~$E$ and $E^\prime$. By $\sigma = \mgu(E_1 \doteq E_1^\prime, \ldots, E_n \doteq E_n^\prime)$ where $n > 1$, we mean that $\sigma$ is a simultaneous mgu of two sequences $E_1, \ldots, E_n$ and $E_1^\prime, \ldots, E_n^\prime$ of expressions.

We distinguish rules in our paper in two types: i) the rules that are applied to a clausal set, and they are framed using bold lines; ii) the rules that are applies to clauses, namely inference rules, and they are framed using non-bold lines. When we refer to function symbols, we mean non-constant ones. In the rest of the paper, we use the following notational convention:
\begin{align*}
\begin{array}{lll}
\bullet \  x,y,z,u,v,x_1, \ldots \text{for variables} & & 
\bullet \  a, b, c, a_1, \ldots \text{for constant symbols} \\
\bullet \  f, g, h, \ldots \text{for function symbols} & &
\bullet \ P,P_1,A,B, \ldots \text{for predicate symbols} \\
\bullet \ p, p_1, \ldots \ \text{for propositional variables} & &
\bullet \ F, F_1, \ldots \text{for formulas} \\
\bullet \ C, D, Q, C_1, \ldots \text{for clauses} & & \bullet \ s, t, u, \ldots \ \text{for terms} \\
\bullet \ L, L_1, \ldots \text{for literals} & & \bullet \ A, B, G, G_1, \ldots \ \text{for atoms}
\end{array}	
\end{align*}

\subsection*{\textbf{Guarded quantification fragments}}
In the following definitions, constants are allowed but not equality.

\begin{defi}
\label{def:gf}	
The \emph{guarded fragment} (\emph{\textsf{GF}}) is a fragment of first-order logic without function symbols, inductively defined as follows:
\begin{enumerate}
\item $\top$ and $\bot$ belong to \textsf{GF}.
\item If $A$ is an atom, then $A$ belongs to \textsf{GF}.
\item \textsf{GF} is closed under Boolean connectives.
\item Let $F$ be a guarded formula and $G$ an atom. Then $\exists \overline x (G \land F)$ and $\forall \overline x (G \to F)$ belong to \textsf{GF} if all free variables of $F$ occur in $G$. 
\end{enumerate}      
\end{defi}

\begin{defi}
\label{def:lgf}	
The \emph{loosely guarded fragment}~(\emph{\textsf{LGF}}) is a fragment of first-order logic without function symbols, inductively defined as follows: 
\begin{enumerate}[noitemsep]
\item $\top$ and $\bot$ belong to \textsf{LGF}.
\item If $A$ is an atom, then $A$ belongs to \textsf{LGF}.
\item \textsf{LGF} is closed under Boolean connectives.
\item Let $F$ be a loosely guarded formula and $\mathbb{G}$ a conjunction of atoms. Then $\forall \overline x (\mathbb{G} \to F)$ and $\exists \overline x (\mathbb{G} \land F)$ belong to \textsf{LGF} if
\begin{enumerate}[noitemsep]
\item all free variables of $F$ occur in $\mathbb{G}$, and 
\item for each variable $x$ in $\overline x$ and each variable $y$ occurring in $\mathbb{G}$ that is distinct from $x$, $x$ and $y$ co-occur in an atom of $\mathbb{G}$.
\end{enumerate}
\end{enumerate}
\end{defi}

\begin{defi}
\label{def:cgf}
The \emph{clique-guarded fragment} (\emph{\textsf{CGF}}) is a fragment of first-order logic without function symbols, inductively defined as follows:
\begin{enumerate}[noitemsep]
\item $\top$ and $\bot$ belong to \textsf{CGF}.
\item If $A$ is an atom, then $A$ belongs to \textsf{CGF}.
\item \textsf{CGF} is closed under Boolean connectives.
\item Let $F$ be a clique-guarded formula and $\mathbb{G}(\overline x, \overline y)$ a conjunction of atoms. Then $\forall \overline z (\exists \overline x \mathbb{G}(\overline x, \overline y) \to F)$ and $\exists \overline z (\exists \overline x \mathbb{G}(\overline x, \overline y) \land F)$ belong to \textsf{CGF}, if 
\begin{enumerate}[noitemsep]
\item all free variables of $F$ occur in $\overline y$, and
\item each variable in $\overline x$ occurs in only one atom of $\mathbb{G}(\overline x, \overline y)$, and
\item for each variable $z$ in $\overline z$ and each variable $y$ occurring in $\mathbb{G}(\overline x, \overline y)$ that is distinct from $z$, $z$ and $y$ co-occur in an atom of $\exists \overline x \mathbb{G}(\overline x, \overline y)$. 
\end{enumerate}
\end{enumerate}   
\end{defi}

In 4.~of Definitions \ref{def:gf}--\ref{def:cgf}, the atom $G$, the conjunctions of atoms $\mathbb{G}$ and $\exists \overline x (\mathbb{G}(\overline x, \overline y))$ are, respectively, the \emph{guard}, the \emph{loose guard} and the \emph{clique-guard} for~$F$. We say a formula is a \emph{guarded quantification formula} if it belongs to either \textsf{GF}, or \textsf{LGF} and \textsf{CGF}. Definition \ref{def:gf} defines \textsf{GF} in the same way as~\cite[Definition 2.1]{dNdR03} and~\cite[Definition 2.1]{GdN99} modulo equality. Definition~\ref{def:lgf} improves the previous definitions of \textsf{LGF} in~\cite{dNdR03,GdN99}: \cite[Definition 4.1]{dNdR03} misses Condition 4(a) of Definition~\ref{def:lgf}, and Condition~(ii) in the definition of \textsf{LGF} in \cite{GdN99} is amended in Condition~4(b) of Definition~\ref{def:lgf}. Unlike the definitions of \textsf{CGF} in \cite{M07,HT01}, Definition~\ref{def:cgf} is defined in accordance with Definitions~\ref{def:gf}--\ref{def:lgf} and disallows equality symbols.


Among the following formulas, $F_1, F_2, F_4, F_6$ and $F_7$ are guarded formulas, but not the rest. The formula $F_7$ is the standard translation \cite[chapter 2]{BRV01} of the modal formula $P \rightarrow \Diamond \Box P$ and the description logic axiom $P \sqsubseteq \exists R. \forall R. P$. For the relationship between \textsf{GF} and modal logic see \cite[section 7.4]{BRV01}, and for that between \textsf{GF} and description logic see \cite{SCM07}.
\begin{align*}
&F_1 = A(x) \qquad \qquad \qquad F_2 = \forall x \lbrack A(x,y) \rightarrow B(x,y) \rbrack \qquad \qquad \qquad F_3 = \forall x \lbrack A(x) \rbrack \\
&F_4 = \forall x \lbrack A(x,y) \rightarrow \bot \rbrack \qquad \qquad \qquad \qquad \qquad \ \quad F_5 = \forall x \lbrack A(x,y) \rightarrow \exists y(B(y,z)) \rbrack \\ 
&F_6 = \exists x \lbrack A(x,y) \land \forall z (B(x,z) \rightarrow \exists u (R(z,u))) \rbrack\\
&F_7 = \forall x \lbrack P(x) \rightarrow \exists y (R(x,y) \land \forall z (R(y,z) \rightarrow P(z)))) \rbrack
\end{align*}

Extended from \textsf{GF}, \textsf{LGF} allows a restricted form of a conjunction of atoms in the guard positions. For example, $\forall z \lbrack (R(x,z) \land R(z,y)) \rightarrow P(z) \rbrack$ and the first-order translation of a temporal logic formula $A \ \mathtt{until} \ B$:
\begin{align*}
\exists y \lbrack R(x,y) \land B(y) \land \forall z((R(x,z) \land R(z,y))\rightarrow A(z))) \rbrack,
\end{align*}
are loosely guarded formulas, but are not guarded.  Extended from \textsf{LGF}, \textsf{CGF} allows existentially quantified variables in loose guards. In the clique-guarded formula
\begin{align*}
F = \forall x_1x_2
\left\lbrack
\begin{array}{rll}
G(x_1,x_2)  \to & \forall x_3 ( & \\
& (\exists x_4x_5 (A(x_1,x_3,x_4) \land B(x_2,x_3,x_5)) \ ) \to & \\
& (\exists x_6 D(x_1,x_6) \land \top) & )
\end{array}
\right\rbrack,
\end{align*} 
$\exists x_6 D(x_1,x_6)$, $\exists x_4x_5 (A(x_1,x_3,x_4) \land B(x_2,x_3,x_5)) \ \text{and} \ G(x_1,x_2)$ are respectively the clique-guards of $\exists x_6 D(x_1,x_6) \land \top$, 
\begin{align*}
	&\forall x_3 ( \exists x_4x_5 (A(x_1,x_3,x_4) \land B(x_2,x_3,x_5)) \to (\exists x_6 D(x_1,x_6) \land \top)) \ \text{and} \ F.
\end{align*}
The transitivity formula $\forall xyz \lbrack (R(x,y) \land R(y,z)) \rightarrow R(x,z) \rbrack$ is neither a guarded nor a loosely guarded nor a clique-guarded formula.

\subsection*{\textbf{\textsf{BCQ} answering and saturation-based \textsf{BCQ} rewriting problems}}

First, we give the formal definition of \textsf{BCQ}s and unions thereof.
\begin{defi}
A \emph{Boolean conjunctive query} (\emph{\textsf{BCQ}}) is a first-order sentence of the form $\exists \overline x \varphi(\overline x)$, where $\varphi(\overline x)$ is a conjunction of atoms containing only constants and variables as arguments. A \emph{union of \textsf{BCQ}s} is a disjunction of \textsf{BCQ}s. 	
\end{defi}

This paper aims to answer the following question.

\begin{ques}
\label{ques:qans}
Given a set $\Sigma$ of formulas in \textsf{GF}, \textsf{LGF} and \textsf{CGF}, a set $D$ of ground atoms and a union $q$ of \textsf{BCQ}s, can we devise a practical decision procedure to check whether $\Sigma \cup D \models q$?
\end{ques}
In this paper, the above question is reduced to check whether $\Sigma \models q$, since ground atoms $D$ belong to the guarded quantification fragments $\Sigma$. To answer this question, we use a \emph{saturation-based method}, which computes the closure of a given set of formulas under a set of inference rules. 

If we answer Question \ref{ques:qans} positively, then we consider a follow-up question:

\begin{ques}
\label{ques:qrew}
Suppose $\Sigma$ is a set of formulas in \textsf{GF}, \textsf{LGF} and \textsf{CGF}, $D$ is a set of ground atoms and $q$ is a union of \textsf{BCQ}s. Further, suppose $N$ is the saturation obtained by applying the procedure devised for Question 1 to $\{\lnot q\} \cup \Sigma$. Can $N$ be back-translated to a (Skolem-symbol-free) first-order formula~$\Sigma_q$ such that $\Sigma \cup D \models q$ if and only if $D \models \Sigma_q$?
\end{ques}


\section{From formulas to clausal sets}
\label{sec:trans}
%
In this section, we formally define a clausal class to which the considered problems can be reduced, and then define our clausal normal form translation.

\subsection*{\textbf{Loosely guarded clauses and query clauses}}
It is helpful to use the \emph{flatness}, \emph{simpleness}, \emph{compatibility} and \emph{covering} properties to formally define our clausal forms, namely \emph{loosely guarded clauses} and \emph{query clauses}. 

A compound term is \emph{flat} if each argument in it is either a constant or a variable. A literal is \emph{flat} if each argument in it is either a constant or a variable. A clause is \emph{flat} if the literals in it are flat. A clause is \emph{simple} if each argument in it is either a variable or a constant or a flat compound term. A \emph{simple compound-term literal} (\emph{clause}), or plainly a \emph{compound-term literal} (\emph{clause}), is a simple literal (clause) containing at least one flat compound term. For example, $\lnot A(f(x,y))$ is a compound-term literal, but not $\lnot A(f(g(x),y))$ because ofs the presence of the nested compound term $f(g(x),y)$. A clause $C$ is \emph{covering} if each compound term $t$ in it satisfies $\Var(t)=\Var(C)$. Two compound terms $t$ and $s$ are \emph{compatible} if the argument sequences of $t$ and $s$ are identical. A clause $C$ is \emph{compatible} if in $C$, compound terms that are under the same function symbol are compatible. A clause is \emph{strongly compatible} if all compound terms in it are compatible. For example, $A_1(f(x,y)) \lor \lnot A_2(g(x,y)) \lor A_3(y,x)$ is covering and strongly compatible, and $A_1(f(x,y)) \lor \lnot A_2(g(y,x))$ is covering and compatible, but not strongly compatible.

\begin{defi}
\label{def:query}
A \emph{query clause} is a flat negative clause.   
\end{defi}
\begin{defi}
\label{def:lgc}
A \emph{loosely guarded clause} $C$ is a simple, covering and strongly compatible clause, satisfying the following conditions: 
\begin{enumerate}[noitemsep]
\item $C$ is either ground, or
\item $C$ contains a set of negative flat literals $\lnot G_1, \ldots, \lnot G_n$ such that each pair of distinct variable in $C$ co-occurs in at least one literal of $\lnot G_1, \ldots, \lnot G_n$.
\end{enumerate}
\end{defi}
In 2.~of Definition \ref{def:lgc}, $\lnot G_1, \ldots, \lnot G_n$ is called a  \emph{loose guard} of $C$. When a clause contains only one variable, then it is a loosely guarded clause if it is simple, covering, strongly compatible, and it contains a flat negative literal that contains its variable. A loosely guarded clause is a \emph{guarded clause} if its loose guards contain only one literal, which we call a \emph{guard} of this clause. A \emph{Horn guarded clause} is a guarded clause containing at most one positive literal. A clause is (\emph{loosely}) \emph{guarded} if it contains at least one (loose) guard.

Consider the clauses
\begin{align*}
&C_1 = \lnot A_1(x,y) \lor \lnot A_2(y,z) \lor \lnot A_3(z,x),\\
&C_2 = \lnot B_1(x,y,a) \lor \lnot B_2(y,z,b) \lor \lnot B_3(z,x,w), \\
&C_3 = \lnot A_1(x,y) \lor A_2(f(y,x),f(x,y)).
\end{align*}
The clause $C_1$ is a loosely guarded clause; $C_2$ is not as $w$ and $y$ do not co-occur in any negative flat literal; $C_3$ is not a loosely guarded clause either since $f(y,x)$ and $f(x,y)$ are not compatible. A query clause is not necessarily loosely guarded or vice-versa. For example, $C_1$ is a query clause; $\lnot A(x,y) \lor B(f(x,y))$ is a (loosely) guarded clause but not a query clause; and $\lnot A_1(x,y) \lor \lnot A_2(y,z)$ is a query clause, but not (loosely) guarded. 

\begin{figure}[t]
\normalsize 
\center
\scalebox{1}{
\begin{forest}
 [\textsf{LGQ} clauses, name=lgq
 [\textsf{LG} clauses, name=lg
 [\textsf{CGF}
 [\textsf{LGF},name=lgf
 [\textsf{GF},name=gf]]]
 [\thead{\normalsize guarded \\ \normalsize clauses}, name=gc
 [\thead{\normalsize Horn \\ \normalsize guarded clauses}, name=hgc]]
 ]
 [query clauses, name=qc
 [\thead{\normalsize negated \\ \normalsize a union of \textsf{BCQ}s}]
 ]]
 \draw[] (gf) to (gc);
 \draw[] (lgf) to[out=west,in=west] (lg);
\end{forest}
}
 \caption{Relationships between the investigated clausal classes and fragments}
 \label{fig:gfs_gcs}
\end{figure}
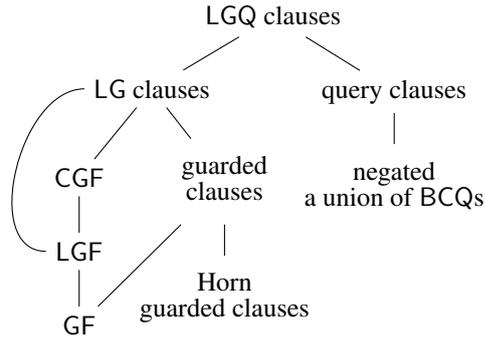

We use \textsf{LG} to denote the class of loosely guarded clauses, and \textsf{LGQ} to denote the class of both query and \textsf{LG} clauses. The class of \textsf{LG} clauses is more expressive than the guarded quantification fragments. For example, $\lnot G(x) \lor A(f(x))$ is an \textsf{LG} clause but it does not belong to the guarded quantification fragments. \textbf{Figure \ref{fig:gfs_gcs}} summarises the relationships between \textsf{BCQ}s, the guarded quantification fragments and the considered clausal classes. In \textbf{Figure \ref{fig:gfs_gcs}}, an upper node is more expressive than the one linked below it.

\subsection*{\textbf{Clausal normal form translation}}
We use the \emph{formula renaming} technique \cite[section 4]{NW01} in our clausification processes. Let $F[F_1(\overline x)]$ denote a first-order formula $F$ in which $F_1$ is a subformula of $F$ and~$\overline x$ are the free variables in $F_1$. Using a predicate symbol $P$, say, not occurring in~$F[F_1(\overline x)]$, \emph{formula renaming with positive literals} transforms $F[F_1(\overline x)]$ to 
\begin{align*}
F[P(\overline x)] \land \forall \overline x (\lnot P(\overline x) \lor F_1(\overline x))
\end{align*}
and \emph{formula renaming with negative literals} transforms $F[F_1(\overline x)]$ to 
\begin{align*}
F[\lnot P(\overline x)] \land \forall \overline x (P(\overline x) \lor F_1(\overline x)),
\end{align*}
where every occurrence of $F_1(\overline x)$ in $F[F_1(\overline x)]$ are replaced by $P(\overline x)$ and $\lnot P(\overline x)$, respectively. In the above \emph{formula renaming with positive literals}, $F[P(\overline x)]$ and $\forall \overline x (\lnot P(\overline x) \lor F_1(\overline x))$ are called the \emph{replacement} of $F[F_1(\overline x)]$ and the \emph{definition} of~$P$, respectively. In the above \emph{formula renaming with negative literals}, $F[\lnot P(\overline x)]$ and $\forall \overline x (P(\overline x) \lor F_1(\overline x))$ are called the \emph{replacement} of~$F[F_1(\overline x)]$ and the \emph{definition} of $P$, respectively. If a formula~$F$ is the definition of a predicate symbol $P$, we say $P$ \emph{defines} $F$. For a comprehensive description of clausification techniques, we refer the reader to~\cite{BEL01,NW01}.

Given a union $q_1 \lor \ldots \lor q_n$ of \textsf{BCQ}s and a set $\Sigma$ guarded quantification formulas, we reduce the entailment checking problem of $\Sigma \models q_1 \lor \ldots \lor q_n$ to the problem of checking unsatisfiability of $\{\lnot q_1 \land \ldots \land \lnot q_n\} \cup \Sigma$. We assume that all free variables in $\Sigma$ are existentially quantified as we are interested in satisfiability checking. We use \textbf{Trans} to denote our clausification process, detailed below.
\begin{enumerate}
\item Negate the union of \textsf{BCQ}s to obtain a set of \emph{query clauses}.
\item Clausify \emph{loosely guarded formulas} following the steps below, illustrated on 
\begin{align*}
F = \exists y \lbrack R(x,y) \land B(y) \land \forall z((R(x,z) \land R(z,y))\rightarrow A(z))) \rbrack.
\end{align*}
\begin{enumerate}
\item Add existential quantifiers to all free variables, equivalently express (double) implications as disjunctions and then perform negation normal form translation. From $F$ we obtain
\begin{align*}
F_1 = \exists xy \lbrack R(x,y) \land B(y) \land \forall z(\lnot R(x,z) \lor \lnot R(z,y) \lor A(z)) \rbrack.
\end{align*}
\item Use \emph{formula renaming with positive literals}~for all universally quantified subformulas in the formula obtained in 2(a). From $F_1$ we obtain  
\begin{align*}
F_2 = 
\left\lbrack
\begin{array}{rll}
\exists xy (& R(x,y) \land B(y) \land P_1(x,y) & ) \land\\
\forall xy ( & \lnot P_1(x,y) \lor \forall z (\lnot R(x,z) \lor \lnot R(z,y) \lor A(z)) & ) 
\end{array}
\right\rbrack,
\end{align*}
where $P_1$ is a fresh predicate symbol. We say that 
\begin{align*}
&\exists xy (R(x,y) \land B(y) \land P_1(x,y)) \ \text{is the \emph{replacement} of $F_1$, and} \\
& \forall xy ( \lnot P_1(x,y) \lor \forall z (\lnot R(x,z) \lor \lnot R(z,y) \lor A(z)) \ \text{is the \emph{definition} of $P_1$.}	
\end{align*}
\item Transform immediate subformulas of the formulas obtained in 2(b) that are connected by conjunctions to prenex normal form and then apply Skolemisation. By introducing Skolem constants $a$ and $b$, from $F_2$ we obtain 
\begin{align*}
F_3 = 
\left\lbrack
\begin{array}{rll}
& R(a,b) \land B(b) \land P_1(a,b) & \land \\
\forall xyz ( & \lnot P_1(x,y) \lor \lnot R(x,z) \lor \lnot R(z,y) \lor A(z) & ) 
\end{array}
\right\rbrack.
\end{align*}
\item Drop universal quantifiers and then perform conjunctive normal form transformation to formulas obtained in 2(c). From $F_3$ we obtain a set of \emph{\textsf{LG} clauses}:
\begin{align*}
R(a,b), \ B(b), \ P_1(a,b) \ \textnormal{and} \  \lnot P_1(x,y) \lor \lnot R(x,z) \lor \lnot R(z,y) \lor A(z).
\end{align*}  
\end{enumerate}
\item Clausify \emph{clique-guarded formula} following the steps below, illustrated on
\begin{align*}
F^\prime = \forall x_1x_2
\left\lbrack
\begin{array}{rll}
G(x_1,x_2)  \to & \forall x_3 ( & \\
& (\exists x_4x_5 (A(x_1,x_3,x_4) \land B(x_2,x_3,x_5)) \ ) \to & \\
& (\exists x_6 D(x_1,x_6) \land \top) & )
\end{array}
\right\rbrack.
\end{align*} 
\begin{itemize}
\item[(a)] Add existential quantification for all free variables and simplify $\top$ and~$\bot$. Unlike 2(a)~we first apply the \emph{miniscoping rule} \cite{NW01} to existential quantified variables in clique-guards, and then perform the negation normal form transformation. From $F^\prime$ we obtain 
\begin{align*}
F_1^\prime = \forall x_1x_2
\left\lbrack
\begin{array}{rll}
G(x_1,x_2) \to  & \forall x_3 ( & \\
& (\exists x_4 A(x_1,x_3,x_4) \land \exists x_5 B(x_2,x_3,x_5) \ ) \to & \\
& (\exists x_6 D(x_1,x_6) \land \top) & )
\end{array}
\right\rbrack.
\end{align*}
Then transform $F_1^\prime$ to negation normal form and drop $\top$, obtaining
\begin{align*}
F_2^\prime = \forall x_1x_2
\left\lbrack
\begin{array}{rll}
\lnot G(x_1,x_2) \ \lor & \forall x_3 ( & \\
& (\forall x_4 (\lnot A(x_1,x_3,x_4)) \lor \forall x_5  (\lnot B(x_2,x_3,x_5)) \ ) & \lor \\
& \exists x_6 D(x_1,x_6) & )
\end{array}
\right\rbrack,
\end{align*}
\item[(b1)] Apply \emph{formula renaming} to all universally quantified subformulas in the formula obtained in 3(a). For universally quantified subformulas in the \emph{clique-guards}, namely $\forall x_4 (\lnot A(x_1,x_3,x_4))$ and $\forall x_5  (\lnot B(x_2,x_3,x_5))$, we apply \emph{formula renaming with negative literals} to them. From $F_2^\prime$ we obtain an intermediate formula
\begin{align*}
F_3^\prime = 
\left\lbrack
\begin{array}{l}
\forall x_1x_3 ( P_1(x_1, x_3) \lor \forall x_4 (\lnot A(x_1,x_3,x_4))) \land \\
\forall x_2x_3 ( P_2(x_2, x_3) \lor \forall x_5 (\lnot B(x_2,x_3,x_5))) \land \\
\forall x_1x_2 ( \lnot G(x_1,x_2) \lor \forall x_3 (\lnot P_1(x_1, x_3) \lor \lnot P_2(x_2, x_3) \lor \exists x_6 D(x_1,x_6)))
\end{array}
\right\rbrack,
\end{align*}
where $P_1$ and $P_2$ are the fresh predicate symbols. 

\item[(b2)] For the remaining universally quantified subformulas in the formula obtained in 3(a) and 3(b1), we apply \emph{formula renaming with positive literals}. From $F_3^\prime$ we eventually obtain
\begin{align*}
F_4^\prime = 
\left\lbrack
\begin{array}{l}
p_1 \land \\
(\lnot p_1 \lor \forall x_1x_2 ( \lnot G(x_1,x_2) \lor P_3(x_1, x_2) )) \land  \\
\forall x_1x_3 ( P_1(x_1, x_3) \lor \forall x_4 (\lnot A(x_1,x_3,x_4))) \land \\
\forall x_2x_3 ( P_2(x_2, x_3) \lor \forall x_5 (\lnot B(x_2,x_3,x_5))) \land \\
\forall x_1x_2 ( \lnot P_3(x_1, x_2) \lor \forall x_3 (\lnot P_1(x_1, x_3) \lor \lnot P_2(x_2, x_3) \lor \exists x_6 D(x_1,x_6)))
\end{array}
\right\rbrack,
\end{align*} 
where $p_1$ and $P_3$ are the fresh predicate symbols. In $F_4^\prime$, $p_1$ is the \emph{replacement} of~$F_2^\prime$ and the remaining four conjuncts respectively \emph{defines} $p_1$, $P_1$, $P_2$ and~$P_3$.
\item[(c)] Transform immediate subformulas of the formulas obtained in 3(b2) that are connected by conjunctions to prenex normal form and then apply Skolemisation. Using a Skolem function symbol $f$, $F_4^\prime$ is transformed into
\begin{align*}
F_5^\prime = 
\left\lbrack
\begin{array}{l}
p_1 \land \\
(\lnot p_1 \lor \forall x_1x_2 ( \lnot G(x_1,x_2) \lor P_3(x_1, x_2) ) ) \land  \\
\forall x_1x_3x_4 ( P_1(x_1, x_3) \lor \lnot A(x_1,x_3,x_4) ) \land \\
\forall x_2x_3x_5 ( P_2(x_2, x_3) \lor \lnot B(x_2,x_3,x_5) ) \land \\
\forall x_1x_2x_3 ( \lnot P_3(x_1, x_2) \lor \lnot P_1(x_1, x_3) \lor \lnot P_2(x_2, x_3) \lor D(x_1,f(x_1,x_2,x_3)) ) \\
\end{array}
\right\rbrack.
\end{align*}
\item[(d)] Transform the formula obtained in 3(c) to conjunctive normal form and then drop universal quantifiers. From $F_5^\prime$ we obtain a set of \emph{\textsf{LG} clauses}:
\begin{align*}
&p_1, \qquad \qquad \qquad \qquad \qquad \ \lnot p_1 \lor \lnot G(x_1,x_2) \lor P_3(x_1, x_2), \\
&P_1(x_1, x_3) \lor \lnot A(x_1,x_3,x_4), \qquad \ \ \ P_2(x_2, x_3) \lor \lnot B(x_2,x_3,x_5), \\
&\lnot P_3(x_1, x_2) \lor \lnot P_1(x_1, x_3) \lor \lnot P_2(x_2, x_3) \lor D(x_1,f(x_1,x_2,x_3)).
\end{align*}
\end{itemize}
\end{enumerate}
To sum up, the \textbf{Trans} process transforms unions of \textsf{BCQ}s to \emph{query clauses}, clausifies guarded formulas to a set of \emph{guarded clauses}, and loosely guarded and clique-guarded formulas to a set of \emph{\textsf{LG} clauses}.  

By i) renaming universally quantified subformulas and ii) applying prenex normal form transformation and then Skolemisation to each conjunctively connected immediate subformulas, the \textbf{Trans} process intentionally introduces Skolem functions of a higher arity. More specifically, i)--ii) ensure that \textsf{LG} clauses have the \emph{covering} and the \emph{strong compatibility properties}. The covering property is essential to guarantee termination in our \textsf{BCQ} answering procedures, and the strong compatibility property makes the back-translation from an \textsf{LG} clausal set to a  first-order formula possible. 

The \textbf{Trans} process provides the most general and crucial clausification steps, but this can be further optimised in implementation. For example, in 3(c) of the \textbf{Trans} process, renaming the top-most formula $\forall x_1x_2 (\lnot G(x_1,x_2) \lor P_1(x_1, x_2))$ is not critical. Another possible optimisation is using \emph{formula renaming} to avoid the exponential blow-up of distributing disjunctions over conjunctions.

\begin{lem}
\label{lem:trans}
i) Applying the \textbf{Trans} process to a (loosely) guarded formula transforms it into a set of (loosely) guarded clauses, and ii) applying the \textbf{Trans} process to a clique-guarded formula transforms it into a set of loosely guarded clauses.
\end{lem}
\begin{proof}
i): Suppose $F$ is a loosely guarded formula. Suppose $F_2$ is a result of applying 2(a)--2(b) of \textbf{Trans} to $F$, and further suppose $P_1, \ldots, P_n$ are the fresh predicate symbols introduced in 2(b). W.l.o.g.~we say $F_2 = F_{2,1} \land \ldots \land F_{2,n} \land F_{2,r}$ where $F_{2,1}, \ldots, F_{2,n}$ are respectively the \emph{definitions} of $P_1, \ldots, P_n$ and $F_{2,r}$ is the \emph{replacement} of~$F_2$. We prove that \textbf{Trans} clausifies every conjunct of $F_2$ to a set of \textsf{LG} clauses.

 Consider $F_{2,r}$. By 2(b), no universally quantified subformulas occur in $F_{2,r}$, therefore~$F_{2,r}$ is a closed existentially quantified formula. The fact that $F_{2,r}$ contains no compound terms implies that 2(c)--2(d) clausify $F_{2,r}$ to a set of \emph{flat ground clauses}, which are \textsf{LG} clauses. 
 
Consider $F_{2,1}, \ldots, F_{2,n}$. W.l.o.g.~we take~$F_{2,1}$. By 2(b), $F_{2,1}$ can be represented as 
\begin{align*}
\forall \overline x (\lnot P_1(\overline x) \lor \forall \overline y (\lnot G_1(\overline {x_1}, \overline {y_1}) \lor \ldots \lor \lnot G_r(\overline {x_k}, \overline {y_k}) \lor F_{a}))
\end{align*}
where $\forall \overline y (\lnot G_1(\overline {x_1}, \overline {y_1}) \lor \ldots \lor \lnot G_r(\overline {x_k}, \overline {y_k}) \lor F_{a})$ is a loosely guarded formula, $F_{a}$ is a loosely guarded formula where all universal quantified formulas are abstracted (hence $F_{a}$ is a formula containing no universal quantification but may contain existential quantifications), $\overline {x_1}, \ldots, \overline {x_k} \subseteq \overline x$ and $\overline {y_1}, \ldots, \overline {y_k} \subseteq \overline y$. By~2(c), $F_{2,1}$ is converted to
\begin{align*}
\forall \overline x \overline y (\lnot P_1(\overline x) \lor \lnot G_1(\overline {x_1}, \overline {y_1}) \lor \ldots \lor \lnot G_r(\overline {x_k}, \overline {y_k}) \lor F_{a}).
\end{align*} 
If $F_{a}$ contains conjunctions, 2(c)--2(d) clausify $F_{2,1}$ to a set of clauses, otherwise $F_{2,1}$ is clausified to one clause. Suppose $C$ is a clause obtained by applying 2(c)--2(d) to~$F_{2,1}$. We use $C_1$ to denote the subclause $\lnot P_1(\overline x) \lor \lnot G_1(\overline {x_1}, \overline {y_1}) \lor \ldots \lor \lnot G_r(\overline {x_k}, \overline {y_k})$. First, we prove that $C_1$ is a loose guard of $C$. By the fact $\Var(F_{2,1}) = \overline x \overline y$, $\Var(C) = \overline x \overline y$. By 4 of Definition \ref{def:lgf}, $C_1$ is flat and $\Var(C_1) = \overline x \overline y$. By 4(b) of Definition~\ref{def:lgf} and the fact that the free variables of $\forall \overline y (\lnot G_1(\overline {x_1}, \overline {y_1}) \lor \ldots \lor \lnot G_r(\overline {x_k}, \overline {y_k}) \lor F_{a})$ are $\overline x$, each pair of variables in~$\overline x \overline y$ co-occurs in at least one literal of $C_1$. Hence $C_1$ is a \emph{loose guard} of $C$. Next, we prove that $C$ satisfies the other properties of \textsf{LG} clauses. We distinguish two cases of whether $F_{a}$ contains existential quantifications. Suppose $F_{a}$ contains existential quantifications and suppose the existentially quantified variables in $F_{a}$ are Skolemised to Skolem functions $f_1, \ldots, f_k$. W.l.o.g.~suppose $f_1$ and $f_2$ are two Skolem symbols occurring in~$C$. By prenex normal form transformation, all compound terms in $C$~that are under neither~$f_1$ or $f_2$ have the same sequence of arguments~$\overline x \overline y$, therefore~$C$ is \emph{covering} and \emph{strongly compatible}. As no function symbol occurs in $F_{a}$, no term in $C$ is nested, and~$C$ is \emph{simple}. Then, $C$ is an \textsf{LG} clause. Suppose~$F_{a}$ contains no existentially quantified formulas. Immediately~$C$ is \emph{flat}. Since we previously proved that $C_1$ is a loose guard of $C$, $C$ is an \textsf{LG} clause. That \textbf{Trans} converts guarded formulas to a set of guarded clauses is immediate since this is the case that a loose guard contains only one literal.

ii): Now we consider the clique-guarded formula. Unlike the clausification for loosely guarded formulas, the existentially quantified variables in clique-guards, mentioned in Condition 4(b) in the \textsf{CGF} definition, need to be handled. Suppose $F^\prime$ is a clique-guarded formula, and w.l.o.g.~suppose $F_2^\prime$ is a result of applying 3(a) to $F^\prime$. Further, suppose $F_3^\prime$ is the result of applying 3(b1) to $F_2^\prime$. Using the fresh predicate symbols $P_{3,1}, \ldots, P_{3,n}$, we say $F_3^\prime = F_{3,1}^\prime \land \ldots \land F_{3,n}^\prime \land F_{3,r}^\prime$ where $F_{3,1}^\prime, \ldots, F_{3,n}^\prime$ are respectively the \emph{definitions} of $P_{3,1}, \ldots, P_{3,n}$ and $F_{3,r}^\prime$ is the \emph{replacement} of $F_3^\prime$. Assume that $F_4^\prime$ is the result of applying 3(b2) to $F_{3,r}^\prime$. Using fresh predicate symbols $P_{4,1}, \ldots, P_{4,m}$, we say $F_4^\prime = F_{3,1}^\prime \land \ldots \land F_{3,n}^\prime \land F_{4,1}^\prime \land \ldots \land F_{4,m}^\prime \land F_{4,r}^\prime$ where $F_{4,1}^\prime, \ldots, F_{4,m}^\prime$ are respectively the \emph{definitions} of $P_{4,1}, \ldots, P_{4,m}$ and $F_{4,r}^\prime$ is the \emph{replacement} of $F_4^\prime$. We prove that by \textbf{Trans} every conjunct of $F_4^\prime$ is clausified as a set of \textsf{LG} clauses.

%
%
%
%
%
%

Consider applying 3(b1) to $F_2^\prime$, deriving $F_3^\prime$, viz., $F_{3,1}^\prime \land \ldots \land F_{3,n}^\prime \land F_{3,r}^\prime$. Suppose~$F_{2,1}^\prime$ is a subformula in $F_2^\prime$ that contains universally quantified subformulas occurring in clique-guards. W.l.o.g.~we assume that $F_{3,1}^\prime \land \ldots \land F_{3,n}^\prime \land F_{3,r}^\prime$ is obtained by applying 3(b1) to $F_{2,1}^\prime$ and w.l.o.g.~we present $F_{2,1}^\prime$ as
\begin{align*}
\forall \overline z (\forall \overline {x_1} \lnot G_1 (\overline {x_1}, \overline {y_1}) \lor \ldots \lor \forall \overline {x_k} \lnot G_t(\overline {x_k}, \overline {y_k}) \lor F_{a}^\prime)	
\end{align*}
where $\overline {x_1}, \ldots, \overline {x_k}$ respectively only occur in $\lnot G_1 (\overline {x_1}, \overline {y_1}), \ldots, \lnot G_t(\overline {x_k}, \overline {y_k})$ and $F_{a}^\prime$ is a clique-guarded formula. W.l.o.g.~we use $P_{3,1}, \ldots, P_{3_t}$ such that $t \leq n$ to apply 3(b1) to~$F_{2,1}^\prime$, obtaining
\begin{align*}
& \forall \overline z (\lnot P_{3,1}(\overline {y_1}) \lor \ldots \lor \lnot P_{3,r}(\overline {y_k}) \lor F^{\prime\prime}_{a}) \land	\\
& \forall \overline {y_1} (P_{3,1}(\overline {y_1}) \lor \forall \overline {x_1} \lnot G_1(\overline {x_1}, \overline {y_1})) \land \ldots \land \forall \overline {y_k} (P_{3_t}(\overline {y_k}) \lor \forall \overline {x_n} \lnot G_t(\overline {x_k}, \overline {y_k}))
\end{align*}
where $F^{\prime\prime}_{a}$ is a clique-guarded formula and no universal quantification occurs in its clique-guards (since 3(b1) abstracts universal quantified formulas in clique-guards). The subformula $\forall \overline z (\lnot P_{3,1}(\overline {y_1}) \lor \ldots \lor \lnot P_{3,r}(\overline {y_k}) \lor F^{\prime\prime}_{a})$ is the \emph{replacement} of $F_{2,1}^\prime$. This replacement represents a conjunct in $F_{3,r}^\prime$ and we consider $F_{3,r}^\prime$ in the next paragraph. The subformulas
\begin{align*}
& \forall \overline {y_1} (P_{3,1}(\overline {y_1}) \lor \forall \overline {x_1} \lnot G_1(\overline {x_1}, \overline {y_1})), \ \ldots, \  \forall \overline {y_k} (P_{3_t}(\overline {y_k}) \lor \forall \overline {x_n} \lnot G_t(\overline {x_k}, \overline {y_k})).	
\end{align*}
are the \emph{definitions} of $P_{3,1}, \ldots, P_{3_t}$ such that $t \leq n$, respectively. By 3(c)--3(d) these definitions are clausified to \emph{flat \textsf{LG} clauses consisting of two literals}. Hence, 3(c)--3(d) clausify $F_{3,1}^\prime \land \ldots \land F_{3,n}^\prime$ to a set of \textsf{LG} clauses. 

Next consider $F_{3,r}^\prime$. Since $F_{3,r}^\prime$ contains no quantification in its clique-guard, by the definitions of \textsf{LGF} and \textsf{CGF}, $F_{3,r}^\prime$ is a loosely guarded formula. Suppose applying~3(b2) to $F_{3,r}^\prime$ derives $F_4^\prime = F_{4,1}^\prime \land \ldots \land F_{4,m}^\prime \land F_{4,r}^\prime$. W.l.o.g. we discuss $F_{4,1}^\prime$. The fact that no universal quantification occurs in clique-guards of~$F_{3,r}^\prime$ implies that $F_{4,1}^\prime$ can be presented as
\begin{align*}
\forall \overline x (\lnot P_{4,1}(\overline x) \lor \forall \overline y (\lnot G_1(\overline {x_1}, \overline {y_1}) \lor \ldots \lor \lnot G_l(\overline {x_k}, \overline {y_k}) \lor F_{a}^{\prime\prime\prime})
\end{align*}
where $\forall \overline y (\lnot G_1(\overline {x_1}, \overline {y_1}) \lor \ldots \lor \lnot G_l(\overline {x_k}, \overline {y_k}) \lor F_{a}^{\prime\prime\prime})$ is a loosely guarded formula, $F_{a}^{\prime\prime\prime}$ is a loosely guarded formula where all universal quantified formulas are abstracted (hence it is a formula containing no universal quantification but may contain existential quantifications) and $\overline {x_1}, \ldots, \overline {x_k} \subseteq \overline x$ and $\overline {y_1}, \ldots, \overline {y_k} \subseteq \overline y$. Note that~$F_{a}^{\prime\prime}$ is obtained by abstracting \emph{universally quantified subformulas in clique-guards} in $F_{2,1}^\prime$, and $F_{a}^{\prime\prime\prime}$ is obtained by abstracting \emph{all universally quantified formulas} in $F_{3,r}^\prime$. By the result established in applying 2(c)--2(d) of \textbf{Trans} to $F_{2,1}, \ldots, F_{2,n}$, 3(c)--3(d) of \textbf{Trans} clausify $F_{4,1}^\prime$ to an \textsf{LG} clause or a set of \textsf{LG} clauses if $F_{a}^{\prime\prime\prime}$ contains conjunctions. Finally consider~$F_{4,r}^\prime$. By the result established in applying 2(b) of \textbf{Trans} to $F_{2,r}$, 3(c)--3(d) clausify $F_{4,r}^\prime$ to a set of \emph{flat ground clauses}, viz., \textsf{LG} clauses.
\end{proof}


\begin{thm}
\label{thm:trans}
The \textbf{Trans} process reduces the problem of \textsf{BCQ} answering for \textsf{GF}, \textsf{LGF} and \textsf{CGF} to that of deciding satisfiability of a set of \textsf{LGQ} clauses.   
\end{thm}
\begin{proof}
Suppose $q_1 \lor \ldots \lor q_n$ is a union of \textsf{BCQ}s, $\Sigma$ is a set of guarded quantification formulas and $D$ is a set of ground atoms. Since ground atoms $D$ belong to \textsf{GF}, \textsf{LGF} and \textsf{CGF}, it suffices to reduce checking entailment of $\Sigma \models q_1 \lor \ldots \lor q_n$ to checking unsatisfiability of $\{\lnot q_1, \ldots, \lnot q_n\} \cup \Sigma$. By the definition of \textsf{BCQ}, $\{\lnot q_1, \ldots, \lnot q_n\}$ is a set of query clauses. By Lemma~\ref{lem:trans}, $\Sigma$ is clausified to a set of \textsf{LG} clauses.
\end{proof}

\section{Top-variable inference system}
\label{sec:tinf}
In this section, we present three systems: a basic \emph{selection-based resolution} system, a \emph{partial selection-based resolution} system and a \emph{top-variable resolution} system.




\subsection*{\textbf{Basic notions in the saturation-based resolution framework}}

In our systems, admissible orderings and selection functions are the two main parameters to refine and guide the inference process. The following notions are standard in the resolution framework of \cite{BG01}. 

 Let $\succ$ be a strict ordering, called a \emph{precedence}, on the symbols in $\mathtt{C}$, $\mathtt{F}$ and $\mathtt{P}$. An ordering~$\succ$ on expressions is \emph{liftable} if $E_1 \succ E_2$ implies $E_1\sigma \succ E_2\sigma$ for any expressions $E_1$, $E_2$ and substitution $\sigma$. An ordering $\succ$ on literals is \emph{admissible}, if the following conditions are satisfied.
 \begin{itemize}
 \item It is \emph{liftable}, \emph{well-founded} and \emph{total} on ground literals,
 \item $\lnot{A} \succ A$ for all ground atoms $A$,
 \item if $B \succ A$, then $B \succ \lnot A$ for all ground atoms $A$ and $B$.
 \end{itemize}
 
Let $\succ$ be an ordering and $C$ a ground clause. A literal $L$ in $C$ is \emph{(strictly) maximal with respect to the ground clause} $C$ if and only if for all~$L^\prime$ in $C$, $L \succeq L^\prime$~($L \succ L^\prime$). A non-ground literal $L$ is \emph{(strictly) maximal with respect to a clause}~$C$ if and only if there exist some ground substitutions $\sigma$ such that $L\sigma$ is~(strictly) maximal with respect to~$C\sigma$, that is, for all $L^\prime$ in $C$, $L\sigma \succeq L^\prime\sigma$ ($L\sigma \succ L^\prime\sigma$). A \emph{selection function} maps a clause $C$ to a multiset of negative literals in $C$. The literals in the range of selection functions are said to be \emph{selected}. An \emph{eligible literal with respect to a clause} is either a (strictly) maximal literal or a selected literal.
 
A ground clause $C$ is \emph{redundant with respect to} a ground clausal set $N$ if there exist $C_1, \ldots, C_n$ in $N$ such that $C_1, \ldots, C_n \models C$ and $C \succ C_i$ for each $i$ with $1 \leq i \leq~n$. Let $N$ be a clausal set. Then a ground clause $C$ is \emph{redundant with respect to}~$N$ if there exists ground instances $C_1\sigma, \ldots, C_n\sigma$ of clauses $C_1, \ldots, C_n$ in $N$ such that $C_1\sigma, \ldots, C_n\sigma \models C$ and $C \succ C_i\sigma$ for each $i$ with $1 \leq i \leq n$. A non-ground clause $C$ is \emph{redundant with respect to} $N$ if every ground instance of $C$ is redundant with respect to~$N$. Let $C$ and $C_1, \ldots, C_n$ be premises and $D$ a conclusion in an inference~\textbf{I}. Then \emph{the inference \textbf{I} is redundant with respect to $N$} if there exist clauses $D_1, \ldots, D_k$ in $N$ that are smaller than~$C$ such that $C_1, \ldots, C_n,D_1, \ldots, D_k \models D$. A \emph{non-ground inference \textbf{I} is redundant with respect to $N$} if every ground instance of \textbf{I} is redundant in the ground instances of the clauses of $N$. A clausal set $N$ is \emph{saturated up to redundancy with respect to an inference system} \textbf{R} if all inferences in~\textbf{R} with non-redundant premises in $N$ are redundant with respect to $N$. 

\subsection*{\textbf{The \textbf{S-Res} system}}
In this section, we fine a \emph{selection-based resolution} system, referred to as the \textbf{S-Res} system. This is a standard instance of the resolution framework in \cite{BG01}.

The \textbf{S-Res} system consists of two types of rules: the \textbf{Deduce} and \textbf{Delete} rules. New conclusions are derived using the \textbf{Deduce} rule.

\begin{mdframed}[linewidth=2pt]
\begin{displaymath}
\prftree[l]{\textbf{Deduce}: \quad }
  {N}
  {N \cup \{C\}}
\end{displaymath}
if $C$ is a conclusion of applying resolution or positive factoring rules to $N$.
\end{mdframed}

To ensure decidability, we minimally need the following \textbf{Delete} rule.

\begin{mdframed}[linewidth=2pt]
\begin{displaymath}
\prftree[l]{\textbf{Delete}: \ }
  {N \cup \{C\}}
  {N}
\end{displaymath}
if $C$ is a tautology, or $N$ contains a variant of $C$.
\end{mdframed}

The \textbf{Factor} rule is the \emph{positive factoring rule}, defined by:
\begin{mdframed}
\begin{displaymath}
\prftree[l]{\textbf{Factor}: \quad}
  {C \lor A_1 \lor A_2}
  {(C \lor A_1)\sigma}
\end{displaymath}
if the following conditions are satisfied.
\begin{enumerate}
\item Nothing is selected in $C \lor A_1 \lor A_2$.
\item $A_1\sigma$ is $\succ$-maximal with respect to $C\sigma$.
\item $\sigma = \mgu(A_1 \doteq A_2)$	
\end{enumerate}
\end{mdframed}

The \textbf{S-Res} rule is the \emph{selection-based resolution rule} defined by
\begin{mdframed}
\begin{align*}
 \prftree[l]{\textbf{S-Res}: }
    {B_1 \lor D_1, \ \ldots, \ B_n \lor D_n}
    {}
    {\lnot A_1 \lor \ldots \lor \lnot A_n \lor D}
    {(D_1 \lor \ldots \lor D_n \lor D)\sigma}
\end{align*}
if the following conditions are satisfied.
\begin{enumerate}
\item[1.] No literal is selected in $D_1, \ldots, D_n, D$ and $B_1\sigma, \ldots, B_n\sigma$ are strictly $\succ$-maximal with respect to $D_1\sigma, \ldots, D_n\sigma$, respectively. 
\item[2a.] If $n = 1$, then i) either $\lnot A_1$ is selected, or nothing is selected in $\lnot A_1 \lor D$ and $\lnot A_1\sigma$ is $\succ$-maximal with respect to $D\sigma$, and ii) $\sigma = \mgu(A_1 \doteq B_1)$ or
\item[2b.] if $n > 1$, then $\lnot A_1, \ldots, \lnot A_n$ are selected and $\sigma = \mgu(A_1 \doteq B_1, \ldots, A_n \doteq B_n)$.
\item[3.] All premises are variable disjoint.
\end{enumerate} 
\end{mdframed}

In the \textbf{S-Res} rule, the right-most premise is the \emph{main premise} and the others are the \emph{side premises}. Unlike the standard \emph{hyperresolution rule} \cite{R65a} (like the \emph{hyperresolution rule} in \cite{W19}), the \textbf{S-Res} rule does not require the side premises to be positive and all negative literals in the main premise to be selected, e.g., $D$ in the main premise is not nessarily positive. Standard hyperresolution is only applied when the selection function selects all negative literals in the premises of the \textbf{S-Res} rule. The \emph{binary resolution rule}~\cite{BG01} is an instance of the \textbf{S-Res} rule whenever it only has one selected literal in the main premise. 

The \textbf{S-Res} system is defined in the spirit of the resolution framework of~\cite{BG01}, therefore, more sophisticated simplification and redundant elimination techniques, such as forward and backward subsumption elimination and condensation in~\cite[section 4.3]{BG01}, can be freely added to the system.

\begin{thm}
\label{thm:sres}
The \textbf{S-Res} system is sound and refutationally complete for general first-order clausal logic.
\end{thm}
\begin{proof}
By the fact that the \textbf{S-Res} system strictly follows the principles of the resolution framework in \cite{BG01}.
\end{proof}

\subsection*{\textbf{The \textbf{P-Res} system}}
Next, we describe a new \emph{partial selection-based resolution inference system}, denoted as \textbf{P-Res}. This system is built on the top of the \textbf{S-Res} system, but the \textbf{S-Res} rule is replaced by the following \emph{partial selection-based resolution rule}.
\begin{mdframed}
\begin{align*}
 \prftree[l]{\textbf{P-Res}: }
    {B_1 \lor D_1, \ \ldots, \ B_m \lor D_m, \ \ldots, \ B_n \lor D_n}
    {\lnot A_1 \lor \ldots \lor \lnot A_{m} \lor \ldots \lor \lnot A_n \lor D}
    {(D_1 \lor \ldots \lor D_m \lor \lnot A_{m+1} \lor \ldots \lor \lnot A_{n} \lor D)\sigma}
\end{align*}
if the following conditions are satisfied.
\begin{enumerate}
\item[1.] No literal is selected in $D_1, \ldots, D_n,D$ and $B_1\sigma, \ldots, B_n\sigma$ are strictly $\succ$-maximal with respect to $D_1\sigma, \ldots, D_n\sigma$, respectively. 
\item[2a.] If $n = 1$, then i) either $\lnot A_1$ is selected, or nothing is selected in $\lnot A_1 \lor D$ and $\lnot A_1\sigma$ is $\succ$-maximal with respect to $D\sigma$, and ii) $\sigma = \mgu(A_1 \doteq B_1)$ or
\item[2b.] there must exist an mgu $\sigma^\prime$ such that $\sigma^\prime = \mgu(A_1 \doteq B_1, \ldots, A_n \doteq B_n)$, then the mgu used to perform the inference is $\sigma = \mgu(A_1 \doteq B_1, \ldots, A_m \doteq B_m)$ where $1 \leq m \leq n$.
\item[3.] All premises are variable disjoint.
\end{enumerate} 
\end{mdframed} 

The \textbf{P-Res} rule is \emph{not} a selection-based resolution rule where a sub-multiset of the negative literals in the main premise is selected. The literals $\lnot A_1, \ldots, \lnot A_m$ are resolved \emph{not} because they are selected, but because the application of the \textbf{S-Res} rule makes the inference on a sub-multiset of the \textbf{S-Res} side premises and the \textbf{S-Res} main premise possible. Condition 2b.~stipulates the existence of an mgu between $A_1, \ldots, A_n$ and $B_1, \ldots, B_n$ as a pre-requisite for the application of the \textbf{P-Res} rule. This means that whenever the \textbf{S-Res} rule applies to
\begin{align*}
C_1 = B_1 \lor D_1, \ \ldots,\  C_n = B_n \lor D_n \ \text{and}  \ C = \lnot A_1 \lor \ldots \lor \lnot A_m \lor \ldots \lor \lnot A_n \lor D
\end{align*}
with $\lnot A_1, \ldots, \lnot A_n$ selected, one can apply the \textbf{P-Res} rule with $m$ of the side premises where $1 \leq m \leq n$. We say that $\lnot A_1, \ldots, \lnot A_m$ are the \emph{\textbf{P-Res} eligible literals} with respect to an \textbf{S-Res} inference.

Unlike the \textbf{S-Res} rule, Condition 2b.~in the \textbf{P-Res} rule includes the case of~$n=1$, meaning that the pre-requisites for Conditions~2a.~and 2b.~are not exclusive. Though when~$n = 1$, using either Condition~2a.~or 2b.~to the main premise derives the same conclusion, the mechanism is different: Condition~2a.~considers the situation when the \textbf{P-Res} rule is reduced to a \emph{binary \textbf{S-Res} rule}, but Condition~2b.~considers the partial inferences when the main premise contains only one \textbf{P-Res} eligible literal. Both mechanisms are useful in practice: for example, Condition~2a.~is used when a main premise contains only one negative literal, but when a main premise contains multiple negative literals, Condition~2b.~allows us to decide that among all these negative literals, the one we want to resolve, to derive a partial resolvent. This partial resolvent can have properties that the resolvent, when we resolve all the negative literals, does not have.


Although the \textbf{S-Res} rule has the advantage of \emph{avoiding intermediate resolvents} that are derived by binary resolution rules, the \textbf{S-Res} resolvents can be difficult to tame as the rule is performed on a macro level. The \textbf{P-Res} rule, on the other hand, amends the \textbf{S-Res} rule by allowing one to resolve \emph{any} non-empty and non-strict sub-multiset of the \textbf{S-Res} side premises with the \textbf{S-Res} main premise. This means that the \textbf{P-Res} rule provides new flexibility to capture the \textbf{S-Res} resolvents and thus generalises the \textbf{S-Res} rule. This flexibility is important to tame (and decide) the clausal class we consider.

Next, we show soundness and refutational completeness of the \textbf{P-Res} system. A \textbf{P-Res} inference with the main premise $C$ and a sub-multiset of the side premises $C_1, \ldots, C_n$ makes the \textbf{S-Res} inference on $C$ and $C_1, \ldots, C_n$ redundant. We first consider the ground case.
\begin{lem}
\label{lem:pres_gnd}
Suppose $N$ is a clausal set and $C_1, \ldots, C_n, C$ are ground clauses occurring in~$N$. Suppose~\textbf{I} is an \textbf{S-Res} inference with $C_1, \ldots, C_n$ the side premises and $C$ the main premise. Further suppose $R_p$ is the \textbf{P-Res} resolvent of applying the \textbf{P-Res} rule to a sub-multiset of $C_1, \ldots, C_n$ and~$C$. Then, \textbf{I} is redundant with respect to $N \cup \{R_p\}$.	
\end{lem}
\begin{proof}
Suppose $R$ is the resolvent of \textbf{I} and $\succ$ is the applied admissible ordering. By the notion of redundant inferences for ground clauses, we prove that $C \succ R_p$ and $C_1, \ldots, C_n, R_p \models R$. W.l.o.g.~suppose 
\begin{align*}
C_1 = A_1 \lor D_1, \ldots,  C_n = A_n \lor D_n \ \text{and} \ C = \lnot A_1 \lor \ldots \lor \lnot A_m \lor \ldots \lor \lnot A_n \lor D	
\end{align*}
where $1 \leq m \leq n$. Further suppose a \textbf{P-Res} inference is performed on $C$ and $C_1, \ldots, C_m$. By the definitions of the  \textbf{S-Res} and \textbf{P-Res} rules,
\begin{align*}
R = D_1 \lor \ldots \lor D_n \lor D	\ \text{and} \ 
R_p = \lnot A_{m+1} \lor \ldots \lor \lnot A_{n} \lor D_1 \lor \ldots \lor D_m \lor D.
\end{align*}
By Condition 1.~of the \textbf{S-Res} and \textbf{P-Res} rules, $A_1 \succ D_1, \ldots, A_m \succ D_m$, hence $C \succ R_p$. Next, we prove $C_1, \ldots, C_n, R_p \models R$ by contradiction. Let $I$ be an arbitrary interpretation satisfying that
\begin{align}
& I \models A_1 \lor D_1, \ldots, A_n \lor D_n, \lnot A_{m+1} \lor \ldots \lor \lnot A_{n} \lor D_1 \lor \ldots \lor D_m \lor D, \\
& \text{but} \ I \not \models D_1 \lor \ldots \lor D_n \lor D.
\end{align}
(2) implies $I \not \models D_1, \ldots, I \not \models D_n$, therefore, considering (1) we get that
\begin{align}
I \models A_1, \ldots, A_n, \lnot A_{m+1} \lor \ldots \lor \lnot A_{n} \lor D_1 \lor \ldots \lor D_m \lor D.
\end{align}
(3) implies that $I \models D_1 \lor \ldots \lor D_m \lor D$. As $D_1 \lor \ldots \lor D_m \lor D$ is a subclause of $D_1 \lor \ldots \lor D_n \lor D$, $I \models D_1 \lor \ldots \lor D_n \lor D$, which refutes (2). Then, $C_1, \ldots, C_n, R_p \models R$. By the facts that $C \succ R_p$ and $C_1, \ldots, C_n, R_p \models R$, \textbf{I} is redundant with respect to $N \cup\{R_p\}$.
\end{proof}

Lemma \ref{lem:pres_gnd} shows that given an \textbf{S-Res} inference \textbf{I} on ground clauses of a clausal set~$N$, computing a \textbf{P-Res} resolvent $R_p$ with respect to \textbf{I} makes \textbf{I} redundant with respect to $N \cup \{R_p\}$. Similar justifications can be found in \cite[\normalfont{pages 53--54}]{BG01} and \cite[\normalfont{page 28}]{BG97} described as `partial replacement strategy'. 

Next, we generalise Lemma \ref{lem:pres_gnd} to non-ground inferences.

\begin{lem}
\label{lem:pres_gen}
Suppose $N$ is a clausal set and $C_1, \ldots, C_n, C$ are general clauses occurring in~$N$. Suppose~\textbf{I} is an \textbf{S-Res} inference where $C_1, \ldots, C_n$ are the side premises and $C$ is the main premise. Further suppose $R_p$ is the \textbf{P-Res} resolvent of applying the \textbf{P-Res} rule to a sub-multiset of $C_1, \ldots, C_n$ and~$C$. Then, every ground instance of \textbf{I} is redundant with respect to the ground instances of the clauses in $N \cup \{R_p\}$.	
\end{lem}
\begin{proof}

Suppose $R$ is the \textbf{S-Res} resolvent in \textbf{I}. W.l.o.g.~suppose $C_1, \ldots, C_m$ are side premises of applying the \textbf{P-Res} rule to $C$ and $C_1, \ldots, C_m$ and $R_p$ is the resolvent, where $1 \leq m \leq n$. Suppose $\sigma$ is a ground substitution satisfying that applying the \textbf{S-Res} rule to $C_1\sigma, \ldots, C_n\sigma$ as the side premises and $C\sigma$ as the main premise derives~$R\sigma$. We use $I_{gnd}$ to denote this ground \textbf{S-Res} inference. Since the \textbf{P-Res} rule only requires a sub-multiset of the \textbf{S-Res} side premises, the \textbf{P-Res} rule is applicable to $C_1\sigma, \ldots, C_m\sigma$ as the side premises and $C\sigma$ as the main premise, deriving $R_p\sigma$. By Lemma \ref{lem:pres_gnd}, $I_{gnd}$ is redundant with respect to the ground instances $C_1\sigma, \ldots, C_n\sigma, R_p\sigma$ of the clauses in $N \cup \{R_p\}$. Hence, every ground \textbf{S-Res} inference is redundant with respect to the ground instances of the clauses in $N \cup \{R_p\}$.
\end{proof}

The main result of this section is then as follows.
\begin{thm}
\label{thm:pres}
The \textbf{P-Res} system is sound and refutationally complete for general first-order clausal logic.
\end{thm}
\begin{proof}
By Lemma \ref{lem:pres_gnd} and Theorem \ref{thm:sres}, the \textbf{P-Res} system is sound and complete for ground clauses. By the fact that the \textbf{Factor} rule is the positive factoring rule in the resolution framework of \cite{BG01} and Lemma \ref{lem:pres_gen}, the \textbf{P-Res} system is sound and refutational complete for general first-order clauses.
\end{proof}

\subsection*{\textbf{The \textbf{T-Res} system}}
Finally, we present the \emph{top-variable resolution inference system}, referred to as the \textbf{T-Res} system. As a special case of the \textbf{P-Res} system, the \textbf{T-Res} system uses the customised \emph{admissible orderings}, \emph{selection functions} and a specific version of the \textbf{P-Res} rule, i.e., the \emph{top-variable resolution rule \textbf{T-Res}}, particularly devised for deciding satisfiability of the \textsf{LGQ} clausal class.

First, we give the \emph{top-variable resolution rule \textbf{T-Res}}. Suppose in an \textbf{S-Res} inference with $C_1 = B_1 \lor D_1, \ldots, C_n = B_n \lor D_n$ the side premises and $C = \lnot A_1 \lor \ldots \lor \lnot A_n \lor~D$ the main premise with $\lnot A_1, \ldots, \lnot A_n$ selected. The \emph{top-variable technique} is applied to this inference by the following steps.
\begin{enumerate}
\item Without producing or adding the resolvent, compute an mgu $\sigma^\prime$ for $C_1, \ldots, C_n$ and~$C$ such that $\sigma^\prime = \mgu(A_1 \doteq B_1, \ldots, A_n \doteq B_n)$.
\item Compute the \emph{variable ordering} $>_v$ and $=_v$ over the variables of $\lnot A_1 \lor \ldots \lor \lnot A_n$. By definition, $x >_v y$ and $x =_v y$ with respect to $\sigma^\prime$, if $\Dep(x\sigma^\prime) > \Dep(y\sigma^\prime)$ and $\Dep(x\sigma^\prime) = \Dep(y\sigma^\prime)$, respectively.
\item Based on $>_v$ and $=_v$, the maximal variables in $\lnot A_1 \lor \ldots \lor \lnot A_n$ are the \emph{top variables}. The sub-multiset $\lnot A_1, \ldots, \lnot A_m$ of $\lnot A_1, \ldots, \lnot A_n$ ($1 \leq m \leq n$) are the \emph{top-variable literals} if each literal in $\lnot A_1, \ldots, \lnot A_m$ contains at least one top variable, and $\lnot A_1 \lor \ldots \lor \lnot A_m$ is the \emph{top-variable subclause} of $C$.
\end{enumerate}
The \emph{top-variable resolution rule} is defined by
\begin{mdframed}
\begin{align*}
 \prftree[l]{\textbf{T-Res}: }
    {B_1 \lor D_1, \ \ldots, \ B_m \lor D_m, \ \ldots, \ B_n \lor D_n}
    {\lnot A_1 \lor \ldots \lor \lnot A_{m} \lor \ldots \lor \lnot A_n \lor D}
    {(D_1 \lor \ldots \lor D_m \lor \lnot A_{m+1} \lor \ldots \lor \lnot A_{n} \lor D)\sigma}
\end{align*}
if the following conditions are satisfied.
\begin{enumerate}
\item[1.] No literal is selected in $D_1, \ldots, D_n,D$ and $B_1\sigma, \ldots, B_n\sigma$ are strictly $\succ$-maximal with respect to $D_1\sigma, \ldots, D_n\sigma$, respectively. 
\item[2a.] If $n = 1$, then i) either $\lnot A_1$ is selected, or nothing is selected in $\lnot A_1 \lor D$ and $\lnot A_1\sigma$ is $\succ$-maximal with respect to $D\sigma$, and ii) $\sigma = \mgu(A_1 \doteq B_1)$ or
\item[2b.] there must exist an mgu $\sigma^\prime$ such that $\sigma^\prime = \mgu(A_1 \doteq B_1, \ldots, A_n \doteq B_n)$, then $\lnot A_1, \ldots, \lnot A_m$ are the \emph{top-variable literals} of $\lnot A_1 \lor \ldots \lor \lnot A_{m} \lor \ldots \lor \lnot A_n \lor D$ and $\sigma = \mgu(A_1 \doteq B_1, \ldots, A_m \doteq B_m)$ where $1 \leq m \leq n$.
\item[3.] All premises are variable disjoint.
\end{enumerate} 
\end{mdframed}
\emph{Top variables}, \emph{top-variable literals} and \emph{top-variable subclauses} are only in effect with respect to an \textbf{S-Res} inference, since the \textbf{T-Res} rule is a very specific application of the \textbf{P-Res} rule, built on the top of the \textbf{S-Res} rule. Suppose \textbf{I} is an \textbf{S-Res} inference with $C_1, \ldots, C_n$ the side premises and $C$ the main premise. As shown in the previous section, the \textbf{P-Res} rule allows one to perform an inference on $C$ and \emph{any} sub-multiset of $C_1, \ldots, C_n$. Suppose \textbf{I}$^\prime$ is a \textbf{P-Res} inference based on \textbf{I}. Then, in the computation of~\textbf{I}$^\prime$, the \textbf{T-Res} rule further specifies the sub-multiset $N$ of $C_1, \ldots, C_n$ by the top-variable technique. Let~\textbf{I}$^{\prime\prime}$ be a \textbf{T-Res} inference based on \textbf{I}$^{\prime}$ in which $C$ is the main premise and the side premises are clauses in $N$. To ensure that the clauses in $N$ are the \textbf{P-Res} side premises in~\textbf{I}$^{\prime\prime}$, we use the complementary literals of the eligible literals of $N$ to restrict the inference and name these literals the \emph{top-variable literals}. Therefore, although the \textbf{T-Res} rule identifies the top-variable literals as per \textbf{S-Res} inference, the top-variable literals are not determined by a dynamic selection function, but by the presence of \textbf{S-Res} side premises. This top-variable technique provides the basis for our decision procedures discussed later. Since a \textbf{T-Res} inference is based on the existence of an \textbf{S-Res} inference, the mgu for the \textbf{T-Res} inference is ensured to exist, hence the top-variable literals in \textbf{T-Res} inferences can always be identified. To distinguish the mgus of the \textbf{T-Res} and the \textbf{S-Res} rules, we use $\sigma$ and $\sigma^\prime$ to denote them, respectively.

Now we provide the customised \emph{admissible orderings} and \emph{selection functions}. As admissible orderings, we choose to use any \emph{lexicographic path ordering} $\succ_{lpo}$ with a precedence in which function symbols are larger than constants, which are larger than predicate symbols. This is a requirement also for any admissible ordering with the same precedence restriction. For selection functions, we require the selection function $\SelectNC$ to select one of the negative compound-term literals in \textsf{LGQ} clauses containing negatively occurring compound-term literals.

Algorithm \ref{algorithm:refine} details how the admissible ordering $\succ_{lpo}$, the selection function $\SelectNC$ and the \textbf{T-Res} rule are applied to \textsf{LGQ} clauses. The algorithm contains the following functions:
\begin{itemize}
\item $\Max(C)$ returns a (\emph{strictly}) \emph{$\succ_{lpo}$-maximal literal} with respect to the clause $C$.
\item $\SelectNC(C)$ returns one of the \emph{negative compound-term literals} in the clause $C$.
\item $\TRes(N, C)$ performs a \textbf{T-Res} inference with clauses in $N$ the side premises and~$C$ the main premise, returning
\begin{enumerate}
\item either \emph{all negative literals} of the clause $C$, or 
\item the \emph{top-variable literals} of the clause $C$ (with respect to this \textbf{T-Res} inference).
\end{enumerate}
\end{itemize}

\begin{algorithm}[t]
\normalsize
  \DontPrintSemicolon
 \KwIn{An \textsf{LGQ} clausal set $N$ and a clause $C$ in $N$}
 \KwOut{The eligible or the top-variable literals in $C$}
 \If{$C$ is ground}{
  \Return $\Max(C)$
 }
 \ElseIf{$C$ has negatively occurring compound-term literals}{
  \Return $\SelectNC(C)$
 }
  \ElseIf{$C$ has positively occurring compound-term literals}{
  \Return $\Max(C)$
 }
 \lElse{
  \Return $\TRes(N, C)$
  {\tcp*[f]{$C$ is a non-ground flat clause}}
}
\caption{Find the eligible or the top-variable literals for \textsf{LGQ} clauses}
\label{algorithm:refine}
\end{algorithm}

\begin{algorithm}[b]
\normalsize
  \DontPrintSemicolon
 \KwIn{An \textsf{LGQ} clausal set $N$ and a non-ground flat clause $C$ in $N$}
 \KwOut{The eligible or the top-variable literals in $C$}
  \SetKwFunction{FMain}{$\TRes$}
  \SetKwProg{Fn}{Function}{:}{}
   \Fn{\FMain{$N, C$}}{
 Select all negative literals in $C$
 
 Find some clauses $C_1, \ldots, C_n$ in $N$ so that an \textbf{S-Res} inference is possible when $C$ is the main premise and $C_1, \ldots, C_n$ are the side premises
 
 \If{$C_1, \ldots, C_n$ exist}{\Return $\ComT(C_1, \ldots, C_n, C)$
 }
 \lElse{
 \Return all negative literals in $C$
}
}
\caption{The $\TRes$ function}
\label{algorithm:top}
\end{algorithm} 

Algorithm \ref{algorithm:top} specifies the $\TRes(N,C)$ function, describing the application of the \textbf{T-Res} rule to a non-ground flat \textsf{LGQ} clause $C$ as the main premise and $C_1, \ldots, C_n$ occurring in $N$ as the side premises. In Algorithm \ref{algorithm:top}, the $\ComT(C_1, \ldots, C_n, C)$ function finds the top-variable literals in $C$ with respect to the \textbf{S-Res} inference when $C_1, \ldots, C_n$ are the side premises and $C$ is the main premise. Algorithm \ref{algorithm:top} first tries to perform an \textbf{S-Res} inference on $C_1, \ldots, C_n$ and $C$, and if it is possible, the \textbf{S-Res} inference is immediately replaced by a \textbf{T-Res} inference. In the algorithm Lines 2--3 check whether the \textbf{S-Res} rule applies to $C_1, \ldots, C_n$ as the side premises and $C$ as the main premise with all negative literals selected. If so, Line 5 uses the $\ComT(C_1, \ldots, C_n, C)$ function to compute the \emph{top-variable literals} in $C$ with respect to this \textbf{S-Res} inference, ensuring that the \textbf{T-Res} rule is applicable to $C$ and \emph{the sub-multiset of $C_1, \ldots, C_n$ mapping to the top-variable literals in $C$}. Otherwise, Line 6 returns \emph{all negative literals} of $C$, meaning that no \textbf{S-Res} inference, hence no \textbf{T-Res} inference, is possible for $C_1, \ldots, C_n$ and $C$. Though the \textbf{T-Res} rule does not require one to select all negative literals in the \textbf{S-Res} main premise, the $\TRes$ function requires it because it is essential for deciding satisfiability of the \textsf{LGQ} clausal class.

The following sample derivation shows how the \textbf{T-Res} system decides an unsatisfiable set of \emph{\textsf{LG} clauses}. Consider an unsatisfiable set $N$ of \textsf{LG} clauses $C_1, \ldots, C_9$:
\begin{align*}
C_1 = \ & \lnot A_1(x,y) \lor \lnot A_2(y,z) \lor \lnot A_3(z,x) \lor B(x,y,b), \\
C_2 = \ & A_3(x,f(x)) \lor \lnot G_3(x), \qquad \qquad  \ \ \quad C_3 = A_2(f(x),f(x)) \lor \lnot G_2(x), \\
C_4 = \ & A_1(f(x),x) \lor D(g(x)) \lor \lnot G_1(x), \qquad \qquad \quad \qquad C_5 = \lnot B(x,y,b), \\
C_6 = \ & \lnot D(x), \ \ \quad C_7 = G_1(f(a)), \ \ \qquad C_8 = G_3(f(a)), \qquad C_9 = G_2(a).
\end{align*}
Suppose the precedence on which $\succ_{lpo}$ is based is $f > g > a > b > B > A_1 > A_2 > A_3 > D > G_1 > G_2 > G_3$. By $\boxed{L}$ or $L^\ast$ we mean that $L$ is selected or $L$ is a (strictly) maximal literal, respectively. In the \textbf{T-Res} system, $C_1, \ldots, C_9$ are presented as:
\begin{align*}
C_1 = \ & \boxed{\lnot A_1(x,y)} \lor \boxed{\lnot A_2(y,z)} \lor \boxed{\lnot A_3(z,x)} \lor B(x,y,b), \\
C_2 = \ & A_3(x,f(x))^\ast \lor \lnot G_3(x), \qquad \qquad \quad \ \ \ C_3 = A_2(f(x),f(x))^\ast \lor \lnot G_2(x), \\
C_4 = \ & A_1(f(x),x)^\ast \lor D(g(x)) \lor \lnot G_1(x), \qquad \qquad \qquad \quad C_5 = \boxed{\lnot B(x,y,b)}, \\
C_6 = \ & \boxed{\lnot D(x)}, \ \ \quad C_7 = G_1(f(a))^\ast, \ \ \quad C_8 = G_3(f(a))^\ast, \quad \quad C_9 = G_2(a)^\ast.
\end{align*}
One can use any clause to start a derivation, w.l.o.g.~we begin with $C_1$. For each newly derived clause, Algorithm \ref{algorithm:refine} is applied to determine the eligible or the top-variable literals of the clause. 
\begin{enumerate}
\item By Algorithm \ref{algorithm:refine} and the fact that $C_1$ is a non-ground flat \textsf{LG} clause, the $\TRes$ function is applied to $C_1$ and clauses in $N$. In Algorithm \ref{algorithm:top}, all negative literals in~$C_1$ are temporarily selected to check if the \textbf{S-Res} rule is applicable to $C_1$.  
\item As an \textbf{S-Res} inference step is applicable to $C_2, C_3, C_4$ as the side premises and $C_1$ as the main premise, the $\ComT(C_2, C_3, C_4, C_1)$ function computes an mgu 
\begin{align*}
\sigma^\prime = \{x \mapsto f(f(x^\prime)), y \mapsto f(x^\prime), z \mapsto f(x^\prime)\}	
\end{align*}
for variables of $C_1$. Hence $x$ is the only \emph{top variable} in $C_1$ and therefore $\lnot A_1(x,y)$ and $\lnot A_3(z,x)$ are the \emph{top-variable literals}. This means that based on the \textbf{S-Res} inference on $C$ and $C_2, C_3, C_4$, we intend to perform a special \textbf{P-Res} inference, viz., a \textbf{T-Res} inference, with $C$ the main premise and $C_2$ and $C_4$ the side premises.
\item The \textbf{T-Res} rule is applied to $C_2$ and $C_4$ as the side premises and $C_1$ as the main premise with an mgu $\sigma = \{x \mapsto f(x^\prime), y \mapsto x^\prime, z \mapsto x^\prime \}$, deriving 
\begin{align*}
C_{10} = \lnot A_2(x,x) \lor B(f(x),x,b)^\ast \lor D(g(x)) \lor \lnot G_1(x) \lor \lnot G_3(x),
\end{align*} 
with $x^\prime$ renamed as $x$. No resolution step can be performed on $C_3$ and $C_{10}$ for the lack of complementary eligible literals, nonetheless a resolution inference step can be performed between $C_{5}$ and $C_{10}$. 
\item By Algorithm \ref{algorithm:top}, the \textbf{S-Res} rule is applicable to $C_{5}$ as the main premise and $C_{10}$ as the side premise. Since $C_{5}$ contains only one negative literal, the literal is the \emph{top-variable literal} in $C_{5}$. Then applying the \textbf{T-Res} rule to $C_{10}$ and $C_{5}$ derives 
\begin{align*}
C_{11} = \lnot A_2(x,x) \lor D(g(x))^\ast \lor \lnot G_1(x) \lor \lnot G_3(x).
\end{align*}  
\item By Algorithm \ref{algorithm:top}, the \textbf{T-Res} rule is applicable to $C_{11}$ as the side premise and $C_{6}$ as the main premise with $\lnot D(x)$ the top-variable literal, deriving 
\begin{align*}
C_{12} = \boxed{\lnot A_2(x,x)} \lor \boxed{\lnot G_1(x)} \lor \boxed{\lnot G_3(x)}.
\end{align*} 
\item Due to the presence of $C_3, C_7, C_8$ and $C_{12}$ satisfy conditions of the $\TRes$ function, the $\ComT(C_3, C_7, C_8, C_{12})$ function finds that $x$ is the only top variable in $C_{12}$ with an mgu $\sigma^\prime = \{x \mapsto f(a)\}$. Hence all negative literals in $C_{12}$ are the \emph{top-variable literals}. Applying the \textbf{T-Res} rule to $C_3, C_7, C_8$ as the side premises and $C_{12}$ as the main premise derives $C_{13} = \boxed{\lnot G_2(a)}$.
\item Applying the \textbf{T-Res} rule to $C_9$ and $C_{13}$ derives $\bot$.
\end{enumerate}

Recall that by the \emph{term depth} of a clause, we mean the depth of the deepest term in that clause. As shown by the above example, the \textbf{T-Res} rule \emph{avoids term depth increase} in resolvents of \textsf{LGQ} clauses. Suppose the $\ComT(C_1, \ldots, C_n, C)$ function takes \textsf{LGQ} clauses $C_1, \ldots, C_n$ and~$C$ as input and $C$ is a non-ground flat \textsf{LGQ} clause. In the application of the \emph{top-variable technique} to $C_1, \ldots, C_n$ and~$C$, Step~1.~first computes an \textbf{S-Res} mgu of $C_1, \ldots, C_n$ and~$C$, and Steps 2.--3.~then find the variable~$x$ in $C$ that is unified to be the deepest term $x\sigma^\prime$ in~$C\sigma^\prime$ as the top variable. As  $x\sigma^\prime$ may become a nested term in the \textbf{S-Res} resolvent, the \textbf{T-Res} rule computes a partial resolvent, by only resolving the top-variable literals of~$C$, to avoid this potential term depth increase caused by $x\sigma^\prime$. In the previous example, if an \textbf{S-Res} inference is computed on~$C_1$ as the main premise and $C_2, C_3, C_4$ as the side premises, a nested compound term $f(f(x))$ will occur in the \textbf{S-Res} resolvent.

Now we give the main result of this section.
\begin{thm}
\label{thm:tres}
The \textbf{T-Res} system is sound and refutationally complete for general first-order clausal logic.
\end{thm}
\begin{proof}
By Theorem \ref{thm:pres} and since \textbf{T-Res} is a special case of the \textbf{P-Res} system.
\end{proof}

The definitions in the resolution framework of \cite{BG01} and most resolution-based decision procedures \cite{FATZ93} stipulate that eligibility, in particular (strict) maximality, of literals is determined on the instantiated premises with the mgus, i.e., \emph{a-posteriori eligibility} is used. Instead, \emph{a-priori eligibility} determines eligibility, in particular (strict) maximality, of literals on the non-instantiated premises. A-posteriori eligibility is more general and stronger than a-priori eligibility. However, a-priori eligibility is possible is more efficient, due to the overhead of pre-computing unifications.

The \textbf{T-Res} system uses a-posteriori eligibility, however, thanks to the \emph{covering} and \emph{strong compatibility properties} of the \textsf{LGQ} clausal class, one can use \emph{a-priori eligibility}. This is briefly mentioned in deciding satisfiability of guarded clauses with equality in \cite{GdN99}. We now formally prove this claim.


\begin{lem}
\label{lem:com_large}
Let a covering clause $C$ contain a compound-term literal $L_1$ and a non-compound-term literal $L_2$. Then $L_1 \succ_{lpo} L_2$.
\end{lem}
\begin{proof}
We distinguish two cases: i) Suppose $L_1$ contains a ground compound term. By the covering property, $C$ is ground. Then $L_1 \succ_{lpo} L_2$ as $L_1$ contains at least one function symbol but $L_2$ does not. 

ii) Suppose $L_1$ contains a non-ground compound term $t$. By the covering property, $\Var(t) = \Var(L_1) = \Var(C)$. By the facts that $\Var(L_2) \subseteq \Var(L_1)$ and $L_1$ contains at least one function symbol but $L_2$ does not, $L_1 \succ_{lpo} L_2$.	
\end{proof}

By the \emph{covering} and the \emph{strong compatibility properties} of \textsf{LGQ} clauses, a literal identified as eligible by a-posteriori eligibility is the same as the one identified by a-priori eligibility. This is formally stated as:

\begin{lem}
\label{lem:apri}
When applying the refinement of the \textbf{T-Res} system to an \textsf{LGQ} clause $C$, if a literal $L$ is (strictly) $\succeq_{lpo}$-maximal with respect to $C$, then $L\sigma$ is (strictly) $\succeq_{lpo}$-maximal with respect to $C\sigma$, for any substitution $\sigma$.
\end{lem}
\begin{proof}
In Algorithm~\ref{algorithm:refine}, the maximality checking is done in either Lines 1--2 or 5--6. 

For the case in Lines 1--2 the claim trivially holds as $C$ is ground. Lines 5--6 mean that $C$ contains compound-term literals. By Lemma \ref{lem:com_large}, $L$ is a compound-term literal. Suppose $L^\prime$ is a literal in $C$ distinct from~$L$. First, suppose $L^\prime$ is not a compound-term literal. By the covering property, $L \succeq_{lpo} L^\prime$ implies $L\sigma \succeq_{lpo} L^\prime\sigma$ for any substitution~$\sigma$. Next, suppose $L^\prime$ is a compound-term literal. By the fact that $C$ is strongly compatible, $L \succeq_{lpo} L^\prime$ implies $L\sigma \succeq_{lpo} L^\prime\sigma$ for any substitution $\sigma$. Thus, $L\sigma$ is (strictly) maximal with respect to $C\sigma$.
\end{proof}

Lemma \ref{lem:apri} is generalisable to any \emph{covering} and \emph{strongly compatible} clause, as it is these properties that make \emph{a-priori eligibility} determination possible. From now on we assume the use of \emph{a-priori eligibility} to determine (strictly) maximal literals in the \textbf{T-Res} system. This also streamlines the discussions and simplifies proofs.

\section{Deciding satisfiability of the \textsf{LG} clausal class}
\label{sec:lgc}
Having shown in the previous section that the \textbf{T-Res} system is sound and refutational complete, now we prove the system decides the \textsf{LG} clausal class. Our goal is to show: given a finite signature, applying the conclusion-deriving \textbf{Deduce} rules in the \textbf{T-Res} system to a set of \textsf{LG} clauses only derives \textsf{LG} clauses that are of bounded depth and width. This claim is achieved by restricting that in an \textsf{LG} clause~$C$, the eligible literals or the top-variable literals
\begin{enumerate}
\item have the same variables set as $C$, and
\item are the deepest literals in $C$.	 
\end{enumerate}


First, we show 1.
\begin{lem}
\label{lem:eligible_covering}
By the \textbf{T-Res} system, the eligible literals or the top-variable literals in an \textsf{LG} clause $C$ have the same variable set as $C$.
\end{lem}
\begin{proof}
Being led by Algorithm \ref{algorithm:refine}, we distinguish three cases:

Lines 1--2: When $C$ is ground the statement trivially holds.

Lines 3--6: Suppose $C$ is a non-ground compound-term \textsf{LG} clause and $L$ is an eligible literal in~$C$. Suppose $L$ is positive. By the $\Max$ function and $\succ_{lpo}$, $L$ is a positive compound-term literal. Next, suppose $L$ is negative. By the $\SelectNC$ function,~$L$~is a negative compound-term literal. In either case, by the covering property of \textsf{LG} clauses, $\Var(L) = \Var(C)$.

Lines 7: Suppose $C$ is a non-ground flat \textsf{LG} clause and $\mathbb{L}$ are the top-variable literals in $C$. Suppose $x$ is a top variable in $C$. By 2.~of Definition \ref{def:lgc} and the definition of top-variable literals, $x$ co-occurs with all other variables of $C$ in $\mathbb{L}$, therefore $\Var(\mathbb{L}) = \Var(C)$.
\end{proof}

For 2, the \textbf{T-Res} system ensures that the deepest literals in \textsf{LG} clauses are eligible. Specifically Lines 3--6 of Algorithm \ref{algorithm:refine} ensure that when an \textsf{LG} clause contains non-ground compound-terms, one of the compound-term literals is eligible. 

Compound-term covering clauses have the following property.

\begin{rem}
\label{re:ground}
Suppose $C$ is a covering clause and contains ground compound terms. Then,~$C$ is ground.	
\end{rem}
\begin{proof}
By the definition of the covering property.	
\end{proof}

Next, we look at the unification for the eligible literals of \textsf{LG} clauses. We first investigate the \emph{pairing} property of compound-term eligible literals. Recall the definition of pairing from \textbf{Section \ref{sec:pre}}: Given two atoms $A(\ldots, s, \ldots)$ and $B(\ldots, t, \ldots)$ with terms $s$ and $t$, we say $s$ \emph{pairs} $t$ if the argument position of $s$ in $A(\ldots, s, \ldots)$ is the same as that of $t$ in $B(\ldots, t, \ldots)$.

\begin{lem}
\label{lem:mat_comp_lit}
Let $A_1$ and $A_2$ be two simple and covering compound-term atoms, and suppose $A_1$ and $A_2$ are unifiable using an mgu $\sigma$. Then, compound terms in $A_1$ pair only compound terms in $A_2$ and vice-versa.  
\end{lem}
\begin{proof}
%

We distinguish three cases: i) The statement trivially holds when both $A_1$ and~$A_2$ are ground atoms. 

ii) Suppose one of $A_1$ and $A_2$ is a ground atom and the other one is a non-ground atom. By Remark \ref{re:ground}, the non-ground atom in $A_1$ and $A_2$ contains no ground compound terms. Hence, in this case, a non-ground compound term pairs either a ground compound term or a constant. As unifying a non-ground compound term with a constant is not possible, a non-ground compound term must pair a ground compound term.

iii) Suppose both $A_1$ and $A_2$ are non-ground. W.l.o.g., $A_1$ and $A_2$ are represented as $A_1(t,t^\prime,\ldots)$ and $A_2(u,u^\prime,\ldots)$, respectively. By Remark \ref{re:ground} and the fact that $A_1$ and~$A_2$ are non-ground atoms, if any of $t$, $t^\prime$, $u$ and $u^\prime$ is a compound term, then it is a non-ground compound term.

Suppose $t$ is a compound term. We prove that $u$ is a compound term by contradiction. Then $u$ can be either a constant or a variable. The case that $u$ is a constant prevents the unification of $t\sigma = u\sigma$. Now suppose $u$ is a variable. As $A_2$ is a compound-term literal, w.l.o.g., suppose $u^\prime$ is a compound term in $A_2$. Then $t^\prime$ is not a constant as it prevents the unification of $u^\prime$ and $t^\prime$, therefore, $t^\prime$ is a variable or a compound term. We distinguish the two cases of $t^\prime$: 1) Suppose $t^\prime$ is a variable. By the covering property, w.l.o.g., we use $f(\ldots, x, \ldots)$, $x$, $y$ and $g(\ldots, y, \ldots)$ to represent $t$, $t^\prime$, $u$ and $u^\prime$ respectively. Then $A_1(t,t^\prime,\ldots)$ and $A_2(u,u^\prime,\ldots)$ are represented as $A_1(f(\ldots, x, \ldots), x, \ldots)$ and $A_2(y, g(\ldots, y, \ldots), \ldots)$, respectively. The unification between these two atoms is impossible due to occur-check failure. 

2) Suppose $t^\prime$ is a compound term. By the covering property, w.l.o.g., we use~$f(\overline x)$, $g(\overline x)$, $y$ and $g(\ldots, y, \ldots)$ to represent $t$, $t^\prime$, $u$ and $u^\prime$ respectively. Then $A_1(t,t^\prime,\ldots)$ and $A_2(u,u^\prime,\ldots)$ are represented as $A_1(f(\overline x), g(\overline x), \ldots)$ and $A_2(y, g(\ldots, y, \ldots), \ldots)$, respectively. Then there exists no unifier for these two atoms again due to occur-check failure. The fact that $u$ is neither a constant nor a variable implies that $u$ is a compound term.
\end{proof}

The loose guard in the premise of \textbf{Factor} inferences or the loose guard in the side premise of \textbf{T-Res} inferences act as the loose guard of the conclusion. Formally:

\begin{lem}
\label{lem:guard_tinf} 
Let $A_1$ and $A_2$ be two simple and covering atoms, and suppose $A_1$ and $A_2$ are unifiable using an mgu $\sigma$. Further suppose $\mathbb{G}$ is a set of flat literals satisfying $\Var(A_1) = \Var(\mathbb{G})$. Then, if $A_1$ is a compound-term atom, $\Var(A_1\sigma) = \Var(\mathbb{G}\sigma)$ and all literals in $\mathbb{G}\sigma$ are flat.
\end{lem}
\begin{proof}
Since $\Var(A_1) = \Var(\mathbb{G})$, it is immediate that $\Var(A_1\sigma) = \Var(\mathbb{G}\sigma)$. 

We prove that $\mathbb{G}\sigma$ is a set of flat literals by distinguishing two cases: i) Assume that $A_2$ is flat. This implies that $\sigma$ substitutes variables in $A_1$ with either variables or constants. By the facts that $\mathbb{G}$ is a set of flat literals and $\Var(A_1) = \Var(\mathbb{G})$, all literals in $\mathbb{G}\sigma$ are flat.

ii) Assume that $A_2$ is a compound-term literal. By Lemma \ref{lem:mat_comp_lit}, compound terms in~$A_1$ pair compound terms in $A_2$ and vice-versa. Since $A_1$ and $A_2$ are simple, the mgu~$\sigma$ substitutes variables in $A_1$ with either variables or constants. Since~$\mathbb{G}$ is a set of flat literals and $\Var(A_1) = \Var(\mathbb{G})$, all literals in $\mathbb{G}\sigma$ are flat.
\end{proof}

Lemmas \ref{lem:simple}--\ref{lem:conclusion} below consider non-loose-guard literals in the conclusions of \textsf{LG} clauses. A similar result to Lemma \ref{lem:simple} is Lemma 4.6 in \cite{GdN99}, but a key `covering' condition is not considered. First, we look at the depth of eligible literals. 

\begin{lem}
\label{lem:simple}
Suppose $A_1$ and $A_2$ are two simple and covering atoms, and they are unifiable using an mgu $\sigma$. Then, $A_1\sigma$ is simple.
\end{lem}
\begin{proof}
If either of $A_1$ and $A_2$ is ground, or either of $A_1$ and $A_2$ is non-ground and flat, then immediately $A_1\sigma$ is simple. 

Let both $A_1$ and $A_2$ be compound-term atoms. By Lemma \ref{lem:mat_comp_lit} and since $A_1$ and $A_2$ are simple, the mgu~$\sigma$ substitutes variables with either constants or variables. Then, the fact that $A_1$ is simple implies that $A_1\sigma$ is simple.
\end{proof}
 
Next we study the depth and width of non-eligible literals in conclusions. 
 
\begin{lem}
\label{lem:conclusion}
Let $A_1$ and $A_2$ be two simple atoms satisfying $\Var(A_2) \subseteq \Var(A_1)$.
Then given an arbitrary substitution $\sigma$, these properties hold:
\begin{enumerate}
\item If $A_1\sigma$ is simple, then $A_2\sigma$ is simple. 
\item $\Var(A_2\sigma) \subseteq \Var(A_1\sigma)$.
\end{enumerate}
Further suppose that $t$ and $u$ are, respectively, compound terms occurring in $A_1$ and~$A_2$ satisfying $\Var(t) = \Var(u) = \Var(A_1)$. Then, $\Var(t\sigma) = \Var(u\sigma) = \Var(A_1\sigma)$.
\end{lem}
\begin{proof}
By the assumption that $A_1$ and $A_1\sigma$ are simple, $\sigma$ does not cause term depth increase in $A_1\sigma$. By the facts that $\Var(A_2) \subseteq \Var(A_1)$ and $A_2$ is simple, $A_2\sigma$ is simple. 

By the facts that $\Var(A_2) \subseteq \Var(A_1)$ and $\Var(t) = \Var(u) = \Var(A_1)$, it is immediate that $\Var(A_2\sigma) \subseteq \Var(A_1\sigma)$ and $\Var(t\sigma) = \Var(u\sigma) = \Var(A_1\sigma)$, respectively.
\end{proof}

Recall that a \emph{flat compound term} is a compound term containing only variables and constants as arguments. We consider how the \emph{strong compatibility} property holds in the conclusions.
\begin{lem}
\label{lem:com_term}
Let $s$, $s^\prime$, $t$ and $t^\prime$ be flat compound terms. Suppose $s$ and $t$ are compatible with $s^\prime$ and $t^\prime$, respectively. Then, if $s\sigma \doteq t\sigma$ with an arbitrary substitution $\sigma$, the following conditions are satisfied.
\begin{enumerate}
\item $s\sigma$ and $s^\prime\sigma$ are compatible, and $t\sigma$ and $t^\prime\sigma$ are compatible.
\item $s$ and $t$ are compatible, and $s\sigma$ and $t\sigma$ are compatible.
\item $s^\prime\sigma$ and $t^\prime\sigma$ are compatible.
\end{enumerate}
\end{lem}
\begin{proof}
Since $s$ and $t$ are, respectively, compatible with $s^\prime$ and $t^\prime$, $s\sigma$ and $t\sigma$ are compatible with $s^\prime\sigma$ and $t^\prime\sigma$, respectively. Since $s$ and $t$ are unifiable by $\sigma$, $s\sigma$ and $t\sigma$ are compatible. Then, 1.~implies that $s^\prime\sigma$ and $t^\prime\sigma$ are compatible.
\end{proof}

A compound-term \textsf{LG} clause with a compound-term literal removed is still an \textsf{LG} clause. We generalise this claim with applications of substitutions.

\begin{lem}
\label{lem:rem_lit}
Suppose $C = D \lor B$ is an \textsf{LG} clause with $B$ a compound-term literal. Further, suppose $\sigma$ is a substitution that substitutes all variables in $C$ with either constants or variables. Then, $D\sigma$ is an \textsf{LG} clause.
\end{lem}
\begin{proof}
If $\sigma$ is a ground substitution, the lemma trivially holds. Suppose $\sigma$ is a non-ground substitution. We prove that $D\sigma$ is simple, covering, strongly compatible and contains a loose guard. Since $C$ is an \textsf{LG} clause and $D$ is a subclause of $C$, $D$ is simple. Because $\sigma$ substitutes variables with either constants or variables, $D\sigma$ is simple. Let~$s$ and $t$ be two arbitrary compound terms in $D$. That $C$ is covering implies that $\Var(t) = \Var(C)$, hence $\Var(t) = \Var(D)$, and therefore $\Var(t\sigma) = \Var(D\sigma)$. Then~$D\sigma$ is covering. Since $C$ is strongly compatible, $s$ and $t$ are compatible. By~2.~of Lemma~\ref{lem:com_term}, $s\sigma$ and $t\sigma$ are compatible, hence $D\sigma$ is strongly compatible. Suppose~$\mathbb{G}$~is a set of flat literals that acts as a loose guard of~$C$. Then $\mathbb{G}$ is a loose guard of~$D$. Since~$\sigma$ substitutes variables with either constants or variables and all literals in~$\mathbb{G}$ are flat, all literals in $\mathbb{G}\sigma$ are flat. Since $\Var(\mathbb{G}) = \Var(C)$ and $D$ is a subclause of~$C$, $\Var(\mathbb{G}\sigma) = \Var(D\sigma)$. By the facts that $\sigma$ substitutes variables with either constants or variables and $\mathbb{G}$ is a loose guard of $D$, each pair of variables of $D\sigma$ co-occurs in a literal of~$\mathbb{G}\sigma$. Hence~$\mathbb{G}\sigma$~is a loose guard of $D\sigma$. Therefore, $D\sigma$ is an \textsf{LG} clause.
\end{proof}

We establish properties of applying the \textbf{T-Res} rule to a flat clause and \textsf{LG} clauses.

\begin{lem}
\label{lem:pres_gnd_pairing}
Suppose a \textbf{T-Res} inference happens to \textsf{LG} clauses as the side premises and a non-ground flat clause as the main premise, with Condition 2b.~of the \textbf{T-Res} rule satisfied. Then, the top variables in the main premise pair constants or compound terms in the side premises, and the non-top variables in the main premise pair constants or variables in the side premises.
\end{lem}
\begin{proof}
Assuming that a-priori eligibility is applied, the \textbf{T-Res} rule is simplified to:
\begin{align*}
 \prftree[l]{\textbf{T-Res}: }
    {B_1 \lor D_1, \ \ldots, \ B_m \lor D_m, \ \ldots, \ B_n \lor D_n}
    {\lnot A_1 \lor \ldots \lor \lnot A_{m} \lor \ldots \lor \lnot A_n \lor D}
    {(D_1 \lor \ldots \lor D_m \lor \lnot A_{m+1} \lor \ldots \lor \lnot A_{n} \lor D)\sigma}
\end{align*}
provided the following conditions are satisfied.
\begin{enumerate}
\item[1.] No literal is selected in $D_1, \ldots, D_n,D$ and $B_1, \ldots, B_n$ are strictly $\succ_{lpo}$-maximal with respect to $D_1, \ldots, D_n$, respectively. 
\item[2a.] If $n = 1$, then i) either $\lnot A_1$ is selected, or nothing is selected in $\lnot A_1 \lor D$ and $\lnot A_1$ is $\succ_{lpo}$-maximal with respect to $D$, and ii) $\sigma = \mgu(A_1 \doteq B_1)$ or
\item[2b.] there must exist an mgu $\sigma^\prime$ such that $\sigma^\prime = \mgu(A_1 \doteq B_1, \ldots, A_n \doteq B_n)$, then $\lnot A_1, \ldots, \lnot A_m$ are the \emph{top-variable literals} of $\lnot A_1 \lor \ldots \lor \lnot A_{m} \lor \ldots \lor \lnot A_n \lor D$ and $\sigma = \mgu(A_1 \doteq B_1, \ldots, A_m \doteq B_m)$ where $1 \leq m \leq n$.
\item[3.] All premises are variable disjoint.
\end{enumerate} 

W.l.o.g.~suppose $\lnot A_t(\ldots, x, \ldots, y, \ldots)$ is a literal in $\lnot A_1, \ldots, \lnot A_m$, $x$ is a top variable and $y$ is a non-top variable (if it exists). Further suppose $\sigma^\prime$ is the \textbf{S-Res} mgu that $\sigma^\prime = \mgu(A_1 \doteq B_1, \ldots, A_n \doteq B_n)$. Suppose $C_t = B_t(\ldots, t_1, \ldots, t_2, \ldots) \lor D_t$ is the side premise such that $A_t(\ldots, x, \ldots, y, \ldots)\sigma^\prime = B_t(\ldots, t_1, \ldots, t_2, \ldots)\sigma^\prime$, and $t_1$ and $t_2$ pair $x$ and~$y$, respectively.

We prove that $t_1$ is either a constant or a compound term and $t_2$ is either a constant or a variable by distinguishing two cases of $C_t$. i) Suppose $C_t$ is ground. Then, immediately $t_1$ is either a constant or a ground compound term. We prove that $t_2$ is a constant by contradiction. Assume that $t_2$ is a ground compound term. The fact that $C_t$ is simple implies $\Dep(t_2) \geq \Dep(t_1)$. Since $t_1$ and $t_2$ are ground, $\Dep(t_2\sigma^\prime) \geq \Dep(t_1\sigma^\prime)$, and $\Dep(y\sigma^\prime) \geq \Dep(x\sigma^\prime)$, which contradicts that $y$ is not a top variable. Therefore,~$t_2$~is a constant.

ii) Suppose $C_t$ is non-ground. Then Lines 3--7 in Algorithm \ref{algorithm:refine} are used to check eligibility in $C_t$. By the fact that the eligible literal in $C_t$ is positive, Lines 5--6 are applied to $C_t$, hence $C_t$ is a non-ground compound-term clause and $B_t(\ldots, t_1, \ldots, t_2, \ldots)$ is the $\succ_{lpo}$-strictly maximal with respect to $C_t$. By Lemma \ref{lem:com_large}, $B_t$ is a compound-term literal. We prove that $t_1$ is a compound term by contradiction. Assume~$t_1$ is either a variable or a constant. Since $B_t$ is a compound-term literal, there exists a compound term in $B_t$. W.l.o.g.,~we suppose $t$ is a compound term in~$B_t$ and suppose $z$ is the variable in $A_t$ that $t$ pairs. The covering property of~$C_t$ implies $\Var(t_1) \subseteq \Var(t)$. The fact that $\Dep(t_1) < \Dep(t)$ implies $\Dep(t_1\sigma^\prime) < \Dep(t\sigma^\prime)$, therefore $\Dep(x\sigma^\prime) < \Dep(z\sigma^\prime)$, which contradicts that $x$ is a top variable. Then, $t_1$ is a compound term. Next, we prove that $t_2$ is either a constant or a variable again by contradiction. Assume $t_2$ is a compound term. Since $C_t$ is covering, $\Var(t_1) = \Var(t_2)$. Since $\Dep(t_1) = \Dep(t_2)$, $\Dep(t_1\sigma^\prime) = \Dep(t_2\sigma^\prime)$, and therefore $\Dep(x\sigma^\prime) = \Dep(y\sigma^\prime)$, which contradicts that $y$ is not a top variable. Hence, $t_2$ is either a variable or a constant.
\end{proof}

Lemma \ref{lem:pres_gnd_pairing} allows us to analyse unification in \textbf{T-Res} inferences, formally stated in the following corollary.
\begin{col}
\label{col:tres_uni}
In an application of the \textbf{T-Res} rule to \textsf{LG} clauses as the side premises and a non-ground flat clause as the main premise, with Condition 2b.~of the \textbf{T-Res} rule satisfied, the following conditions hold.
\begin{enumerate}
\item An mgu $\sigma$ substitutes top variables $x$ with either constants or the compound term pairing $x$ modulo variable renaming and grounding, and substitutes non-top variables with either constants or variables.
\item An mgu $\sigma$ substitutes variables in the eligible literals of the side premises with either constants or variables.
\end{enumerate}
\end{col}
\begin{proof}
1: By the pairing property established in Lemma \ref{lem:pres_gnd_pairing}.

2: Suppose $B(\ldots, x, \ldots)$ is an eligible literal in one of the side premises, and suppose $x$ is a variable argument in $B(\ldots, x, \ldots)$. By Lemma \ref{lem:pres_gnd_pairing} and the fact that the main premise is a non-ground flat clause, $x$ pairs either a constant or a variable, therefore~$\sigma$ substitutes~$x$ with either a constant or a variable.
\end{proof}

If a top-variable pairs a constant, the way a \textbf{T-Res} inference is performed is clear.

\begin{lem}
\label{lem:tres_cons}
Suppose a \textbf{T-Res} inference happens to \textsf{LG} clauses as the side premises and a non-ground flat clause as the main premise, with Condition 2b.~of the \textbf{T-Res} rule satisfied. Then, if a top variable $x$ pairs a constant, then i) all negative literals in the main premise are selected and ii) the mgu is a ground substitution instantiating all variables in the eligible literals and the top-variable literals with only constants.
\end{lem}
\begin{proof}
Suppose $\sigma^\prime$ is the mgu of the \textbf{S-Res} inference that ensures this application of the \textbf{T-Res} rule. By the definition of the top-variable technique, for any non-top variable $y$ in the main premise, $\Dep(x\sigma^\prime) > \Dep(y\sigma^\prime)$. The fact that $x$ pairs a constant indicates that $\Dep(x\sigma^\prime) = 0$, therefore $\Dep(y\sigma^\prime) = 0$. Then, all variables in the main premise are top variables and they pair either constants or variables. By Lemma \ref{lem:pres_gnd_pairing}, these top variables pair constants. Hence, $\sigma^\prime$ is a ground substitution that substitutes all variables with only constants.	
\end{proof}

Next, we formally show that the \textbf{T-Res} rule prevents term depth increase in the \textbf{T-Res} resolvents of a non-ground flat clause and \textsf{LG} clauses.

\begin{lem}
\label{lem:pres_gnd_deep}
In an application of the \textbf{T-Res} rule to \textsf{LG} clauses as the side premises and a non-ground flat clause as the main premise, with Condition 2b.~of the \textbf{T-Res} rule satisfied, the \textbf{T-Res} resolvent is no deeper than at least one of its premises.
\end{lem}
\begin{proof}
By 1.--2.~of Corollary \ref{col:tres_uni}.	
\end{proof}

Finally, we investigate the applications of the \textbf{Factor} and \textbf{T-Res} rules to \textsf{LG} clauses, starting with the \textbf{Factor} rule.

\begin{lem}
\label{lem:fact}
In the application of the \textbf{Factor} rule in the \textbf{T-Res} system to \textsf{LG} clauses, the factors are \textsf{LG} clauses.
\end{lem}
\begin{proof}
Assuming a-priori eligibility, the \textbf{Factor} rule simplifies to:
\begin{displaymath}
\prftree[l]{\textbf{Factor}: \quad}
  {C \lor A_1 \lor A_2}
  {(C \lor A_1)\sigma}
\end{displaymath}
if the following conditions are satisfied.
\begin{enumerate}
\item Nothing is selected in $C \lor A_1 \lor A_2$.
\item $A_1$ is $\succ_{lpo}$-maximal with respect to $C$.
\item $\sigma = \mgu(A_1 \doteq A_2)$.
\end{enumerate}

Suppose $C^\prime = C \lor A_1 \lor A_2$ and the premise $C^\prime$ is an \textsf{LG} clause. By the definition of the \textbf{Factor} rule, $A_1$ is the eligible literal and it is positive. Since Lines 3--4 and Line 7 in Algorithm \ref{algorithm:refine} select negative literals of \textsf{LG} clauses as the eligible or the top-variable literals, either Lines 1--2 or Lines 5--6 in Algorithm \ref{algorithm:refine} are applicable to~$C^\prime$. We distinguish these cases: 

Suppose $C^\prime$ satisfies Lines 1--2. Then the premise $C^\prime$ is a ground \textsf{LG} clause, and it is immediate that the factor $(C \lor A_1)\sigma$ is a ground \textsf{LG} clause.

Suppose $C^\prime$ satisfies Lines 5--6. Then $C^\prime$ is a non-ground \textsf{LG} clause containing positive compound-term literals, but no negative compound-term literals. By Lemma~\ref{lem:com_large} and the fact that $C^\prime$ is covering, $A_1$ is a compound-term literal. By Remark~\ref{re:ground} and the fact that $C^\prime$ is not ground, $A_1$ is a non-ground compound-term literal. By the covering property of $C^\prime$, $\Var(A_2) \subseteq \Var(A_1)$. We prove that $A_2$ is a compound-term literal by contradiction. Suppose $A_2$ is a flat literal. Because $\Var(A_2) \subseteq \Var(A_1)$ and $A_1$ is a compound-term literal, a compound term $t$ in~$A_1$ pairs either a variable that occurs in~$t$, or a constant. Due to occur-check failure, in neither case $A_1$ and $A_2$ are unifiable, which refutes the fact that $A_1$ and $A_2$ are unifiable. Hence, $A_2$ is a compound-term literal. The fact that $C^\prime$ is covering implies that $\Var(A_2) = \Var(A_1)$. By Lemma~\ref{lem:mat_comp_lit} and the fact that $C^\prime$ is covering, the mgu~$\sigma$ substitutes variables with either variables or constants. By Lemma~\ref{lem:rem_lit} and since $C^\prime$ is a compound-term \textsf{LG} clause, the factor $(C \lor A_1)\sigma$ is an \textsf{LG} clause.
\end{proof}

\begin{lem}
\label{lem:res} 
In the application of the \textbf{T-Res} rule to \textsf{LG} clauses, the resolvents are \textsf{LG} clauses.
\end{lem}
\begin{proof}
We consider \textbf{T-Res} inferences by distinguishing all possible cases of the main premise. Suppose an \textsf{LG} clause $C = \lnot A_1 \lor D$ is the \textbf{T-Res} main premise. In Algorithm~\ref{algorithm:refine}, $C$ satisfies either Lines 1--4 or Line 7. 

First, we consider the cases where the main premise satisfies either Lines 1--2 or Lines 3--4 in Algorithm \ref{algorithm:refine}. In these cases, the eligible literal in the main premise $C$ is either selected or is maximal with respect to $C$. Then Condition 2a.~of the \textbf{T-Res} rule is applied to the main premise and the \textbf{T-Res} inference is reduced to a binary \textbf{T-Res} inference without using the top-variable technique. W.l.o.g.,~suppose in a \textbf{T-Res} inference, an \textsf{LG} clause $C_1 = B_1 \lor D_1$ is the side premise and the resolvent $R = (D_1 \lor D)\sigma$ where $\sigma$ the mgu of $B_1$ and $A_1$. Further, suppose $C$ satisfies either Lines 1--2 or Lines 3--4 in Algorithm~\ref{algorithm:refine}. Since the eligible literal in~$C_1$ is positive, $C_1$ satisfies either Lines 1--2 or Lines~5--6 in Algorithm \ref{algorithm:refine}.

Suppose $C$ satisfies Lines 1--2. Then $C$ is a ground \textsf{LG} clause. We distinguish the cases of $\lnot A_1$. 

1) Suppose $\lnot A_1$ is a ground flat literal. The fact that no selection function in the \textbf{T-Res} system selects negative ground literals implies that the eligibility of~$\lnot A_1$, because~$\lnot A_1$ is maximal with respect to $C$, therefore $C$ is a flat clause. The facts that $A_1$ and $B_1$ are unifiable and $A_1$ is a flat ground literal imply that $B_1$ is a flat literal. The fact that $B_1$ is strictly $\succ_{lpo}$-maximal with respect to $C_1$ implies that $C_1$ is a flat clause. Since the eligible literal $B_1$ in the flat \textsf{LG} clause $C_1$ is a flat literal, $C_1$ is a ground clause satisfying Lines 1--2 in Algorithm \ref{algorithm:refine}. Since both $C$ and~$C_2$ are flat ground clauses, the resolvent $R$ is a flat ground clause. Hence, $R$ is an \textsf{LG} clause. 

2) Next, suppose $\lnot A_1$ is a ground compound-term literal. By Remark~\ref{re:ground}, $C$ is a ground compound-term \textsf{LG} clause. Since $C_1$ is an \textsf{LG} clause, $B_1$ is either a compound-term literal or a flat literal. Since $B_1$ is maximal with respect to $C_1$, the assumption that $B_1$ is flat implies that $B_1$ is ground, otherwise, negative literals in $C_1$ will be selected. However, if $B_1$ is ground, the unification between $A_1$ and~$B_1$ is impossible due to a clash.
Then, $B_1$ is a compound-term literal. Suppose $B_1$ is ground. By Remark~\ref{re:ground},~$C_1$ is a ground compound-term \textsf{LG} clause. The fact that $C$ and $C_1$ are both ground compound-term \textsf{LG} clauses implies that applying the \textbf{T-Res} rule to $C$ and $C_1$ derives a ground \textsf{LG} clause. Next, suppose $B_1$ is a non-ground compound-term literal. By Lemma \ref{lem:mat_comp_lit} and since $A_1$ and $B_1$ are unifiable by the mgu $\sigma$, the mgu $\sigma$ substitutes the variables in $B_1$ with constants. By Lemma \ref{lem:eligible_covering} and because $B_1$ is the eligible literal in $C_1$, $\sigma$ substitutes all variables in $C_1$ with constants, therefore $C_1\sigma$ is a ground compound-term \textsf{LG} clause. Since $C$ is ground, applying the \textbf{T-Res} rule to $C$ and $C_1$ derives the same resolvent as applying the \textbf{T-Res} rule to $C$ and $C_1\sigma$. The fact that $C$ and $C_1\sigma$ are ground compound-term \textsf{LG} clauses implies that applying the \textbf{T-Res} rule to $C$ and $C_1\sigma$ derives a ground \textsf{LG} clause. Hence, the resolvent $R$ is an \textsf{LG} clause. 

Suppose $C$ satisfies Lines 3--4. Then $C$ contains negative compound-term literals. By Remark~\ref{re:ground} and since $C$ is not ground, the literal $\lnot A_1$ contains non-ground compound terms, and therefore $\lnot A_1$ is selected by the $\SelectNC$ function. We now distinguish the possible cases of $B_1$. 

i) Suppose $B_1$ is a flat literal. Similar to the proof in~2) that $B_1$ cannot be a flat literal, the assumption that $B_1$ is flat implies that $B_1$ is ground. This makes the unification between $A_1$ and~$B_1$ impossible due to a clash. Hence, $B_1$ cannot be flat.

ii) Suppose $B_1$ is a compound-term literal. We distinguish two cases of $B_1$. 

ii)-i) First, consider $B_1$ as a ground compound-term literal. By Lemma \ref{lem:mat_comp_lit} and the fact that $A_1$ and $B_1$ are unifiable, the mgu $\sigma$ substitutes all variables in $A_1$ with constants. By the fact that $A_1$ is a compound-term literal of $C_1$ and the covering property of the \textsf{LG} clauses, $\sigma$ substitutes all variables in $C_1$ with constants, therefore~$C_1\sigma$ is a ground compound-term \textsf{LG} clause. As $C$ is ground, applying the \textbf{T-Res} rule to $C$ and~$C_1$ derives the same resolvent as the one when applying the \textbf{T-Res} rule to $C$ and~$C_1\sigma$. The fact that $C$ and~$C_1\sigma$ are ground compound-term \textsf{LG} clauses implies that applying the \textbf{T-Res} rule to $C$ and $C_1\sigma$ derives a ground \textsf{LG} clause. Hence, the resolvent $R$ is an \textsf{LG} clause. 

ii)-ii) Next, suppose~$B_1$ is a non-ground compound-term literal. By Lemma \ref{lem:mat_comp_lit} and the fact that $A_1$ and $B_1$ are two unifiable simple compound-term literals, the $\sigma$ substitutes the variables in $A_1$ and~$B_1$ with variables or constants. By Lemma \ref{lem:eligible_covering},~$\sigma$~substitutes the variables in $C$ and~$C_1$ with variables or constants. If the mgu $\sigma$ is a ground substitution, then both $C\sigma$ and~$C_1\sigma$ are ground \textsf{LG} clauses, therefore applying the \textbf{T-Res} rule to $C\sigma$ and $C_1\sigma$ derives a ground \textsf{LG} clause. Suppose $\sigma$ is a non-ground substitution. First, we prove that there is a loose guard in the resolvent $R$. Suppose~$\mathbb{G}$ is a set of flat literals that act as a loose guard of $C_1$. By Lemma~\ref{lem:guard_tinf} and because $A_1$ and $B_1$ are covering, simple and unifiable by the mgu $\sigma$, $\Var(A_1\sigma) = \Var(\mathbb{G}\sigma)$. By Lemma~\ref{lem:eligible_covering}, $\Var(A_1\sigma) = \Var(C\sigma)$ and $\Var(B_1\sigma) = \Var(C_1\sigma)$, therefore, $\Var(\mathbb{G}\sigma) = \Var(C_1\sigma) = \Var(C\sigma)$. Then $\Var(\mathbb{G}\sigma) = \Var(R)$. By the variable co-occurrence property of \textsf{LG} clauses and because $\mathbb{G}$ is a loose guard of $C_1$, each pair of variables in~$C_1$ co-occurs in a literal of $\mathbb{G}$. Since $\Var(\mathbb{G}\sigma) = \Var(C_1\sigma) = \Var(R)$ and $\sigma$ substitutes the variables in $C_1$ and $C$ with variables and constants, each pair of variables in~$R$ co-occurs in a literal of $\mathbb{G}\sigma$ and all literals in $\mathbb{G}\sigma$ are flat. Hence, $\mathbb{G}\sigma$ is a loose guard of the resolvent $R$. Next, we prove that $R$ is simple. Suppose $L$ is a literal in either $C$ or $C_1$. By Lemma~\ref{lem:eligible_covering}, either $\Var(L) \subseteq \Var(A_1)$ or $\Var(L) \subseteq \Var(B_1)$. Because $\sigma$ substitutes the variables in either $A_1$ or $B_1$ with either variables or constants, $A_1\sigma$ and~$B_1\sigma$ are simple. By 1.~in Lemma~\ref{lem:conclusion},~$L\sigma$ is simple. Hence, the resolvent $R$ is simple. Next, we prove that $R$ is covering. Because the mgu $\sigma$ substitutes the variables in $C_1$ and $C$ with variables and constants, the compound terms in $R$ come from compound terms in either $C_1$ or $C$. Suppose $t$ is a compound term in either $C$ or~$C_1$. By Remark~\ref{re:ground} and since both $C$ and $C_1$ are non-ground, $t$ is a non-ground compound term literal. By Lemma~\ref{lem:conclusion} and the covering property of \textsf{LG} clauses, either $\Var(t\sigma) = \Var(A_1\sigma)$ or $\Var(t\sigma) = \Var(B_1\sigma)$. The fact that either $\Var(A_1\sigma) = \Var(R)$ or $\Var(B_1\sigma) = \Var(R)$ implies that $\Var(t\sigma) = \Var(R)$, therefore the resolvent $R$ is covering. Finally, we prove that $R$ is strongly compatible. By the fact that $\sigma$ substitutes the variables in $C$ and~$C_1$ with variables and constants, the compound terms in the resolvent $R$ are inherited from compound terms that exist in $C$ or $C_1$. W.l.o.g.~suppose $s$ and $t$ are respectively compound terms in $A_1$ and $B_1$, and $s$ pairs $t$. Further, suppose $s_1$ is a compound term in $C$ that is distinct from $s$, and $t_1$ is a compound term in~$C_1$ that is distinct from~$t$. By~3. of Lemma \ref{lem:com_term} and the fact that $s$ and $t$ are unifiable by the mgu $\sigma$, $s_1\sigma$ is compatible with~$t_1\sigma$. Then all compound terms in the resolvent $R$ are compatible. Hence,~$R$ is strongly compatible. Because $R$ is simple, covering, strongly compatible and~$R$ contains a loose guard, $R$ is an \textsf{LG} clause.

Next, we consider the case when a \textbf{T-Res} main premise satisfies Line 7. This means that the premise is a non-ground flat \textsf{LG} clause. These \textbf{T-Res} inferences happen when the main premise satisfies Condition~2b.~and hence the top-variable technique is applied. Assume that in an \textbf{T-Res} inference, \textsf{LG} clauses $C_1 = B_1 \lor D_1, \ldots, C_n = B_n \lor D_n$ are the side premises, an \textsf{LG} clause $C = \lnot A_1 \lor \ldots \lor \lnot A_{m} \lor \ldots \lor \lnot A_n \lor D$ is the main premise with $\lnot A_1 \lor \ldots \lor \lnot A_{m}$ the top-variable subclause and the resolvent is $R = (D_1 \lor \ldots \lor D_m \lor \lnot A_{m+1} \lor \ldots \lor \lnot A_{n} \lor D)\sigma$, where $\sigma$ is the the mgu such that $\sigma = \mgu(A_1 \doteq B_1, \ldots, A_m \doteq B_m)$ where $1 \leq m \leq n$. Suppose $C$ is a non-ground flat \textsf{LG} clause. By Corollary~\ref{col:tres_uni}, the mgu $\sigma$ substitutes the top variables in $C$ with constants or compound terms, it substitutes non-top variables in $C$ with constants or variables and it substitutes all variables in $C_1, \ldots, C_m$ with constants or variables. We distinguish two possible cases of the mgu $\sigma$: 

1.~Suppose~$\sigma$ substitutes a top variable with a constant. By Lemma \ref{lem:tres_cons}, all variables in the top-variable subclause $\lnot A_1 \lor \ldots \lor \lnot A_{m}$ are substituted with constants. Hence, $B_1, \ldots, B_n$ are flat literals. Since the strictly $\succ_{lpo}$-maximal literal $B_i$ with respect to $C_i$ is flat, $C_i$ is a flat ground clause, for each $i$ such that $1 \leq i \leq n$. By Lemma~\ref{lem:eligible_covering} and since $C$ is an \textsf{LG} clause, $\sigma$ substitutes all variables in $C$ with constants. Applying the \textbf{T-Res} rule to flat ground \textsf{LG} clauses $C_1, \ldots, C_m$ and $C$ derives the same conclusions as applying the \textbf{T-Res} rule to $C_1, \ldots, C_m$ and $C\sigma$. Since applying the \textbf{T-Res} rule to $C_1, \ldots, C_m$ and $C\sigma$ derive a flat ground clause, applying the \textbf{T-Res} rule to $C_1, \ldots, C_m$ and $C$ also derives flat ground clauses. Hence, the resolvent~$R$ is an \textsf{LG} clause.

2.~Next, suppose the mgu $\sigma$ substitutes no top variables with constants. First, we establish intermediate results of unification on top variables. Suppose $x$ is a top variable and $\lnot A_t$ is the literal in $\lnot A_1, \ldots \lnot A_m$ where $x$ occurs. Further, suppose $B_t$ is a literal in the side premises satisfying $B_t\sigma \doteq A_t\sigma$. W.l.o.g.~suppose $C_t$ is a side premise in $C_1, \ldots, C_m$ and $C_t = B_t \lor D_t$. By the assumption that the mgu $\sigma$ substitutes no top variables with constants and $B_t$ pairs the top-variable literal $A_t$, $B_t$ is a compound-term literal. Suppose~$t$ is the compound term in $B_t$ that pairs $x$. The fact that $B_t\sigma \doteq A_t\sigma$ implies that $\Var(B_t\sigma) = \Var(A_t\sigma)$. By the covering property of \textsf{LG} clauses and the fact that $t$ is a compound term, $\Var(t) = \Var(B_t)$, therefore $\Var(t\sigma) = \Var(B_t\sigma)$. The fact that $x$ pairs $t$ implies that $\Var(x\sigma) = \Var(t\sigma)$, therefore $\Var(x\sigma) = \Var(B_t\sigma)$. Since $\Var(B_t\sigma) = \Var(A_t\sigma)$, $\Var(x\sigma) = \Var(A_t\sigma)$. By the variable co-occurrence property of \textsf{LG} clauses, $x$ co-occurs with all other variables in~$C$. Because $x$ is a top-variable, in the literals of $\lnot A_1, \ldots, \lnot A_m$, $x$ co-occurs with all other variables in $C$. Suppose $y$ is a variable in $\lnot A_1, \ldots, \lnot A_m$, and w.l.o.g.~suppose~$x$ and $y$ co-occurs in $A_1$. The fact that $\Var(x\sigma) = \Var(A_t\sigma)$ implies that $\Var(x\sigma) = \Var(A_1\sigma)$, therefore $\Var(y\sigma) \subseteq \Var(x\sigma)$. Hence for each variable $y$ in $\lnot A_1, \ldots, \lnot A_m$, $\Var(y\sigma) \subseteq \Var(x\sigma)$. Then, for each $A_i$ in $A_1, \ldots, A_m$, $\Var(A_i\sigma) = \Var(x\sigma)$. By the covering property of the \textsf{LG} clauses, for each~$B_i$ in $B_1, \ldots, B_m$, $\Var(B_i) = \Var(D_i)$. Since $A_i$ and~$B_i$ are unifiable using the mgu~$\sigma$, $\Var(A_i\sigma) = \Var(B_i\sigma)$ for each~$i$ such that $1 \leq i \leq m$. Then $\Var(x\sigma) = \Var(B_i\sigma)$, and therefore $\Var(x\sigma) = \Var(D_i\sigma)$ for each~$i$ such that $1 \leq i \leq m$. By Lemma \ref{lem:eligible_covering}, $\Var(\lnot A_1 \lor \ldots \lor \lnot A_m) = \Var(C)$. Hence, $\Var(x\sigma) = \Var((\lnot A_{m+1} \lor \ldots \lor \lnot A_{n} \lor D)\sigma)$. Then, $\Var(x\sigma) = \Var(t\sigma) = \Var(R)$.

Following 2.~we also need to prove that the resolvent $R$ contains a loose guard. Suppose~$C_i = B_i \lor D_i$ is a side premise in $C_1, \ldots, C_m$, $t$ is a compound term in $B_i$, $x$ is the top-variable that~$t$ pairs. Further, suppose $\mathbb{G}$ is a set of negative flat literals acting as a loose guard of~$C_i$. By 2.~of Corollary~\ref{col:tres_uni}, all literals in $\mathbb{G}$ are flat. By the definition of \textsf{LG} clauses, $\Var(\mathbb{G}) = \Var(t)$. By the result established in the previous paragraph and as $\Var(\mathbb{G}\sigma) = \Var(t\sigma)$, $\Var(\mathbb{G}\sigma) = \Var(R)$. By the variable co-occurrence property of \textsf{LG} clauses, each pair of variables in $\mathbb{G}\sigma$ co-occurs in a literal of $\mathbb{G}\sigma$, therefore each pair of variables in $\mathbb{G}\sigma$ co-occurs in a literal of $R$. The fact that all literals in~$\mathbb{G}\sigma$~are flat implies that $\mathbb{G}\sigma$ act as a loose guard of the resolvent~$R$. Next, we prove that $R$ is covering. The fact that $C$ is a flat clause implies that all compound terms in~$R$~come from the side premises. Suppose~$C_i = B_i \lor D_i$ is a side premise in $C_1, \ldots, C_m$ and $t$ is a compound term in $B_i$. W.l.o.g.~further suppose $s$ is a compound term in $D_i$. By the covering property of \textsf{LG} clauses, $\Var(s) = \Var(t)$ and $\Var(s\sigma) = \Var(t\sigma)$ with $\sigma$~as the mgu. By the result established in the previous paragraph, $\Var(s\sigma) = \Var(R)$. Then, the resolvent $R$ is covering. Next, we prove that~$R$ is strongly compatible. Again, we consider compound terms in the side premises since all compound terms in $R$ come from the side premises. Suppose $t_1$ and $t_2$ are two flat compound terms in $D_1, \ldots, D_m$. We prove that $R$ is strongly compatible by showing that $t_1\sigma$ and $t_2\sigma$ are compatible. Suppose $C_1 = B_1 \lor D_1$ and $C_2 = B_2 \lor D_2$ are two side premises in $C_1, \ldots, C_m$ and~w.l.o.g.~suppose $t_1$ and $t_2$ occur in $D_1$ and~$D_2$, respectively. By the assumption that the mgu~$\sigma$ substitutes no top variables with constants and the fact that $B_1$ and $B_2$ pair the top-variable literals, $B_1$ and $B_2$ are compound-term literals. W.l.o.g.~suppose $s_1$ and~$s_2$ are two flat compound terms in~$B_1$ and $B_2$, respectively. Further suppose $s_1$ and $s_2$ pair top variables $x_1$ and $x_2$, respectively. By the variable co-occurrence property of \textsf{LG} clauses, $x_1$ and $x_2$ co-occur in at least one literal in $\lnot A_1, \ldots, \lnot A_m$. W.l.o.g.~suppose $\lnot A_3$ is a literal where`$x_1$ and~$x_2$ co-occur. Suppose $C_3 = B_3 \lor D_3$ is a side premise and $A_3\sigma \doteq B_3\sigma$. Further suppose~$u_1$ and $u_2$ are flat compound terms in $B_3$ that pair~$x_1$ and $x_2$, respectively. By the strong compatibility property of \textsf{LG} clauses, $u_1\sigma$ is compatible with~$u_2\sigma$, therefore, $x_1\sigma$ is compatible with~$x_2\sigma$. Since $x_1$ pairs~$s_1$ and $x_2$ pairs $s_2$, $s_1\sigma$ is compatible with~$s_2\sigma$. By the strong compatibility property of \textsf{LG} clauses, $s_1$ and $s_2$ are compatible with $t_1$ and $t_2$, respectively. Hence~$s_1\sigma$ and $s_2\sigma$ are compatible with $t_1\sigma$ and~$t_2\sigma$, respectively. By the fact that $s_1\sigma$ is compatible with $s_2\sigma$, $t_1\sigma$ is compatible with $t_2\sigma$, therefore all compound terms in the resolvent~$R$~are compatible. Then,~$R$~is strongly compatible. Finally, we prove that the resolvent $R$ is a simple clause. By~1.~of Corollary~\ref{col:tres_uni}, the mgu $\sigma$ substitutes the variables in $\lnot A_{m+1} \lor \ldots \lor \lnot A_{n} \lor D$ with either variables, constants or flat compound terms. By~2.~of Corollary~\ref{col:tres_uni}, the mgu~$\sigma$~substitutes the variables in $D_1, \ldots, D_m$ with either variables or constants. Because $\lnot A_{m+1} \lor \ldots \lor \lnot A_{n} \lor D$ is a flat clause and $D_1, \ldots, D_m$ are simple clauses, the resolvent $(D_1 \lor \ldots \lor D_m \lor \lnot A_{m+1} \lor \ldots \lor \lnot A_{n} \lor D)\sigma$ is a simple clause. Then, the resolvent $R$ is an \textsf{LG} clause.
\end{proof}

Lemmas \ref{lem:fact}--\ref{lem:res} prove that applying the \textbf{Factor} and \textbf{T-Res} rules to \textsf{LG} clauses derive \textsf{LG} clauses. The derived \textsf{LG} clauses are of bounded depth as the clauses are simple. We now investigate the width of the derived clauses. Recall that by the \emph{width} of a clause, we mean the number of distinct variables in the clause.

\begin{lem}
\label{lem:bounded_width}
In applications of the \textbf{T-Res} system to \textsf{LG} clauses, the derived \textsf{LG} clause is no wider than at least one of its premises. 
\end{lem}
\begin{proof}
We distinguish the applications of the \textbf{Factor} rule and the \textbf{T-Res} rule: i) By Lemma \ref{lem:fact}, the conclusions of applying \textbf{Factor}  to \textsf{LG} clauses are \textsf{LG} clauses. The proof in Lemma \ref{lem:fact} shows that the loose guard of the factor is from the loose guard of the premise (modulo variable renaming and ground instantiations). The fact that a loose guard contains all variables of an \textsf{LG} clause implies that the factor of an \textsf{LG} clause is no wider than its premise.

ii) By Lemma \ref{lem:res}, the conclusions of applying \textbf{T-Res} to \textsf{LG} clauses are \textsf{LG} clauses. The proof in Lemma \ref{lem:res} shows that the loose guard of the derived \textsf{LG} clauses is inherited from one of the \textbf{T-Res} side premises (modulo variable renaming and ground instantiation), therefore any derived \textsf{LG} clause is no wider than at least one of its \textbf{T-Res} side premises.
\end{proof}

Finally, we give the main result of this section.

\begin{thm}
\label{thm:lgc}
The \textbf{T-Res} system decides satisfiability of the \textsf{LG} clausal class.
\end{thm}
\begin{proof}
By Lemmas~\ref{lem:fact}--\ref{lem:res}, applying the \textbf{T-Res} system to \textsf{LG} clauses derives \textsf{LG} clauses with bounded depth. By Lemma \ref{lem:bounded_width}, the derived \textsf{LG} clauses have bounded width. As no fresh symbols are introduced in the derivation, the \textbf{T-Res} system decides the \textsf{LG} clausal class.
\end{proof}

\section{Handling query clauses}
\label{sec:qc}



\subsection*{\textbf{Basic notions of query clauses}}
\label{sect:query}

Recall that a \emph{query clause} is a negative flat clause. Since there is no restriction on the occurrences of the variables in query clauses, analysing the conclusions of these clauses is non-trivial. To better manipulate and study query clauses, we introduce the notions of \emph{surface literal}, \emph{chained variables} and \emph{isolated variables}.

\begin{defi}
Let $Q$ be a query clause. Then, a literal $L$ is a \emph{surface literal} in $Q$ if there exists no distinct literal $L^\prime$ in $Q$ such that $\Var(L) \subset \Var(L^\prime)$. Let $L_1$ and~$L_2$ be two surface literals in $Q$ such that $\Var(L_1) \not = \Var(L_2)$. Then, $x$ is a \emph{chained variable} in~$Q$ if $x$ belongs to $\Var(L_1) \cap \Var(L_2)$. The other non-chained variables are the \emph{isolated variables} in $Q$.	
\end{defi}

For example, the literals $\lnot A_1(x_1,x_2), \lnot A_2(x_2,x_3), \lnot A_3(x_3, x_4, x_5), \lnot A_4(x_5, x_6)$ in
\begin{align*} 
Q_1 = \ & \lnot A_1(x_1,x_2) \lor \lnot A_2(x_2,x_3) \lor \lnot A_3(x_3, x_4, x_5) \lor \lnot A_4(x_5, x_6) \lor \lnot A_5(x_3, x_4),
\end{align*}
are surface literals, but the literal $\lnot A_5(x_3, x_4)$ is not as $\Var(A_5) \subset \Var(A_3)$. Then, the variables~$x_2, x_3, x_5$ are the chained variables and $x_1, x_4, x_6$ are the isolated variables in~$Q_1$. In
\begin{align*}
Q_2 = \ & \lnot A_1(x_1,x_2,x_3) \lor \lnot A_2(x_3,x_4,x_5) \lor \lnot A_3(x_5, x_6, x_7) \lor \\
& \lnot A_4(x_1, x_7,x_8) \lor \lnot A_5(x_3, x_4, x_9),	
\end{align*}
all literals are surface literals, therefore, the variables $x_1, x_3, x_4, x_5, x_7$ are the chained variables and $x_2, x_6, x_8, x_9$ are the isolated variables in $Q_2$. 

A hypergraph is used to represent a flat clause, formally defined as follows.

\begin{defi}
Suppose $C$ is a flat clause, and $\mathcal{H}(V, E)$ is a hypergraph which consists of a set $V$ of vertices and a set $E$ of hyperedges. Then $\mathcal{H}(V, E)$ is the \emph{hypergraph associated with~$C$} if the set $V$ of vertices consists of all variables in $C$, and the set $E$ of hyperedges contains, for each literal~$L$ in $C$, the set of variables that appear in $L$.  
\end{defi}

We use rectangles and variable symbols to represent the hyperedges and the vertices of the hypergraph associated with a flat clause, respectively. Dotted-line and solid-line rectangles respectively represent positive and negative literals and negation symbols are omitted. \textbf{Figure~\ref{fig:query_clauses}} displays the hypergraphs associated with the query clauses~$Q_1$ and $Q_2$ above.

\begin{defi}
A \emph{chained-only query clause} and an \emph{isolated-only query clause} are respectively query clauses containing only chained and only isolated variables.

\end{defi}
For example, the query clause $\lnot A(x_1,x_2) \lor \lnot A_2(x_2,x_3) \lor \lnot A_3(x_3,x_1)$ is a chained-only query clause and $\lnot A_1(x_1) \lor \lnot A_2(x_1,x_2,x_3)$ is an isolated-only query clause.

\begin{figure}[t]
\center
\includegraphics[width=.7\textwidth]{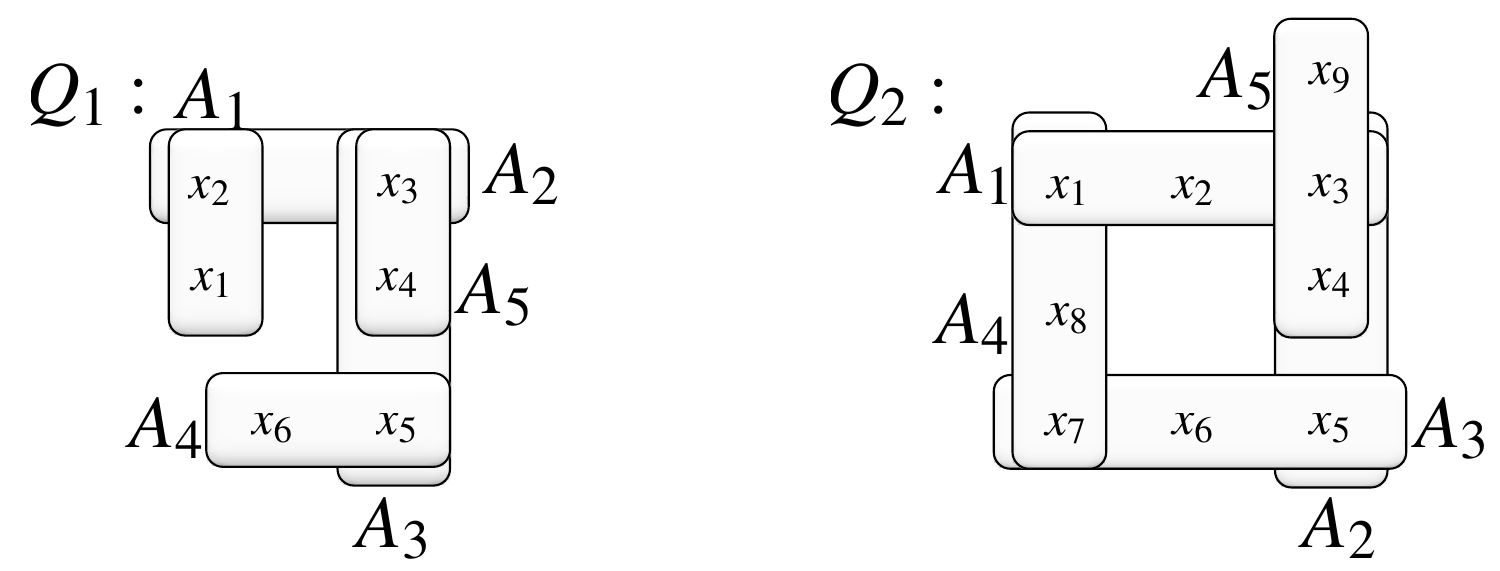}
\caption{Hypergraphs associated with of $Q_1$ and $Q_2$}
\label{fig:query_clauses}
\end{figure} 

\subsection*{\textbf{The separation rules}}
We define the separation rules we need and prove their soundness.

The separation rule \textbf{Sep} replaces a clause $C \lor D$ by two clauses in which the subclauses $C$ and $D$ have been separated by a fresh predicate symbol \cite{SH00}, formally:
\begin{mdframed}[linewidth=2pt]
\begin{displaymath}
 \prftree[l]{\textbf{Sep}: \quad}
   {N \cup \{C \lor D\}}
   {N \cup \{C \lor P(\overline x), \lnot P(\overline x) \lor D\}}
\end{displaymath}
if the following conditions are satisfied.
\begin{enumerate}
\item $C$ and $D$ are non-empty subclauses of $C \lor D$.
\item $\Var(C) \not \subseteq \Var(D)$ and $\Var(D) \not \subseteq \Var(C)$.
\item $\Var(C) \cap \Var(D) = \overline x$.
\item $P$ is a predicate symbol that does not occur in $N \cup \{C \lor D\}$.
\end{enumerate}
\end{mdframed}
The \textbf{Sep} rule is introduced in \cite{SH00} to decide satisfiability of fluted logic, and the rule is referred to as `splitting through new predicate symbols' in~\cite[section 3.5.6]{K06}.

The \textbf{Sep} rule preserves satisfiability equivalence. This proof can be found in Theorem 3 of the technical report version of \cite{SH00}. Formally:

\begin{lem}
\label{lem:sep_sound}
The \textbf{Sep} premise $N \cup \{C \lor D\}$ is satisfiable if and only if the \textbf{Sep} conclusion $N \cup \{C \lor P(\overline x), \lnot P(\overline x) \lor D\}$ is satisfiable.
\end{lem}

The following are separation rules, customised for separating decomposable and indecomposable query clauses. Recall that a clause is \emph{decomposable} if it can be partitioned into two variable-disjoint subclauses, otherwise, the clause is \emph{indecomposable}. 

\begin{mdframed}[linewidth=2pt]
\begin{displaymath}
 \prftree[l]{\textbf{SepDeQ}: \quad}
  {N \cup \{C \lor D\}}
  {N \cup \{C \lor \lnot p_1, \lnot p_2 \lor D, p_1 \lor p_2\}}
\end{displaymath}
if the following conditions are satisfied.
\begin{enumerate}
\item $C \lor D$ is a decomposable query clause.
\item $C$ and $D$ are non-empty subclauses of $C \lor D$.
\item $\Var(C) \cap \Var(D) = \emptyset$.
\item $p_1$ and $p_2$ are propositional variables that do not occur in $N \cup \{C \lor D\}$.
\end{enumerate}
\end{mdframed}
\begin{mdframed}[linewidth=2pt]
\begin{displaymath}
 \prftree[l]{\textbf{SepIndeQ}: \quad}
  {N \cup \{C \lor \lnot A(\overline x, \overline y) \lor D\}}
  {N \cup \{C \lor \lnot A(\overline x, \overline y) \lor P(\overline x), \lnot P(\overline x) \lor D\}}
\end{displaymath}
if the following conditions are satisfied.
\begin{enumerate}
\item $C \lor \lnot A(\overline x, \overline y) \lor D$ is an indecomposable query clause, and $\overline x \not = \emptyset$ and $\overline y \not = \emptyset$.
\item $\lnot A(\overline x, \overline y)$ is a surface literal and $\Var(C) \subseteq \overline x \cup \overline y$.
\item $\overline x$ are chained variables and $\overline x \subseteq \Var(D)$.
\item $\overline y$ are isolated variables and $\overline y \cap \Var(D) = \emptyset$.
\item $P$ is a predicate symbol that does not occur in $N \cup \{C \lor \lnot A(\overline x, \overline y) \lor D\}$.
\end{enumerate}
\end{mdframed}

The \textbf{SepDeQ} rule can be seen as either a form of \emph{formula renaming with positive literals} introduced in \textbf{Section \ref{sec:trans}} or a form of the \emph{splitting rule with propositional symbols}~\cite{RV01a,dN01}. Unlike \emph{splitting} \cite{W01}, the \textbf{SepDeQ} rule does not create a new branch in the derivation, thus no back-tracking is needed. Due to the introduction of the fresh predicate symbols in the \textbf{SepDeQ} conclusions, one cannot use the \emph{subsumption elimination technique} to eliminate the \textbf{SepDeQ} premise by the \textbf{SepDeQ} conclusions, whereas splitting can take the advantage of the subsumption elimination technique as no fresh predicate symbols are needed in the splitting process.

Inspired by the \textbf{Sep} rule, the \textbf{SepDeQ} and \textbf{SepIndeQ} rules are specifically developed for \emph{separating query clauses}. For example, in applications of the \textbf{SepDeQ} and \textbf{SepIndeQ} rules to query clauses, the polarity of the literals using the fresh predicate symbol is assigned in a way such that the \textbf{SepDeQ} and \textbf{SepIndeQ} conclusions are either \emph{query clauses} or \emph{guarded clauses}. The \textbf{Sep} rule is stronger than the \textbf{SepDeQ} and \textbf{SepIndeQ} rules with respect to separating query clauses. Given a query clause
\begin{align*}
Q = \ & \lnot A(z,x_1) \lor \lnot A(x_1,x_2) \lor \lnot A(x_2,x_3) \lor \lnot A(x_3,z) \lor  \\
& \lnot B(z,y_1) \lor \lnot B(y_1,y_2) \lor \lnot B(y_2,y_3) \lor \lnot B(y_3,z),
\end{align*}
the \textbf{Sep} rule separates it into
\begin{align*}
& \lnot A(z,x_1) \lor \lnot A(x_1,x_2) \lor \lnot A(x_2,x_3) \lor \lnot A(x_3,z) \lor P(z),\\
& \lnot B(z,y_1) \lor \lnot B(y_1,y_2) \lor \lnot B(y_2,y_3) \lor \lnot B(y_3,z) \lor \lnot P(z)
\end{align*}
using a fresh predicate symbol $P$. Yet neither \textbf{SepDeQ} nor \textbf{SepIndeQ} is applicable to~$Q$ as $Q$ is an \emph{indecomposable chained-only query clause}. 

Though the \textbf{Sep} rule is stronger and more general than the \textbf{SepDeQ} and \textbf{SepIndeQ} rules, our separation rules provide a \emph{clear view} of how a query clause is separated in a \emph{goal-oriented} way. Consider the \textbf{SepIndeQ} rule. Each application of the \textbf{SepIndeQ} rule removes a surface literal and the subclause it guards, viz., $C \lor \lnot A(\overline x, \overline y)$, from the premise $C \lor \lnot A(\overline x, \overline y) \lor D$. On the other hand, the application of the \textbf{Sep} rule to query clauses is complicated and difficult to analyse. Most importantly, applying the \textbf{Sep} rule to query clauses can derive conclusions that do not belong to the \textsf{LGQ} clausal class, making the conclusions difficult to handle. For example, applying the \textbf{Sep} rule to the above query clause $Q$ guarantees deriving a non-\textsf{LGQ} clause $\lnot A(z,x_1) \lor \lnot A(x_1,x_2) \lor \lnot A(x_2,x_3) \lor \lnot A(x_3,z) \lor P(z)$. 

Now we prove the soundness of the \textbf{SepIndeQ} rule by showing the connection between the rule and the \textbf{Sep} rule, formally stated as:
\begin{lem}
\label{lem:gen_sep}
Suppose $N \cup \{C \lor \lnot A(\overline x, \overline y) \lor D\}$ is a \textbf{SepIndeQ} premise. Then, applying the \textbf{Sep} rule can derive ${N \cup \{C \lor \lnot A(\overline x, \overline y) \lor P(\overline x), \lnot P(\overline x) \lor D\}}$ using a fresh predicate symbol $P$. 
\end{lem}
\begin{proof}
First, we prove that the \textbf{Sep} rule is applicable to $N \cup \{C \lor \lnot A(\overline x, \overline y) \lor D\}$. We distinguish four conditions of the \textbf{Sep} rule.

1)~We prove that both $C \lor \lnot A(\overline x, \overline y)$ and $D$ are non-empty subclauses. The case when $C \lor \lnot A(\overline x, \overline y)$ is empty makes the application of the \textbf{SepIndeQ} rule to $N \cup \{C \lor \lnot A(\overline x, \overline y) \lor D\}$ void. We prove that $D$ is not empty by contradiction. Suppose $D$ is empty. By the fact that $\Var(C) \subseteq \overline x \cup \overline y$, all variables in $C \lor \lnot A(\overline x, \overline y)$ are isolated variables, therefore the \textbf{SepIndeQ} rule is not applicable to $C \lor \lnot A(\overline x, \overline y)$. Hence, $D$ is a non-empty subclause.

2)~We prove that $\Var(C \lor \lnot A(\overline x, \overline y)) \not \subseteq \Var(D)$ and $\Var(D) \not \subseteq \Var(C \lor \lnot A(\overline x, \overline y))$. The fact that $\overline y \cap \Var(D) = \emptyset$ implies $\Var(C \lor \lnot A(\overline x, \overline y)) \not \subseteq \Var(D)$. We prove $\Var(D) \not \subseteq \Var(C \lor \lnot A(\overline x, \overline y))$ by contradiction. Suppose $\Var(D) \subseteq \Var(C \lor \lnot A(\overline x, \overline y))$. As $\Var(C) \subseteq \overline x \cup \overline y$, we also have $\Var(D) \subseteq \overline x \cup \overline y$. Then $\{\overline x \cup \overline y\} = \Var(C \lor \lnot A(\overline x, \overline y) \lor D) = \Var(\lnot A(\overline x, \overline y))$. Hence, $\lnot A(\overline x, \overline y)$ is a surface literal of $C \lor \lnot A(\overline x, \overline y) \lor D$, and therefore for any other surface literals $L$ in $C \lor \lnot A(\overline x, \overline y) \lor D$, $\Var(L) = \Var(\lnot A(\overline x, \overline y))$. Then all variables in $C \lor \lnot A(\overline x, \overline y) \lor D$ are isolated variables, which contradicts that $\overline x$ are the chained variables of $C \lor \lnot A(\overline x, \overline y) \lor D$.

3)~By the result established in 2) and the fact that the chained variables $\overline x$ occur in both subclauses $C \lor \lnot A(\overline x, \overline y)$ and~$D$, $\overline x = \Var(C \lor \lnot A(\overline x, \overline y)) \cap \Var(D)$.

4)~This is the same condition as 5.~of the \textbf{SepIndeQ} rule.

By the results established in 1)--4), applying the \textbf{Sep} rule to $N \cup \{C \lor \lnot A(\overline x, \overline y) \lor D\}$ derives either 
\begin{align*}
	& {N \cup \{C \lor \lnot A(\overline x, \overline y) \lor P(\overline x), \lnot P(\overline x) \lor D\}} \ \text{or} \ {N \cup \{C \lor \lnot A(\overline x, \overline y) \lor \lnot P(\overline x), P(\overline x) \lor D\}}.
\end{align*}
using a fresh predicate symbol $P$.
\end{proof}

The \textbf{SepDeQ} and \textbf{SepIndeQ} rules are sound, formally stated as:
\begin{lem}
\label{lem:qsep_sound}
The \textbf{SepDeQ} and \textbf{SepIndeQ} premises are satisfiable if and only if the \textbf{SepDeQ} and \textbf{SepIndeQ} conclusions are satisfiable, respectively.	
\end{lem}
\begin{proof}
It is immediate that the statement holds for the \textbf{SepDeQ} rule since the rule performs formula renaming. By Lemma \ref{lem:gen_sep}, applying the \textbf{SepIndeQ} rule or the \textbf{Sep} rule to the same premise derives the same conclusions. Hence, each application of the \textbf{SepIndeQ} rule can be seen as an application of the \textbf{Sep} rule. By Lemma \ref{lem:sep_sound}, the \textbf{SepIndeQ} rule is sound.
\end{proof}

Now we extend the \textbf{T-Res} system with the \textbf{SepDeQ} and \textbf{SepIndeQ} rules. Resolution systems in line with the framework of \cite{BG01} follow the principle that a conclusion is always smaller than the premises. To satisfy this condition, we make the fresh predicate symbols introduced in the applications of the \textbf{SepDeQ} and \textbf{SepIndeQ} rules $\succ_{lpo}$-smaller than the predicate symbols in the \textbf{SepDeQ} and \textbf{SepIndeQ} premises. With this restriction and the fact that the \textbf{SepDeQ} and \textbf{SepIndeQ} rules are \emph{replacement rules}, we regard the \textbf{SepDeQ} and \textbf{SepIndeQ} rules as the \emph{simplification rules} in the \textbf{T-Res} system. We use \textbf{T-Res}$^+$ to denote the \textbf{T-Res} system combined with the \textbf{SepDeQ} and \textbf{SepIndeQ} rules.

When infinitely many fresh predicate symbols are introduced in the saturation process of the \textbf{T-Res}$^+$ system, the system may lose refutational completeness. Hence, the main result of this section is formulated as follows.
\begin{thm}
\label{thm:tresp}
Provided that the \textbf{SepDeQ} and \textbf{SepIndeQ} rules	introduce finitely many fresh predicate symbols, the \textbf{T-Res}$^+$ system is sound and refutationally complete for first-order clausal logic.
\end{thm}
\begin{proof}
By Theorem \ref{thm:tres}, Lemma \ref{lem:qsep_sound} and the assumption that the fresh predicate symbols introduced in the applications of the \textbf{SepDeQ} and \textbf{SepIndeQ} rules are $\succ_{lpo}$-smaller than the predicate symbols in the \textbf{SepDeQ} and \textbf{SepIndeQ} premises.
\end{proof}

\subsection*{\textbf{Separating query clauses}}
In this section, we investigate application of the \textbf{SepDeQ} and \textbf{SepIndeQ} rules to query clauses. We start with the \textbf{SepDeQ} rule.

\begin{lem}
\label{lem:sep1_conlusion}
Suppose $Q$ is a decomposable query clause. Then, the \textbf{SepDeQ} rule separates $Q$ into narrower query clauses and narrower guarded clauses.
\end{lem}
\begin{proof}
By the definitions of query clauses and guarded clauses.	
\end{proof}

Next, we consider the \textbf{SepIndeQ} rule.
\begin{rem}
\label{rem:sepindq}
Suppose $Q$ is an indecomposable query clause. Then, the \textbf{SepIndeQ} rule applies to $Q$ if and only if there exists a surface literal in $Q$ containing both isolated variables and chained variables.	
\end{rem}
\begin{proof}
By the definition of the \textbf{SepIndeQ} rule.
\end{proof}
Based on the observation of Remark \ref{rem:sepindq}, we look at how the \textbf{SepIndeQ} rule is applied to indecomposable query clauses.
\begin{lem}
\label{lem:sep2_conlusion}
Suppose $Q$ is an indecomposable query clause, and $Q$ has a surface literal containing both chained variables and isolated variables. Then, \textbf{SepIndeQ} can separate $Q$ into narrower query clauses and narrower Horn guarded clauses.
\end{lem}
\begin{proof}
Suppose $C_1 = C \lor \lnot A(\overline x, \overline y) \lor D$ is an indecomposable query clause, and suppose $\lnot P(\overline x) \lor D$ and $C \lor \lnot A(\overline x, \overline y) \lor P(\overline x)$ are the \textbf{SepIndeQ} conclusions of $C_1$.

First, consider $\lnot P(\overline x) \lor D$. As $D$ is a query clause, $\lnot P(\overline x) \lor D$ is a query clause. By the facts that all variables in $\lnot P(\overline x) \lor D$ occur in $C \lor \lnot A(\overline x, \overline y) \lor D$ and $\lnot P(\overline x) \lor D$ does not contain $\overline y$, $\lnot P(\overline x) \lor D$ is narrower than $C \lor \lnot A(\overline x, \overline y) \lor D$. 

Next consider $C \lor \lnot A(\overline x, \overline y) \lor P(\overline x)$. The fact that $\Var(C) \subseteq \Var(\lnot A(\overline x, \overline y))$ implies $\Var(\lnot A(\overline x, \overline y)) = \Var(C \lor \lnot A(\overline x, \overline y) \lor P(\overline x))$. By the fact that all literals in $C \lor \lnot A(\overline x, \overline y) \lor P(\overline x)$ are flat, $\lnot A(\overline x, \overline y)$ is a guard for $C \lor \lnot A(\overline x, \overline y) \lor P(\overline x)$, therefore $C \lor \lnot A(\overline x, \overline y) \lor P(\overline x)$ is a guarded clause. Because $P(\overline x)$ is the only positive literal in $C \lor \lnot A(\overline x, \overline y) \lor P(\overline x)$, the clause is a Horn guarded clause. We prove that $C \lor \lnot A(\overline x, \overline y) \lor P(\overline x)$ is narrower than $C \lor \lnot A(\overline x, \overline y) \lor D$ by contradiction. Suppose $\Var(C \lor \lnot A(\overline x, \overline y) \lor D) \subseteq \Var(C \lor \lnot A(\overline x, \overline y) \lor P(\overline x))$. The fact that $\Var(D) \cap \overline y = \emptyset$ implies $\Var(D) \subseteq \overline x$, which contradicts that $\overline x$ are chained variables in $C \lor \lnot A(\overline x, \overline y) \lor D$. Hence, $C \lor \lnot A(\overline x, \overline y) \lor P(\overline x)$ is narrower than $C \lor \lnot A(\overline x, \overline y) \lor D$. 
\end{proof}

The \textbf{SepIndeQ} rule is devised to \emph{remove the isolated variables} from a query clause through \emph{separating i) the surface literal containing both the isolated variables and chained variables and ii) the literals guarded by this surface literal} from the query clause. By `a literal $L_1$ is guarded by a literal $L$', we mean that~$L$ acts as a guard of $L_1$, viz., the literal $L$ is a negative flat literal and $\Var(L_1) \subseteq \Var(L)$.

An isolated variable satisfies the following condition:
\begin{rem}
\label{rem:iso}
Suppose $Q$ is a query clause, and $x$ is an isolated variable in $Q$. Further suppose $L_1$ and $L_2$ are $x$-occurring surface literals in $Q$. Then, $\Var(L_1) = \Var(L_2)$.
\end{rem}
\begin{proof}
We prove the claim by contradiction. Suppose $\Var(L_1) \neq \Var(L_2)$.  The facts that $x \in \Var(L_1) \cap \Var(L_2)$ and $L_1$ and $L_2$ are surface literals imply that $x$ is a chained variable, which contradicts the assumption that $x$ is an isolated variable. 
\end{proof}

Lemmas \ref{lem:sep1_conlusion} and \ref{lem:sep2_conlusion} claim that applying the \textbf{SepDeQ} and \textbf{SepIndeQ} rules to a query clause derives new \emph{query clauses}, therefore the separation rules can be recursively applied to query clauses. We use \textbf{Q-Sep} to denote the procedure of recursively applying the \textbf{SepDeQ} and \textbf{SepIndeQ} rules to a query clause. 

Consider an application of the \textbf{Q-Sep} procedure to the query clause
\begin{align*} 
Q_1 = \ & \lnot A_1(\textcolor{red}{x_1},{x_2}) \lor \lnot A_2({x_2},{x_3}) \lor \lnot A_3({x_3}, \textcolor{red}{x_4}, {x_5}) \lor \lnot A_4({x_5}, \textcolor{red}{x_6}) \lor \lnot A_5({x_3}, \textcolor{red}{x_4}).
\end{align*}
Since $Q_1$ is indecomposable and contains surface literals where both isolated variables and chained variables occur, the \textbf{SepIndeQ} rule is applicable to the clause. All literals in $Q_1$ are the surface literals containing both isolated variables and chained variables, except $\lnot A_2(x_2,{x_3})$. To better show how the \textbf{SepIndeQ} rule separates a query clause, we colour the isolated variables \textcolor{red}{red} and the surface literal and the literals guarded by it \textcolor{blue}{blue}. 

\begin{figure}[t]
\center
\includegraphics[width=\textwidth]{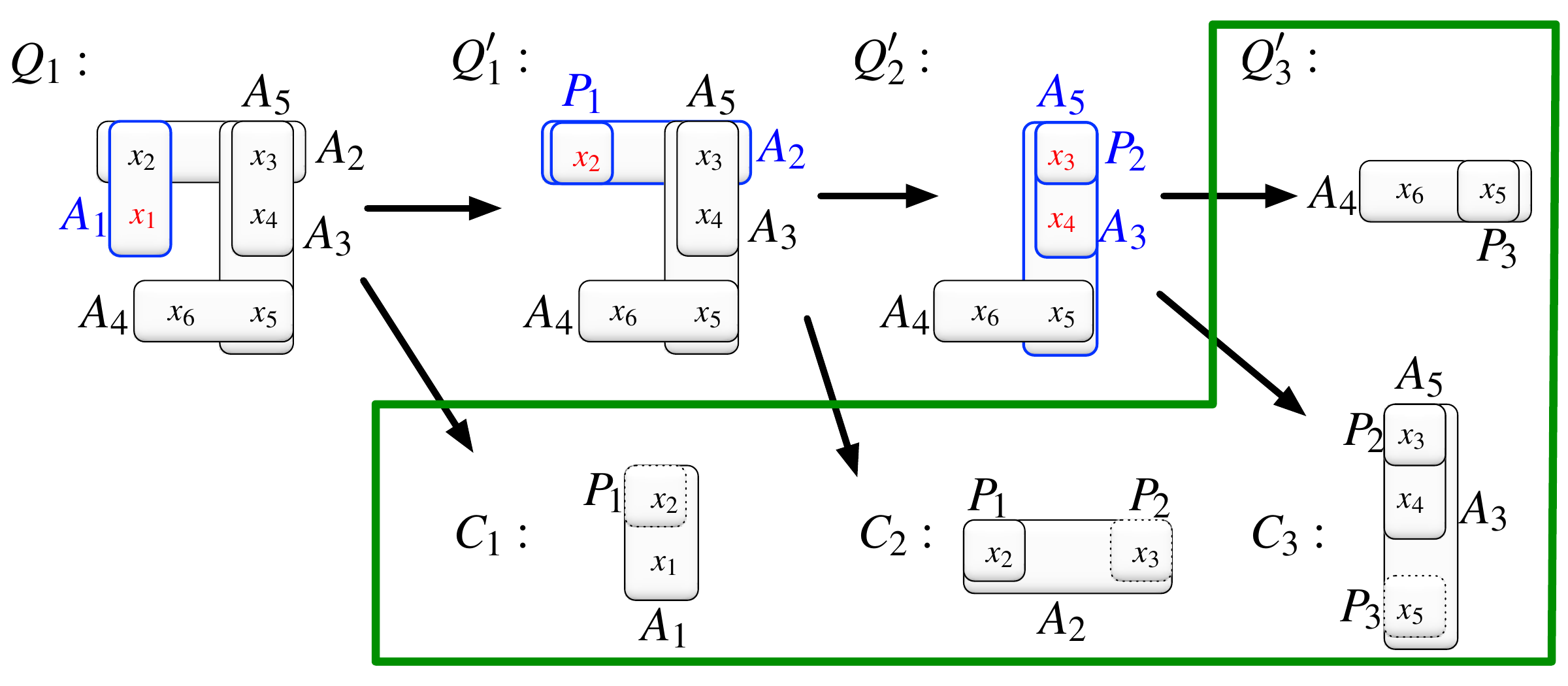}
\caption{The \textbf{Q-Sep} procedure separates $Q_1$ into Horn guarded clauses $C_1, C_2, C_3$ and an indecomposable isolated-only query clause $Q_3^\prime$. The removed isolated variables are \textcolor{red}{red} and the separated surface literal and the literals guarded by it are \textcolor{blue}{blue}.}
\label{fig:sep1}
\end{figure}

The \textbf{Q-Sep} procedure separates $Q_1$ by the following steps:
\begin{enumerate}
	\item W.l.o.g.~we begin with removing the isolated variable~$x_1$ from~$Q_1$. This means we separate the surface literal $\lnot A_1({x_1},{x_2})$ from $Q_1$. Using a fresh predicate symbol~$P_1$, applying the \textbf{SepIndeQ} rule to $Q_1$ derives:
	\begin{align*}
	&C_1 = \textcolor{blue}{\lnot A_1(x_1, x_2)} \lor P_1(x_2) \ \text{and} \\
	&Q_1^\prime = \lnot P_1(\textcolor{red}{x_2}) \lor \lnot A_2(\textcolor{red}{x_2},{x_3}) \lor \lnot A_3({x_3}, \textcolor{red}{x_4}, {x_5}) \lor \lnot A_4({x_5}, \textcolor{red}{x_6}) \lor \lnot A_5({x_3}, \textcolor{red}{x_4}).
	\end{align*}
	\item As $C_1$ is a \emph{guarded clause}, it is not separable. In $Q_1^\prime$ the surface literal $\lnot A_2({x_2},{x_3})$ guards the literal $\lnot P_1({x_2})$. To remove the isolated variable~$x_2$ from~$Q_1^\prime$, we use the \textbf{SepIndeQ} rule to separate $\lnot P_1({x_2}) \lor \lnot A_2({x_2},{x_3})$ from~$Q_1^\prime$. Using a fresh predicate symbol $P_2$, $Q_1^\prime$ is separated into: 
	\begin{align*}
	&C_2 = \textcolor{blue}{\lnot P_1(x_2) \lor \lnot A_2(x_2, x_3)} \lor P_2(x_3) \ \text{and} \\
	&Q_2^\prime = \lnot P_2(\textcolor{red}{x_3}) \lor \lnot A_3(\textcolor{red}{x_3}, \textcolor{red}{x_4}, x_5) \lor \lnot A_4(x_5, \textcolor{red}{x_6}) \lor \lnot A_5(\textcolor{red}{x_3}, \textcolor{red}{x_4}).
	\end{align*}
	\item No separation rule is applicable to $C_2$. We separate the isolated variable $x_3$ from~$Q_2^\prime$: find that $\lnot A_3({x_3}, {x_4}, x_5)$ is the $x_3$-occurring surface literal in $Q_2^\prime$, and then separate this literal and the literals guarded by it, viz., $\lnot P_2({x_3})$ and $\lnot A_5({x_3}, {x_4})$. Using a fresh predicate symbol $P_3$, $Q_2^\prime$ is separated into:  
	\begin{align*}
	&C_3 = \textcolor{blue}{\lnot P_2({x_3}) \lor \lnot A_3({x_3}, {x_4}, x_5) \lor \lnot A_5({x_3}, {x_4})} \lor P_3(x_5) \ \text{and} \\
	&Q_3^\prime = \lnot A_4(\textcolor{red}{x_5}, \textcolor{red}{x_6}) \lor \lnot P_3(\textcolor{red}{x_5}).
	\end{align*}
	\item The conclusions $C_3$ and $Q_3^\prime$ are not separable. Finally, $Q_1$ is replaced by the Horn guarded clauses $C_1, C_2, C_3$ and the indecomposable isolated-only query clause~$Q_3^\prime$. 
\end{enumerate}
Though Step 3.~aims to remove the isolated variable $x_3$ from $Q_2^\prime$, it turns out that both the isolated variables $x_3$ and $x_4$ are removed from $Q_2^\prime$. This is because $x_4$ occurs in the $x_3$-occurring surface literal $\lnot A_3({x_3}, {x_4}, x_5)$, therefore by Remark \ref{rem:sepindq}, Step 3.~also removes all $x_4$-occurring literals from $Q_2^\prime$. \textbf{Figure~\ref{fig:sep1}} shows how the \textbf{Q-Sep} procedure separates $Q_1$ into $C_1, C_2, C_3$ and~$Q_3^\prime$, framed in the green box.


The indecomposable isolated-only query clauses, for example, $Q_3^\prime$ from the previous example, are indeed Horn guarded clauses. Analysis of these two clausal classes reveals the following property:

\begin{lem}
\label{lem:iq}
An indecomposable isolated-only query clause is a Horn guarded clause.
\end{lem}
\begin{proof}
Suppose $Q$ is an indecomposable isolated-only query clause. Recall that if $Q$ contains two surface literals $L_1$ and~$L_2$ such that $\Var(L_1) \not = \Var(L_2)$ and $x \in \Var(L_1) \cap \Var(L_2)$, then $x$ is a chained variable in $Q$. Since $Q$ contains no chained variables, it is the case that either i) $Q$ contains only one surface literal, or ii) $Q$ contains multiple surface literals and each pair $L_1$ and $L_2$ of surface literals satisfies either $\Var(L_1) = \Var(L_2)$ or $\Var(L_1) \cap \Var(L_2) = \emptyset$. We distinguish these two cases:

i) The indecomposable isolated-only query clause $Q$ is flat, negative and contains only one surface literal $L$. By the definition of surface literals, $\Var(L) = \Var(Q)$. Then,~$Q$ is a Horn guarded clause with a guard $L$. 

ii) If any pair $L_1$ and $L_2$ of surface literals in $Q$ satisfies $\Var(L_1) = \Var(L_2)$, then it is the same case as i) but $L_1$ and $L_2$ are both guards of $Q$. If there exists a pair $L_1$ and~$L_2$ of surface literals satisfying $\Var(L_1) \cap \Var(L_2) = \emptyset$, then $Q$ is decomposable, which contradicts the assumption. 
\end{proof}

A chained variable in the \textbf{SepIndeQ} premise may become an isolated variable in the \textbf{SepIndeQ} conclusion, but not vice-versa. For example, in Step 1.~of the previous example, the chained variable $x_2$ in $Q_1$ becomes isolated in $Q_1^\prime$, due to the removal of the isolated variable $x_1$ in $Q_1$. However, since the \textbf{SepIndeQ} rule does not introduce new connections between variables in the conclusions, an isolated variable in the \textbf{SepIndeQ} premise cannot turn into a chained variable in the \textbf{SepIndeQ} conclusion. Since the \textbf{Q-Sep} procedure continuously removes isolated variables in the \textbf{SepIndeQ} conclusions, the procedure handles the freshly converted isolated variables.

Next, we look at another query clause 
\begin{align*}
Q_2 = \ & \lnot A_1(x_1,\textcolor{red}{x_2},x_3) \lor \lnot A_2(x_3,\textcolor{red}{x_4},x_5) \lor \lnot A_3(x_5, \textcolor{red}{x_6}, x_7) \lor \\
& \lnot A_4(x_1, x_7,\textcolor{red}{x_8}) \lor \lnot A_5(x_3, \textcolor{red}{x_4}, \textcolor{red}{x_9}).	
\end{align*}
To remove the isolated variables $x_2, x_4, x_6, x_8$ and $x_9$ from $Q_2$, we apply the \textbf{SepIndeQ} rule to $Q_2$ five times. Using fresh predicate symbols $P_4$, $P_5$, $P_6$, $P_7$ and $P_8$, the \textbf{Q-Sep} procedure separates $Q_2$ into Horn guarded clauses 
\begin{align*}
&\lnot A_1(x_1, x_2, x_3) \lor P_4(x_1, x_3), \quad \quad \lnot A_4(x_1, x_7, x_8) \lor P_5(x_1, x_7), \\
&\lnot A_3(x_5 ,x_6 ,x_7) \lor P_6(x_5, x_7), \quad \quad \lnot A_5(x_3 ,x_4 ,x_9) \lor P_7(x_3, x_4), \\
&\lnot A_2(x_3, x_4, x_5) \lor \lnot P_7(x_3, x_4) \lor P_8(x_3, x_5),
\end{align*}
and an indecomposable chained-only query clause 
\begin{align*}
Q_3 = \lnot P_4(x_1, x_3) \lor \lnot P_8(x_3, x_5) \lor \lnot P_6(x_5, x_7) \lor \lnot P_5(x_1, x_7).	
\end{align*}
\textbf{Figure~\ref{fig:sep2}} shows how the \textbf{Q-Sep} procedure separates $Q_2$ into the above Horn guarded clauses and the above indecomposable chained-only query clause. We see that each application of the \textbf{SepIndeQ} rule separates a coloured surface literal. 

Unlike the \textbf{Q-Sep} conclusions of $Q_1$, applying the \textbf{Q-Sep} procedure to $Q_2$ derives the \emph{indecomposable chained-only query clause}, c.f.~$Q_3$. By Remark \ref{rem:sepindq}, the procedure of recursively applying the \textbf{SepIndeQ} rule to an indecomposable query clause terminates if either an \emph{indecomposable chained-only query clause} or an \emph{indecomposable isolated-only query clause} is derived. We use the notion of \textsf{ICQ} to denote indecomposable chained-only query clauses.

\begin{figure}[t]
\center
\includegraphics[width=.85\textwidth]{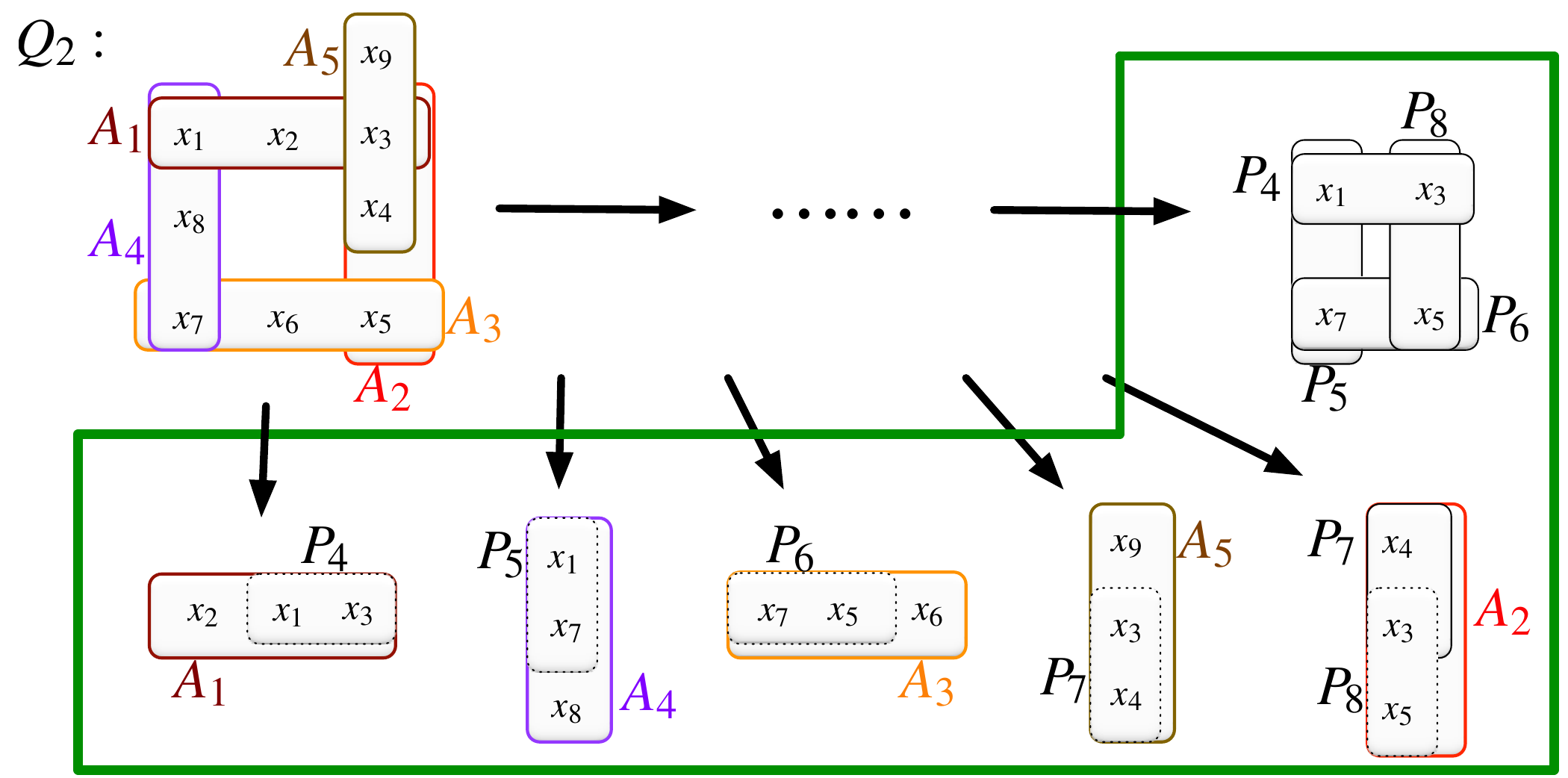}
\caption{The \textbf{Q-Sep} procedure separates $Q_2$ into Horn guarded clauses and an indecomposable chained-only query clause.}
\label{fig:sep2}
\end{figure}

%

The main result of this section is given as follows.

\begin{lem}
\label{lem:sep_conclusion}
Applying the \textbf{Q-Sep} procedure to a query clause replaces it with narrower guarded clauses and optionally narrower \textsf{ICQ} clauses.
\end{lem}
\begin{proof}
i) By Lemma \ref{lem:sep1_conlusion}, recursively applying the \textbf{SepDeQ} rule to a decomposable query clause replaces it with narrower guarded clauses and narrower indecomposable query clauses. ii) By Remark \ref{rem:sepindq} and Lemmas \ref{lem:sep2_conlusion} and \ref{lem:iq}, recursively applying the \textbf{SepIndeQ} rule to an indecomposable query clause, in which a surface literal contains both isolated variables and chained variables, replaces it by narrower Horn guarded clauses and narrower \textsf{ICQ} clauses. iii) Suppose $Q$ an indecomposable query clause that the \textbf{SepIndeQ} rule cannot separate. By Remark \ref{rem:sepindq}, $Q$ is an indecomposable query clause containing either only chained variables or only isolated variables. Then~$Q$ is either an indecomposable chained-only query clause, viz., an \textsf{ICQ} clause, or an indecomposable isolated-only query clause, viz., a Horn guarded clause, thanks to Lemma \ref{lem:iq}. By i)--iii), the claim holds. 
\end{proof}

Following Lemma \ref{lem:sep_conclusion}, we analyse the number of fresh predicate symbols that may be introduced in an application of the \textbf{Q-Sep} procedure to a query clause. 
\begin{lem}
\label{lem:sep_definer}
In the application of the \textbf{Q-Sep} procedure to a query clause, finitely many fresh predicate symbols are introduced.
\end{lem}
\begin{proof}
Suppose $Q$ is a query clause and $n$ is the width, viz., the number of distinct variables, in $Q$. By Lemma \ref{lem:sep1_conlusion}, recursively applying the \textbf{SepDeQ} rule to~$Q$ terminates in at most $n-1$ steps. The fact that each application of the \textbf{SepDeQ} rule to $Q$ introduces two fresh predicate symbols implies that at most $2*(n-1)$ fresh predicate symbols are needed. Similarly, by Lemma \ref{lem:sep2_conlusion}, recursively applying the \textbf{SepIndeQ} rule to $Q$ requires at most $n-1$ fresh predicate symbols. In total at most $3*(n-1)$ fresh predicate symbols are needed in separating $Q$.
\end{proof}

Depending on the surface literal one picks, applying the \textbf{Q-Sep} procedure to a query clause may derive \emph{distinct} sets of guarded clauses and \textsf{ICQ} clauses.

Regarding a query clause as a hypergraph, the \textbf{Q-Sep} procedure is a process of `cutting the branches off' the hypergraph. Interestingly, this procedure handles query clauses like the GYO-reduction in \cite{G79,YO79,GST84}. Using the notion of cyclic queries~\cite{BFMY83}, the GYO-reduction identifies cyclic conjunctive queries by recursively removing branches, viz.,~`ears' in the hypergraph of the queries. This method reduces a conjunctive query to an empty formula if the query is acyclic, otherwise, the query is cyclic. In our definition, an `ear' is the surface literal containing both isolated variables and chained variables, and it is separated from the query clause using the \textbf{Q-Sep} procedure. Hence, the \textbf{Q-Sep} procedure can be regarded as an implementation of the GYO-reduction: \emph{if a query clause can be separated into guarded clauses, then, that query clause is acyclic, otherwise it is cyclic}. However, the \textbf{Sep} rule, which is the basis of the \textbf{Q-Sep} procedure, is more general than the GYO-reduction as its applicability is for any first-order clause. The fact that an acyclic conjunctive query is expressible as a guarded formula is also reflected in \cite{FFG02,GLS03}.



\subsection*{\textbf{Handling indecomposable chained-only query clauses}}
In this section, we show how the \emph{term depth increase problem} is avoided when the \textbf{T-Res} rule is performed on \textsf{ICQ} clauses and \textsf{LG} clauses, and we devise a \emph{formula renaming} technique to manage the \textbf{T-Res} resolvents, which are not necessarily in the \textsf{LGQ} clausal class.

In an \textsf{ICQ} clause 
\begin{align*}
Q_3 = \lnot P_4(x_1, x_3) \lor \lnot P_8(x_3, x_5) \lor \lnot P_6(x_5, x_7) \lor \lnot P_5(x_1, x_7),	
\end{align*}
the chained variables $x_1, x_3, x_5$ and $x_7$ form a `cycle' through the literals $P_4$, $P_5$, $P_6$ and~$P_8$, as shown by the hypergraph representation given in the top-right corner in \textbf{Figure \ref{fig:sep2}}. The application of the \textbf{S-Res} rule can lead to nested compound terms in the resolvents. Consider a set $N$ of the \textsf{LGQ} clause $Q_3$ and the following \textsf{LG} clauses:
\begin{align*}
C_1 = \ & P_4(x, g(x,y,z_1,z_2))^\ast \lor \lnot G_1(x,y,z_1,z_2), \\
C_2 = \ & \lnot G_2(x,y,z_1,z_2) \lor  
P_8(g(x,y,z_1,z_2), x)^\ast \lor A(h(x,y,z_1,z_2)),   \\
C_3 = \  & P_6(f(x), x)^\ast \lor \lnot G_3(x)  \ \text{and} \ C_4 = P_5(f(x), x)^\ast \lor \lnot G_4(x).
\end{align*}
Applying the \textbf{S-Res} rule to $C_1, \ldots, C_4$ as the side premises and $Q_3$ as the main premise \emph{with all negative literals selected} derives the \textbf{S-Res} resolvent:
\begin{align*}
R_1 = & \ \lnot G_3(x) \lor \lnot G_4(x) \lor \lnot G_1(f(x),y,z_1,z_2) \lor \\ 
&\lnot G_2(f(x),y,z_1,z_2) \lor A(h(f(x),y,z_1,z_2)).
\end{align*}
The nested compound term in the literal $A(h(f(x),y,z_1,z_2))$ occurs in $R_1$. Applying the binary \textbf{S-Res} rule to $C_3$ and $Q_3$ \emph{with $\lnot P_6(x_5, x_7)$ selected} derives
\begin{align*}
R_2 = \lnot P_4(x_1, x_3) \lor \boxed{\lnot P_8(x_3, f(x))} \lor \lnot G_3(x) \lor \lnot P_5(x_1, x).
\end{align*}
Then applying the binary \textbf{S-Res} rule to $C_2$ and $R_2$ \emph{with $\lnot P_8(x_3, f(x))$ selected} derives
\begin{align*}
R_3 = \lnot P_4(x_1, x_3) \lor \lnot G_3(x) \lor \lnot P_5(x_1, x) \lor \lnot G_2(f(x),y,z_1,z_2) \lor A(h(f(x),y,z_1,z_2)),
\end{align*}
in which, again, a nested compound-term occurs in the literal $A(h(f(x),y,z_1,z_2))$. The result is predictable since an application of the \textbf{S-Res} rule can be seen as successive applications of the binary \textbf{S-Res} rule.

Now we show how the \emph{top-variable technique} handles this term depth increase. In Algorithms \ref{algorithm:refine}--\ref{algorithm:top}, the \textbf{T-Res} rule is applied to $Q_3$ and $C_1 \ldots, C_4$ as follows.
\begin{enumerate}
\item The $\T-Res(N, Q_3)$ function first selects all negative literals in $Q_3$, and then seeks the \textbf{S-Res} side premises for $Q_3$, which are $C_1, \ldots, C_4$.
\item The \textbf{S-Res} mgu of $C_1, \ldots, C_4$ and $Q_3$ is 
\begin{align*}
\{x_1 \mapsto f(x), x_5 \mapsto f(x), x_7 \mapsto x, x_3 \mapsto g(f(x),y,z_1,z_2)\}
\end{align*}
for the variables in $Q_3$. Hence $x_3$ is the only top variable in $Q_3$.
\item The literals $\lnot P_4(x_1, x_3)$ and $\lnot P_8(x_3, x_5)$ in $Q_3$ are therefore the \emph{top-variable literals}. A \textbf{T-Res} inference is performed on $C_1$, $C_2$ and $Q_3$, deriving:
\begin{align*}
R = \ & \textcolor{red}{\lnot G_1(x,y,z_1,z_2)} \lor \textcolor{blue}{\lnot G_2(x,y,z_1,z_2) \lor A(h(x,y,z_1,z_2))}^\ast \lor \\
& \textcolor{brown}{\lnot P_6(x, x_7) \lor \lnot P_5(x, x_7)},
\end{align*}
Notice that $R$ contains no nested compound terms.
\item No further inference is possible for $N \cup \{R\}$, hence $N \cup \{R\}$ is saturated.
\end{enumerate}

Though the \textbf{T-Res} resolvent $R$ is free of nested compound terms, it is wider than any of its premises; moreover, it is neither a query clause due to the occurrence of the compound term $h(x,y,z_1,z_2)$ nor an \textsf{LG} clause since $R$ contains no loose guard. The resolvent $R$ is formed with the \emph{remainders} of $C_1$, $C_2$ and $Q_3$ coloured in \textcolor{red}{red}, \textcolor{blue}{blue} and \textcolor{brown}{brown} above, respectively. Observe that: i) the remainders of $C_1$ and $C_2$ are \emph{\textsf{LG} clauses} and the remainder of $Q_3$ is a \emph{query clause}, and ii) due to the covering property of \textsf{LG} clauses, after unification, the remainders of $C_1$ and $C_2$ form an \textsf{LG} clause in $R$. Based on this observation, we devise a formula renaming technique which introduces a fresh predicate symbol~$P_9$ to abstract the remainders of $C_1$ and $C_2$ from $R$ and replaces $R$ by its equisatisfiable set of \textsf{LGQ} clauses:
\begin{align*}
C_5 = \ & \textcolor{red}{\lnot G_1(x,y,z_1,z_2)} \lor \textcolor{blue}{\lnot G_2(x,y,z_1,z_2) \lor  A(h(x,y,z_1,z_2))} \lor P_9(x,y,z_1,z_2), \\
Q_4 = \ & \textcolor{brown}{\lnot P_7(x, x_7) \lor \lnot P_6(x, x_7)} \lor \lnot P_9(x,y,z_1,z_2)
\end{align*} 
where $C_5$ is an \textsf{LG} clause and $Q_4$ is an indecomposable query clause. Since the \textbf{SepIndeQ} rule is applicable to $Q_4$, one can remove the isolated variable $x_7$ from $Q_4$ via separating the literals $\lnot P_7(x, x_7)$ and $\lnot P_6(x, x_7)$ from~$Q_4$. Using a new predicate symbol~$P_{10}$, one separates $Q_4$ into the Horn guarded clauses:
\begin{align*}
& C_6 = \lnot P_7(x, x_7) \lor \lnot P_6(x, x_7) \lor \lnot P_{10}(x) \ \text{and} \ C_7 = \lnot P_9(x,y,z_1,z_2) \lor P_{10}(x).
\end{align*} 
\textbf{Figure~\ref{fig:sep3}} shows how the \textbf{Q-Sep} procedure separates~$Q_4$ into $C_6$ and $C_7$. Then, the \textbf{T-Res} resolvent $R$ is replaced by the \textsf{LG} clauses $C_5, C_6$ and~$C_7$. To sum up, i) given an \textsf{LGQ} clausal set $\{Q_3, C_1, \ldots, C_4\}$, a saturated \textsf{LGQ} clausal set $\{Q_3, C_1, \ldots, C_7\}$ is derived, and ii) the newly derived clauses $C_5, C_6$ and $C_7$ are no wider than the \textbf{T-Res} side premises $C_1$ and $C_2$. 

\begin{figure}[t]
\center
\includegraphics[width=.65\textwidth]{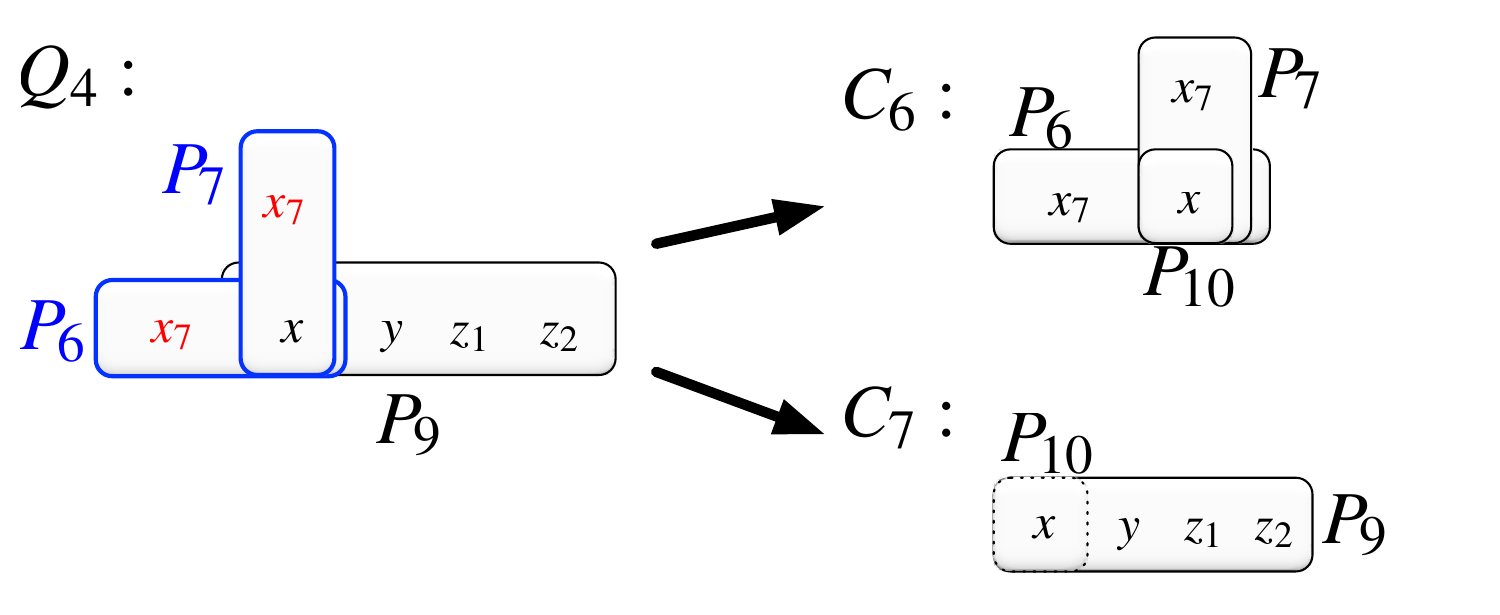}
\caption{Applying the \textbf{Q-Sep} procedure to $Q_4$ separates it into Horn guarded clauses. The removed isolated variables are coloured in \textcolor{red}{red} and the separated literals are coloured in \textcolor{blue}{blue}.}
\label{fig:sep3}
\end{figure}  



The other challenge in applying the \textbf{T-Res} rule to an \textsf{ICQ} clause and \textsf{LG} clauses is that the \textbf{T-Res} resolvents may have a \emph{wider variable cycle} than the \textbf{T-Res} main premise. For example, applying the \textbf{T-Res} rule to the \textsf{LG} clauses
\begin{align*}
C_1^\prime = & \lnot A_1(x_1, x_2)	\lor \lnot A_1(x_2, x_3) \lor \lnot A_1(x_3, x_1) \lor P_4(x_1, x_3), \\
C_2^\prime = & \lnot A_1(x_3, x_4)	\lor \lnot A_1(x_4, x_5) \lor \lnot A_1(x_5, x_3) \lor P_8(x_3, x_5), \\
C_3^\prime = & \lnot A_1(x_5, x_6)	\lor \lnot A_1(x_6, x_7) \lor \lnot A_1(x_7, x_5) \lor P_6(x_5, x_7), \\
C_4^\prime = & \lnot A_1(x_1, x_4)	\lor \lnot A_1(x_4, x_7) \lor \lnot A_1(x_7, x_1) \lor P_5(x_1, x_7)
\end{align*}
as the side premises and $Q_3$ as the main premise derives the \textsf{ICQ} clause 
\begin{align*}
& \lnot A_1(x_1, x_2)	\lor \lnot A_1(x_2, x_3) \lor \lnot A_1(x_3, x_1) \lor \lnot A_1(x_3, x_4)	\lor \lnot A_1(x_4, x_5) \lor \lnot A_1(x_5, x_3) \lor \\
& \lnot A_1(x_5, x_6)	\lor \lnot A_1(x_6, x_7) \lor \lnot A_1(x_7, x_5) \lor \lnot A_1(x_1, x_4) \lor \lnot A_1(x_4, x_7) \lor \lnot A_1(x_7, x_1)
\end{align*}
in which the variable cycle is significantly wider than the one in the query clause~$Q_3$. However, the \textbf{T-Res} system avoids this \textbf{T-Res} inference by selecting all negative literals in $C_1^\prime, C_2^\prime, C_3^\prime$ and $C_4^\prime$, forcing these clauses to act as the main premises in the resolution inferences. Specifically, the \textbf{T-Res} system restricts that only \emph{ground simple clauses} and \emph{non-ground compound-term clauses} can be side premises for \textsf{ICQ} clauses. Without introducing wider variable cycles, the application of the \textbf{T-Res} rule to~$Q_3$ and $C_1, \ldots, C_4$ breaks the variable cycle in $Q_3$. This is due to the covering property of the \textsf{LG} clauses in the \textbf{T-Res} side premises, ensuring that the variables in the side premises are simultaneously unified, therefore the new variable relations in the remainders of the side premises, occurring in the \textbf{T-Res} resolvent, remain controlled by the loose guards of the \textsf{LG} side premises.

Transforming the \textbf{T-Res} resolvent of an \textsf{ICQ} clause and \textsf{LG} clauses to the smallest number of \textsf{LGQ} clauses is not straightforward. We use the notions of \emph{connected top variables} and \emph{closed top-variable subclauses} to identify the \textsf{LG} subclauses in the \textbf{T-Res} resolvents.

\begin{defi}
\label{def:close_tv}
In a \textbf{T-Res} inference on an \textsf{ICQ} clause as the main premise with the top-variable subclause $C$, and \textsf{LG} clauses as the side premises, 
\begin{enumerate}
\item top variables $x_i$ and $x_j$ are \emph{connected} in $C$ if there exists a sequence of top variables $x_i, \ldots, x_j$ in $C$ such that each pair of adjacent variables co-occurs in a top-variable literal, and
\item the clause $C^\prime$ is a \emph{closed top-variable subclause} of $C$ if
\begin{enumerate}
\item each pair of top variables in $C^\prime$ are connected, and
\item the top variables in $C^\prime$ do not connect to the top variables that are in $C$ but not in $C^\prime$.
\end{enumerate}
\end{enumerate}
\end{defi}

Suppose $Q_{icq}$ is an \textsf{ICQ} clause and $N_{lg}$ are \textsf{LG} clauses. Further, suppose $Q_{icq}$ is the main premise and $N_{lg}$ are the side premises in a \textbf{T-Res} inference. Then, each \emph{closed top-variable subclause} in $Q_{icq}$ is resolved with a \emph{subset $N_{lg}^\prime$ of~$N_{lg}$}, and the disjunction of the remainders of all clauses in $N_{lg}^\prime$ forms an \textsf{LG} clause in the \textbf{T-Res} resolvent. In the previous example, the top-variable subclause $\lnot P_5(x_1, x_3) \lor \lnot P_9(x_3, x_5)$ in $Q_3$ is the only \emph{closed top-variable subclause} in~$Q_3$, since $x_3$ is the only top variable in~$Q_3$. The fact that the \textbf{T-Res} side premises of $\lnot P_5(x_1, x_3)$ and $\lnot P_9(x_3, x_5)$ are $C_1$ and $C_2$ implies that the disjunction of remainders of $C_1$ and $C_2$ forms an \textsf{LG} clause 
\begin{align*}
C_{lg}^\prime = \lnot G_1(x,y,z_1,z_2) \lor \lnot G_2(x,y,z_1,z_2) \lor A(h(x,y,z_1,z_2))
\end{align*}
in the \textbf{T-Res} resolvent
\begin{align*}
R = \lnot G_1(x,y,z_1,z_2) \lor \lnot G_2(x,y,z_1,z_2) \lor A(h(x,y,z_1,z_2))^\ast \lor \lnot P_7(x, x_7) \lor \lnot P_6(x, x_7).
\end{align*}
In the previous example, we abstracted~$C_{lg}^\prime$ from $R$ by introducing a fresh predicate symbol $P_9$, obtaining an \textsf{LG} clause $C_5$ and a query clause $Q_4$.

The \textbf{T-Res} resolvents of an \textsf{ICQ} clause and \textsf{LG} clauses is handled by the following formula renaming:
\begin{mdframed}[linewidth=2pt]
Given an \textsf{ICQ} clause $Q = \lnot A_1 \lor \ldots \lor \lnot A_m \lor \ldots \lor \lnot A_n$ and \textsf{LG} clauses $C_1 = B_1 \lor D_1, \ldots, C_n = B_n \lor D_n$, applying the \textbf{T-Res} rule to $Q$ as the main premise and $C_1, \ldots, C_n$ as the side premises derives the \textbf{T-Res} resolvent
\begin{align*}
R = (\lnot A_{m+1} \lor \ldots \lor \lnot A_n)\sigma \lor D_1\sigma \lor \ldots \lor D_m\sigma	
\end{align*}
where $\sigma = \mgu(A_1 \doteq B_1, \ldots, A_m \doteq B_m)$ and the \emph{top-variable subclause} is $\lnot A_1 \lor \ldots \lor \lnot A_m$ in $Q$ where $1 \leq m \leq n$.

Suppose $\lnot A_1 \lor \ldots \lor \lnot A_m$ is partitioned into the \emph{closed top-variable subclauses} $C_1^\prime, \ldots, C_t^\prime$. Then, we can represent $R$ as 
\begin{align*}
R = (\lnot A_{m+1} \lor \ldots \lor \lnot A_n)\sigma \lor D_1^\prime(\overline {x_1}\sigma) \lor \ldots \lor  D_t^\prime(\overline {x_t}\sigma),	
\end{align*}
where $\overline {x_i}$ are the variables occurring in $D_i^\prime$ for all $i$ such that $1 \leq i \leq t$. Then, $R$ is transformed using the following rule:
\begin{mathpar}
\inferrule* [left=\normalsize {\textnormal{\textbf{T-Trans:}}} \ , right={}]{N \cup \{(\lnot A_{m+1} \lor \ldots \lor \lnot A_n)\sigma \lor D_1^\prime(\overline {x_1}\sigma) \lor \ldots \lor  D_t^\prime(\overline {x_t}\sigma)\} }
{N \cup \{P_1(\overline {x_1}\sigma) \lor D_1^\prime(\overline {x_1}\sigma), \ \ldots, \ P_t(\overline {x_t}\sigma) \lor D_t^\prime(\overline {x_t}\sigma), \\\\
    (\lnot A_{m+1} \lor \ldots \lnot A_n)\sigma \lor \lnot P_1(\overline {x_1}\sigma) \lor \ldots \lor \lnot P_t(\overline {x_t}\sigma)\}
    }
\end{mathpar}
where $P_1, \ldots, P_t$ are the fresh predicate symbols. 
\end{mdframed}
Applying the \textbf{T-Trans} rule to a \textbf{T-Res} resolvent of an \textsf{ICQ} clause and \textsf{LG} clause replaces it with a set of \textsf{LGQ} clauses and preserves satisfiability equivalence. Formally:


\begin{lem}
\label{lem:ttrans}
Let $R$ be a \textbf{T-Res} resolvent of an \textsf{ICQ} clause $Q_{icq}$ as the main premise and \textsf{LG} clauses $N_{lg}$ as the side premises. Then, the following properties hold.
\begin{enumerate}
\item Applying the \textbf{T-Trans} rule to $R$ replaces it by a set $N_{lg}^\prime$ of \textsf{LG} clauses and a query clause $Q_r$.
\item Applying the \textbf{Q-Sep} procedure to $Q_r$ separates it into a set $N_{g}$ of guarded clauses and optionally a set $N_{icq}$ of \textsf{ICQ} clauses.
\item For each clause $C^\prime$ in $N_{lg}^\prime$, there exists a clause $C$ in $N_{lg}$ such that $C^\prime$ is no wider than $C$.
\item For each clause $C^\prime$ in $N_{g}$, it is the case that either $C^\prime$ is narrower than $Q_{icq}$, or there exists a clause~$C$ in $N_{lg}$ such that $C^\prime$ is not wider than $C$.
\item For each clause $Q_{icq}^\prime$ in $N_{icq}$, $Q_{icq}^\prime$ is narrower than $Q_{icq}$.
\item Suppose $N$ is a clausal set. Then, $N \cup \{R\}$ is satisfiable if and only if $N \cup N_{lg}^\prime \cup N_{g} \cup N_{icq}$ is satisfiable.
\end{enumerate}

\end{lem}
\begin{proof}
Recall the \textbf{T-Res} rule with a-priori eligibility.
\begin{align*}
 \prftree[l]{\textbf{T-Res}: }
    {B_1 \lor D_1, \ \ldots, \ B_m \lor D_m, \ \ldots, \ B_n \lor D_n}
    {\lnot A_1 \lor \ldots \lor \lnot A_{m} \lor \ldots \lor \lnot A_n \lor D}
    {(D_1 \lor \ldots \lor D_m \lor \lnot A_{m+1} \lor \ldots \lor \lnot A_{n} \lor D)\sigma}
\end{align*}
if the following conditions are satisfied.
\begin{enumerate}
\item[1.] No literal is selected in $D_1, \ldots, D_n,D$ and $B_1, \ldots, B_n$ are strictly $\succ_{lpo}$-maximal with respect to $D_1, \ldots, D_n$, respectively. 
\item[2a.] If $n = 1$, i) either $\lnot A_1$ is selected, or nothing is selected in $\lnot A_1 \lor D$ and $\lnot A_1$ is $\succ_{lpo}$-maximal with respect to $D$, and ii) $\sigma = \mgu(A_1 \doteq B_1)$ or
\item[2b.] there must exist an mgu $\sigma^\prime$ such that $\sigma^\prime = \mgu(A_1 \doteq B_1, \ldots, A_n \doteq B_n)$, then $\lnot A_1, \ldots, \lnot A_m$ are the \emph{top-variable literals} of $\lnot A_1 \lor \ldots \lor \lnot A_{m} \lor \ldots \lor \lnot A_n \lor D$ and $\sigma = \mgu(A_1 \doteq B_1, \ldots, A_m \doteq B_m)$ where $1 \leq m \leq n$.
\item[3.] All premises are variable disjoint.
\end{enumerate} 
Suppose $Q_{icq} = \lnot A_1 \lor \ldots \lor \lnot A_{m} \lor \ldots \lor \lnot A_n$ is the \textbf{T-Res} main premise and an \textsf{ICQ} clause, and $C_1 = B_1 \lor D_1, \ldots, C_m = \ B_m \lor D_m, \ldots, C_n = B_n \lor D_n$ are the \textbf{T-Res} side premises and \textsf{LG} clauses. Further suppose $R$ is the \textbf{T-Res} resolvent~$(D_1 \lor \ldots \lor D_m \lor \lnot A_{m+1} \lor \ldots \lor \lnot A_{n})\sigma$ of $C_1, \ldots, C_n$ and $C$. The variables occurring in the \textbf{T-Trans} rule are omitted in this proof.

Suppose $C_i$ is a clause in $C_1, \ldots, C_m$. By Algorithm \ref{algorithm:refine}, $C_i$ is either a ground flat clause or a compound-term clause. Suppose $C_i$ is a ground flat clause. This means that a top variable in $Q_{icq}$ pairs a constant in $C_i$. By Lemma~\ref{lem:tres_cons}, $C_1, \ldots, C_m$ are ground flat clauses and all negative literals in $Q_{icq}$ are selected. Hence, the \textbf{T-Res} resolvent $R$ is a ground flat clause, viz., an \textsf{LG} clause, and the case of applying the \textbf{T-Trans} rule to $R$ is trivial. Hence, $C_1, \ldots, C_m$ are compound-term clauses. We now prove 1.--6.~by in sequential order.

1.-1: We first prove that $(\lnot A_{m+1} \lor \ldots \lor \lnot A_{n})\sigma$ is a query clause. By 1.~of Corollary~\ref{col:tres_uni}, the mgu $\sigma$ substitutes all variables in $\lnot A_{m+1} \lor \ldots \lor \lnot A_{n}$ with either variables or constants. Then, $(\lnot A_{m+1} \lor \ldots \lor \lnot A_{n})\sigma$ is a query clause. When $m=n$ the statement trivially holds.

1.-2: We prove that $(D_1 \lor \ldots \lor D_m)\sigma$ is a disjunction of \textsf{LG} clauses, and each disjunct maps to a closed top-variable subclause. This is done by proving: 
\begin{enumerate}
	\item[i] The subclause $D_i\sigma$ is an \textsf{LG} clause for each $i$ such that $1 \leq i \leq m$.
	\item[ii] Suppose $\lnot A_i$ and $\lnot A_j$ are two distinct literals containing connected top variables where $1 \leq i \leq m$ and $1 \leq j \leq m$. Then, $(D_i \lor D_j)\sigma$ is an \textsf{LG} clause.
	\item[iii] Suppose $\lnot A_{i_1} \lor \ldots \lor \lnot A_{i_k}$ is a closed top-variable subclause of $\lnot A_1 \lor \ldots \lor \lnot A_{m}$, and suppose $D_i^\prime$ represents $D_{i_1} \lor \ldots \lor D_{i_k}$. Then, $(D_1 \lor \ldots \lor D_m)\sigma$ can be represented as $(D_1^\prime \lor \ldots \lor D_t^\prime)\sigma$ where $1 \leq t \leq m$.
\end{enumerate}

1.-2-i: By Lemma \ref{lem:com_large} and the fact that $C_i$ is a compound-term clause, the eligible literal $B_i$ in $C_i$ is a compound-term literal. By the covering property of \textsf{LG} clauses, $\Var(B_i) = \Var(C_i)$. By 2.~of Corollary~\ref{col:tres_uni}, the mgu $\sigma$ substitutes variables in $C_i$ with variables and constants. By the fact that $C_i$ is an \textsf{LG} clause and Lemma \ref{lem:rem_lit}, $D_i\sigma$ is an \textsf{LG} clause.

1.-2-ii: Suppose $x$ and $y$ are top variables in $\lnot A_i$ and $\lnot A_{j}$, respectively. Further suppose $x$ and $y$ are connected. By the definition of connected top variables, there exists a sequence of top variables $x, \ldots, y$ in $C$ such that each pair of adjacent variables co-occurs in a top-variable literal. By Lemma \ref{lem:pres_gnd_pairing}, $x, \ldots, y$ only pair compound terms. Suppose $x^\prime$ and $y^\prime$ are two adjacent top variables in $x, \ldots, y$. W.l.o.g.~suppose $\lnot A_t$ is a top-variable literal in~$C$ where $x^\prime$ and $y^\prime$ co-occur. Suppose $B_t$ is the compound-term literal in the \textbf{T-Res} side premises that resolves~$\lnot A_t$, satisfying that $A_t\sigma \doteq B_t\sigma$. Further suppose $s^\prime$ and $t^\prime$ are the compound terms in~$B_t$ that~$x^\prime$ and~$y^\prime$ pair, respectively. By 1.~of Corollary~\ref{col:tres_uni} and the covering property of \textsf{LG} clauses, $\Var(s^\prime\sigma) = \Var(t^\prime\sigma)$, therefore $\Var(x^\prime\sigma) = \Var(y^\prime\sigma)$. Hence, $\Var(x\sigma) = \Var(y\sigma)$. By the strong compatibility of \textsf{LG} clauses,~$s^\prime\sigma$ is compatible with $t^\prime\sigma$, therefore $x^\prime\sigma$ is compatible with~$y^\prime\sigma$. Hence,~$x\sigma$~is compatible with $y\sigma$. W.l.o.g. suppose $x$ pairs a compound term $t$ in~$B_i$ and $y$ pairs a compound term $s$ in $B_j$. By the result established in 1.-2-ii, $D_i\sigma$ and~$D_j\sigma$ are \textsf{LG} clauses, The fact that $\Var(x\sigma) = \Var(y\sigma)$ implies $\Var(s\sigma) = \Var(t\sigma)$. By the covering property of \textsf{LG} clauses, $\Var(D_i\sigma) = \Var(D_j\sigma)$, therefore $D_i\sigma \lor D_j\sigma$ is covering. Since~$x\sigma$ is compatible with $y\sigma$, $s\sigma$ is compatible with $t\sigma$. By the strong compatibility property of \textsf{LG} clauses, the compound terms in $D_i\sigma$ and $D_j\sigma$ are compatible, therefore $D_i\sigma \lor D_j\sigma$ are strongly compatible. The fact that $D_i\sigma$ and $D_j\sigma$ are \textsf{LG} clauses implies that $D_i\sigma \lor D_j\sigma$ is a simple clause. Since $D_i\sigma$ is an \textsf{LG} clause,~$D_i\sigma$~contains a loose guard. By the fact that $\Var(D_i\sigma) = \Var(D_j\sigma)$, $D_i\sigma \lor D_j\sigma$ contains a loose guard. Hence, $D_i\sigma \lor D_j\sigma$ is an \textsf{LG} clause. 

1.-2-iii: Suppose $\lnot A_{i_1} \lor \ldots \lor \lnot A_{i_k}$ is a closed top-variable subclause of $\lnot A_1 \lor \ldots \lor \lnot A_{m}$. Further suppose $D_i^\prime$ represents $D_{i_1} \lor \ldots \lor D_{i_k}$ where $\sigma = \mgu(A_{i_1} \doteq B_{i_1}, \ldots, A_{i_k} \doteq B_{i_k})$. We first prove that $D_i^\prime$ is an \textsf{LG} clause. Suppose~$C^\prime$ is the top-variable subclause $\lnot A_1 \lor \ldots \lor \lnot A_{m}$. By the fact that each literal in $C^\prime$ contains at least one top variable, and 2b.~of Definition \ref{def:close_tv} that each pair of closed top-variable subclauses of $C^\prime$~has no connected top variables, one can partition $C^\prime$ into a set of closed top-variable subclauses. We use $C_1^\prime, \ldots, C_t^\prime$ to denote this set of subclauses. W.l.o.g.~we use $C_i^\prime$ to represent $\lnot A_{i_1} \lor \ldots \lor \lnot A_{i_k}$. By 2a.~of Definition \ref{def:close_tv}, each pair of top variables in~$C_i^\prime$~is connected. By the result established in 1.-2-ii, $(D_{i_1} \lor \ldots \lor D_{i_k})\sigma$ is an \textsf{LG} clause, therefore $D_i^\prime$ is an \textsf{LG} clause. We represent $(D_1 \lor \ldots \lor D_m)\sigma$ as $(D_1^\prime \lor \ldots \lor D_t^\prime)\sigma$ where each $D_i^\prime$ in $D_1^\prime, \ldots, D_t^\prime$ maps to a closed top-variable subclause~$C_i^\prime$. Now we can present the \textbf{T-Res} resolvent as follows.
\begin{align*}
R = (D_1^\prime \lor \ldots \lor D_t^\prime \lor \lnot A_{m+1} \lor \ldots \lor \lnot A_{n})\sigma
\end{align*}
Applying the \textbf{T-Trans} rule to $R$ transforms it into 
\begin{align*}
D_1^\prime\sigma \lor P_1, \ \ldots, \  D_t^\prime\sigma \lor P_t, \ Q_r =  (\lnot A_{m+1} \lor \ldots \lor \lnot A_{n})\sigma \lor \lnot P_1 \lor \ldots \lor \lnot P_t.
\end{align*}
We prove that $D_i^\prime\sigma \lor P_i$ is an \textsf{LG} clause for all $i$ such that $1 \leq i \leq t$. The case is trivial when $D_i^\prime\sigma$ is ground. Now assume that $D_i^\prime\sigma$ is non-ground. By 1.-2-iii, $D_i^\prime\sigma$ is an \textsf{LG} clause. By the definition of the \textbf{T-Trans} rule, $P_i$ is a flat literal and $\Var(D_i^\prime\sigma) = \Var(P_i)$, hence $D_i^\prime\sigma \lor P_i$ is an \textsf{LG} clause. Next, we prove that $Q_r$ is a query clause. By the definition of the \textbf{T-Trans} rule, $\lnot P_1 \lor \ldots \lor \lnot P_t$ is a negative flat clause. By the result established in 1.-1, $Q_r$ is a query clause.

2.: This is a consequence of Lemma \ref{lem:sep_conclusion}.

3.: We prove that for each clause $D_i^\prime\sigma \lor P_i$ in $D_1^\prime\sigma \lor P_1, \ \ldots, D_t^\prime\sigma \lor P_t$, there exists a \textbf{T-Res} side premise~$C$ in $C_1, \ldots, C_m$ such that $D_i^\prime\sigma \lor P_i$ is no wider than $C$. By 1.-2-i, the loose guard $\mathbb{G}\sigma$ in $D_i^\prime\sigma$ is inherited from a loose guard $\mathbb{G}$ in $C_1, \ldots, C_m$. W.l.o.g.~suppose a side premise $C$ contains the loose guard $\mathbb{G}$. The fact that a loose guard contains all variables of an \textsf{LG} clause implies that $\Var(D_i^\prime\sigma \lor P_i) = \Var(\mathbb{G}\sigma)$ and $\Var(C) = \Var(\mathbb{G})$. Then, $\Var(D_i^\prime\sigma \lor P_i) = \Var(C\sigma)$. By~2.~of Corollary~\ref{col:tres_uni}, the mgu $\sigma$ substitutes all variables in~$\mathbb{G}$ with either constants or variables, therefore~$C$ contains no less distinct variables than $D_i^\prime\sigma \lor P_i$. 

4.: Suppose $C^\prime$ is a guarded clause obtained by applying the \textbf{Q-Sep} procedure to 
\begin{align*}
Q_r = \lnot A_{m+1}\sigma \lor \ldots \lor \lnot A_{n}\sigma \lor \lnot P_1 \lor \ldots \lor \lnot P_t.	
\end{align*}
Then, $C^\prime$ can only be derived due to the fact that a surface literal in $Q_r$ is separated by the \textbf{Q-Sep} procedure. We prove that 
\begin{enumerate}
	\item[1] if the separated surface literal belongs to $\lnot A_{m+1}\sigma, \ldots, \lnot A_{n}\sigma$, then $C^\prime$ is narrower than~$Q_{icq}$, or
	\item[2] if the separated surface literal belongs to $\lnot P_1, \ldots, \lnot P_t$, then there exists a \textbf{T-Res} side premise~$C$ in $C_1, \ldots, C_m$ such that~$C^\prime$ is no wider than~$C$.
\end{enumerate}

4.-1: Suppose $C^\prime$ is a guarded clause that is obtained by separating a surface literal in~$Q_r$ belonging to $\lnot A_{m+1}\sigma,  \ldots, \lnot A_{n}\sigma$. The fact that $\lnot A_{m+1} \lor \ldots \lor \lnot A_{n}$ contains only non-top variables implies that $\lnot A_{m+1} \lor \ldots \lor \lnot A_{n}$ is narrower than $Q_{icq}$. By 1.~of Corollary~\ref{col:tres_uni}, the mgu $\sigma$ substitutes the variables in $\lnot A_{m+1} \lor \ldots \lor \lnot A_{n}$ with either variables or constants, hence $\lnot A_{m+1}\sigma \lor \ldots \lor \lnot A_{n}\sigma$ is narrower than $Q_{icq}$. By Lemma~\ref{lem:sep_conclusion}, $C^\prime$ is narrower than $\lnot A_{m+1}\sigma \lor \ldots \lor \lnot A_{n}\sigma$, hence $C^\prime$ is narrower than $Q_{icq}$.

4.-2: W.l.o.g.~suppose $\lnot P_1$ is a surface literal in $\lnot P_1, \ldots, \lnot P_t$ that is separated from~$Q_r$ and suppose $D_1^{\prime}\sigma$ is the subclause that  $P_1$ defines. Further, suppose $C^\prime$ is the guarded clause obtained by separating $\lnot P_1$ from $Q_r$. By the definition of the \textbf{T-Trans} rule, $\Var(P_1) = \Var(D_1^{\prime}\sigma)$. By 1.-2-iii, $D_1^{\prime}\sigma$ is a disjunction of the remainders from the \textbf{T-Res} side premises that map to a closed top-variable clause. W.l.o.g.~suppose~$D_1$~is one of those remainders and $D_1\sigma$ is a disjunct in $D_1^{\prime}\sigma$. Suppose $C$ is the \textbf{T-Res} side premise where $D_1$ occurs. By 2.~of Corollary~\ref{col:tres_uni}, the mgu $\sigma$ substitutes variables in the \textbf{T-Res} side premises with variables and constants, therefore $D_1\sigma$ is no wider than~$D_1$. By 1.-2-ii, $\Var(D_1\sigma) = \Var(D_1^{\prime}\sigma)$. Hence, $D_1^{\prime}\sigma$ is no wider than~$D_1$, thus~$D_1^{\prime}\sigma$ is no wider than $C$. The fact that $\Var(P_1) = \Var(D_1^{\prime}\sigma)$ implies that $P_1$ is no wider than $C$. Since the guarded clause~$C^\prime$ is obtained by separating the surface literal~$\lnot P_1$ from $Q_r$, $\lnot P_1$ acts as a guard in $C^\prime$, hence $\Var(P_1) = \Var(C^\prime)$. Then, $C^\prime$ is no wider than $C$.

5.: Suppose applying the \textbf{Q-Sep} procedure to 
\begin{align*}
Q_r = \lnot A_{m+1}\sigma \lor \ldots \lor \lnot A_{n}\sigma \lor \lnot P_1 \lor \ldots \lor \lnot P_t	
\end{align*}
derives a set $N_{icq}$ of \textsf{ICQ} clauses, and $Q_{icq}^\prime$ is an \textsf{ICQ} clause in $N_{icq}$. W.l.o.g.~we assume that the mgu $\sigma$ substitutes the variable arguments in the \textbf{T-Res} side premises $C_1, \ldots, C_m$ with the variable arguments in the \textbf{T-Res} main premise $Q_{icq}$. We prove that~$Q_{icq}^\prime$ is narrower than $Q_{icq}$ by showing that $Q_{icq}^\prime$ contains only the non-top-variables from~$Q_{icq}$. The following three steps prove this claim.

5.-1: First we prove that the chained variables (in $Q_r$) occurring in $\lnot P_1, \ldots, \lnot P_t$ belong to the non-top-variables from $Q_{icq}$. W.l.o.g.~suppose $\lnot P_1$ and $\lnot P_2$ are two surface literals in $Q_r$ that have common variables. Suppose $D_1^{\prime}\sigma$ and $D_2^{\prime}\sigma$ are the subclauses that $P_1$ and $P_2$ define, respectively. Further suppose $D_1$ is a disjunct in $D_1^{\prime}$ and $D_2$ is a disjunct in $D_2^{\prime}$. Suppose $C_1 = B_1 \lor D_1$ and $C_2 = B_2 \lor D_2$ are \textbf{T-Res} side premises. By 1.-2-ii, $\Var(D_1\sigma) = \Var(D_1^{\prime}\sigma)$ and $\Var(D_2\sigma) = \Var(D_2^{\prime}\sigma)$. By the definition of the \textbf{T-Trans} rule, $\Var(P_1) = \Var(D_1^{\prime}\sigma)$ and $\Var(P_2) = \Var(D_2^{\prime}\sigma)$, therefore $\Var(P_1) = \Var(D_1\sigma)$ and $\Var(P_2) = \Var(D_2\sigma)$. Hence, the overlapping variables between $\lnot P_1$ and $\lnot P_2$ are the same as those of $D_1\sigma$ and~$D_2\sigma$. Now we consider how the mgu $\sigma$ substitutes the variables in $D_1$ and~$D_2$. W.l.o.g~suppose $\lnot A_1$ and $\lnot A_2$ are top-variable literals in $Q_{icq}$ satisfying $A_1\sigma = B_1\sigma$ and $A_2\sigma = B_2\sigma$. To understand how the mgu $\sigma$ substitutes the variables in $D_1$ and $D_2$ is to understand how~$\sigma$, respectively, unifies the pair $A_1$ and~$B_1$ and the pair $A_2$ and~$B_2$. By 2.~in Corollary \ref{col:tres_uni} and the assumption that the mgu $\sigma$ substitutes the variable arguments in~$B_i$ with that in $A_i$, $\sigma$ substitutes all variable arguments in $B_1$ and~$B_2$ with either non-top-variables or constants from $Q_{icq}$. Hence, the overlapping variables between $B_1\sigma$ and $B_2\sigma$ are non-top-variables in $Q_{icq}$. Then, the overlapping variables between $D_1\sigma$ and $D_2\sigma$, and the ones between $P_1$ and~$P_2$, are non-top-variables from~$Q_{icq}$. By the definition of chained variables and the assumption that $\lnot P_1$ and $\lnot P_2$ are the surface literals in~$Q_r$, the overlapping variables of $P_1$ and $P_2$ are the chained variables in $Q_r$. Hence, the chained variables occurring in $\lnot P_1, \ldots, \lnot P_t$ come from the non-top variables in $Q_{icq}$.

5.-2: Next we prove that the chained variables occurring in $\lnot A_{m+1}\sigma \lor \ldots \lor \lnot A_{n}\sigma$ are the non-top-variables from $Q_{icq}$. By 1.~in Corollary \ref{col:tres_uni}, the fact that $\lnot A_{m+1} \lor \ldots \lor \lnot A_{n}$ contains only non-top-variables and the assumption that the mgu $\sigma$ substitutes the variable arguments in $C_1, \ldots, C_m$ with the variable arguments in $Q_{icq}$, the variables in $\lnot A_{m+1}\sigma \lor \ldots \lor \lnot A_{n}\sigma$ are the non-top-variables in $Q_{icq}$. Hence, the chained variables in $\lnot A_{m+1}\sigma \lor \ldots \lor \lnot A_{n}\sigma$ belong to the non-top-variables in $Q_{icq}$.

5.-3: By 5.-1 and 5.-2 and the fact that applying the \textbf{Q-Sep} procedure to a query clause does not introduce new chained variables to the query clause in the conclusions, $Q_{icq}^\prime$ contains no more distinct variables than the non-top-variables in $Q_{icq}$. Since the top variables in $Q_{icq}$ do not occur in $Q_{icq}^\prime$, $Q_{icq}^\prime$ is narrower than $Q_{icq}$.

6.: By Lemma \ref{lem:qsep_sound}, the \textbf{Q-Sep} procedure is sound. The fact that the \textbf{T-Trans} rule is formula renaming implies that the rule itself is sound. Hence, satisfiability equivalence is preserved.
\end{proof}

\begin{figure}[t]
\center
\includegraphics[width=.85\textwidth]{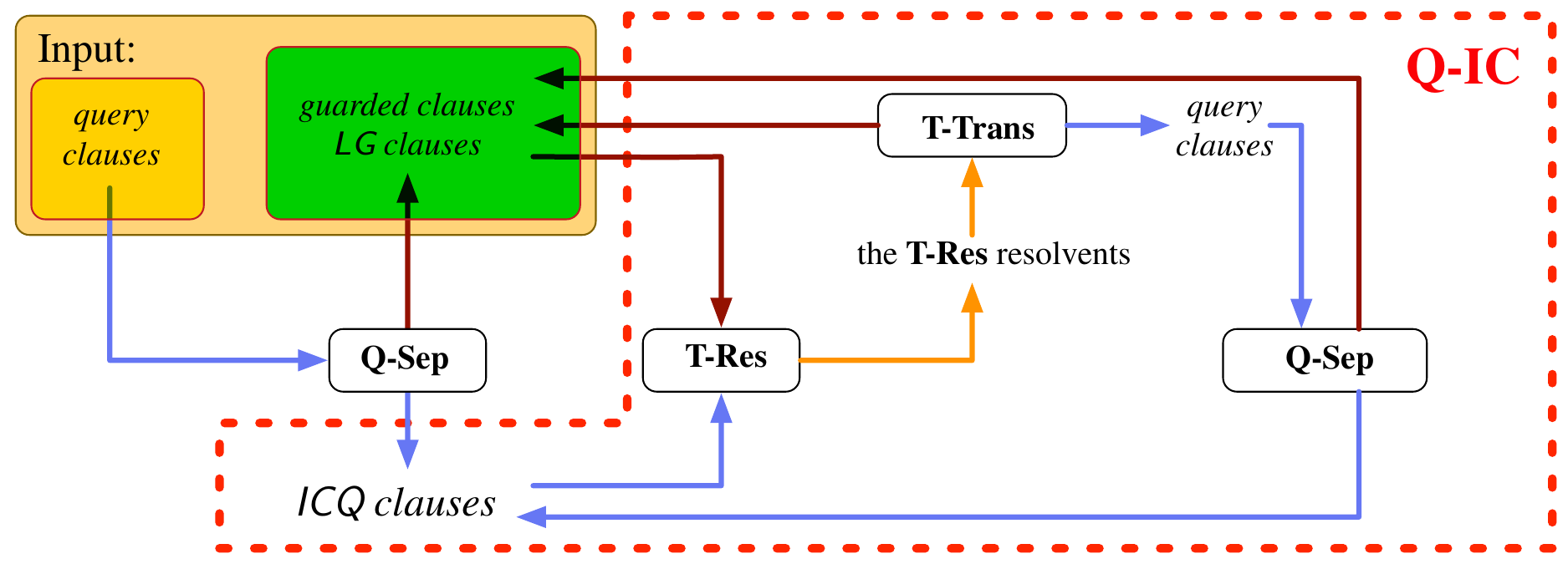}
\caption{\small{Overview of handling query clauses}}
\label{fig:q_process}
\end{figure}

We use \textbf{Q-IC} to denote the procedure of applying our rules to \textsf{ICQ} clauses. This procedure consists of the following steps: 
\begin{enumerate}
\item Apply the \textbf{T-Res} rule to an \textsf{ICQ} clause as the main premise and \textsf{LG} clauses as the side premises, deriving the \textbf{T-Res} resolvent $R$.
\item Apply the \textbf{T-Trans} rule to $R$, deriving a query clause $Q$ and \textsf{LG} clauses.
\item Apply the \textbf{Q-Sep} procedure to $Q$, deriving guarded clauses and optionally \textsf{ICQ} clauses.
\end{enumerate}
\textbf{Figure~\ref{fig:q_process}} gives an overview of the query handling process for \textsf{LG} clauses presented in this section.

The idea behind the \textbf{Q-IC} procedure is: whenever the \textbf{T-Res} resolvent $R$ of an \textsf{ICQ} clause $Q$ and \textsf{LG} clauses $C_1, \ldots, C_n$ is derived, we use the \textbf{T-Trans} rule and the \textbf{Q-Sep} procedure to replace $R$ by a set $N$ of \textsf{LGQ} clauses, which can be decided by the \textbf{T-Res}$^+$ system that we introduced above Theorem \ref{thm:tresp}. Most importantly, for each clause $C$ in~$N$, there exists a clause~$C^\prime$ in $Q, C_1, \ldots, C_n$ satisfying that $C$ is no wider than $C^\prime$. Another optional implementation for 2.--3.~of the \textbf{Q-IC} procedure is to devise a customised separation rule that separates the \textbf{T-Res} resolvent $R$ into \textsf{LGQ} clauses in one step. This implementation is feasible due to the analysis of the variable relations of $R$, as explored in Lemma~\ref{lem:ttrans}.

The main result of this section is given as follows.

\begin{lem}
\label{lem:q-co}
In the application of the \textbf{Q-IC} procedure to an \textsf{ICQ} clause $Q_{icq}$ and \textsf{LG} clauses $N_{lg}$, the \textbf{Q-IC} conclusions satisfy the following conditions.  
\begin{enumerate}
\item They are a set $N_{lg}^\prime$ of \textsf{LG} clauses and optionally a set $N_{icq}$ of \textsf{ICQ} clauses.
\item For each clause $C^\prime$ in $N_{lg}^\prime$, it is the case that either $C^\prime$ is narrower than $Q_{icq}$, or there exists a clause~$C$ in $N_{lg}$ such that $C^\prime$ is no wider than $C$.
\item For each clause $Q_{icq}^\prime$ in $N_{icq}$, $Q_{icq}^\prime$ is narrower than $Q_{icq}$.
\item The replacement of $\{Q_{icq}\} \cup N_{lg}$ by $N_{lg}^\prime \cup N_{icq}$ preserves satisfiability equivalence.
\end{enumerate}
\end{lem}
\begin{proof}
By Lemma \ref{lem:ttrans} and the fact that the guarded clauses are \textsf{LG} clausal clauses.	
\end{proof}

\section{Answering \textsf{BCQ}s for the guarded quantification fragments}
\label{sec:qans}
In \textbf{Section \ref{sec:tinf}} we introduce the top-variable inference system, in \textbf{Section \ref{sec:lgc}} we show that this system decides loosely guarded clauses, and in \textbf{Section \ref{sect:query}} we show how we handle query clauses. Now we combine the results from these sections and we are ready to describe a concrete \emph{saturation-based procedure for answering \textsf{BCQ}s for the guarded quantification fragments}. 

We use the notation \textbf{Q-Ans} to denote this procedure. To show that the \textbf{Q-Ans} procedure is suitable for implementation in modern saturation-based first-order theorem provers, we devise the procedure in accordance with the \emph{given-clause algorithm}~\cite{W01,MW97} in Algorithm \ref{algorithm:q-ans}.
\begin{algorithm}[h]
  \normalsize
  \DontPrintSemicolon
  \KwIn{A union $q$ of \textsf{BCQ}s and a set $\Sigma$ of the guarded quantification formulas}
 \KwOut{`Yes' or `No'}
 
 workedOff $\gets \emptyset$
 
 $N_{lg}, N_{q} \gets \Trans(\Sigma, q)$
 
 usable $\gets N_{lg}$
 
 \ForEach{$Q$ in $N_{q}$} {
 $N_{g}, N_{icq} \gets \Sep(Q)$
 
 usable $\gets \text{usable} \cup N_{icq} \cup N_{g}$
 }
 
 
 usable $\gets \Simplify(\text{usable}, \text{usable})$
 
 \;

\While{(\emph{usable} $\not = \emptyset$ and $\bot \not \in$ \emph{usable})}
{
given $\gets \Pick(\text{usable})$

workedOff $\gets$ workedOff $\cup$ \{given\}

\If{(\emph{given} is an \textsf{ICQ} clause)}
{
$R_{tres} \gets \T-Res(\text{given, workedOff})$

$N_{lg}, Q \gets \TTrans(R_{tres})$

$N_{g}, N_{icq} \gets \Sep(Q)$

new $\gets N_{lg} \cup N_{g} \cup N_{icq}$
}

\Else{
new $\gets \T-Res(\text{given, workedOff}) \cup \Factor(\text{given})$
}

new $\gets \Simplify(\text{new}, \text{new})$

new $\gets \Simplify(\Simplify(\text{new, workedOff}), \text{usable})$

workedOff $\gets \Simplify(\text{workedOff, new})$

usable $\gets \Simplify(\text{usable, new}) \cup \text{new}$
}
\

\lIf{\emph{usable} = $\emptyset$}{Print `No'}
\lIf{$\bot \in$ \emph{usable}}{Print `Yes'}
\caption{The \textbf{Q-Ans} algorithm for answering \textsf{BCQ}s for the guarded quantification fragments}
\label{algorithm:q-ans}
\end{algorithm}

The functions in Algorithm \ref{algorithm:q-ans} are listed below.
\begin{enumerate}
\item $\Trans(\Sigma, q)$ applies the \textbf{Trans} process to a set $\Sigma$ of guarded quantification formulas and a union $q$ of \textsf{BCQ}s, returning a set $N_{lg}$ of \textsf{LG} clauses and a set $N_{q}$ of query clauses.
\item $\Sep(Q)$ applies the \textbf{Q-Sep} procedure to a query clause $Q$, and returns a set $N_g$ of guarded clauses and optionally a set $N_{icq}$ of \textsf{ICQ} clause.
\item $\Pick(N)$ picks and then removes a clause from a clausal set $N$.
\item $\T-Res(C, N)$ eagerly applies the \textbf{T-Res} rule to a clause $C$ and clauses in $N$, and returns the \textbf{T-Res} resolvent $R_{tres}$.
\item $\TTrans(R_{tres})$ applies the \textbf{T-Trans} rule to the \textbf{T-Res} resolvents $R_{tres}$, returning a set $N_{lg}$ of \textsf{LG} clauses and a query clause $Q$.
\item $\Factor(C)$ applies the \textbf{Factor} rule (of the \textbf{T-Res} system) to a clause $C$, and returns the factor of $C$.
\item $\Simplify(N_1, N_2)$ returns all clauses from $N_1$ that are not redundant with respect to clauses in $N_2$.
\end{enumerate}

The derivation in Algorithm \ref{algorithm:q-ans} needs to guarantee \emph{fairness}. Let~$N$ be a set of clauses. Then, a derivation $N = N_0, N_1, \ldots$, with \emph{limit} $N_{\infty} = \bigcup_{j} \bigcap_{k \geq j} N_k$ is \emph{fair} if the conclusion of the non-redundant premises in $N_{\infty}$ is contained in $\bigcup_{j} N_j$. Intuitively fairness means that no inference in the derivation is delayed indefinitely. To ensure fairness, the $\Pick(N)$ function should guarantee that every clause in $N$ will eventually be picked. We refer the reader to \cite[page 36]{BG01} for a precise definition of \emph{fairness}.

As a given-clause algorithm, Algorithm \ref{algorithm:q-ans} splits input clauses into a worked-off clausal set $\mathit{workedOff}$ storing the clauses that have already been picked as \emph{given clauses}, and a clausal set $\mathit{usable}$ with clauses needed to be considered for further inferences. For each clause $C$ in $\mathit{usable}$, we remove it from $\mathit{usable}$, and then add $C$, all non-redundant conclusions for $C$ and the non-redundant clauses in $\mathit{workedOff}$ to $\mathit{usable}$. In the inference loop, \emph{reduction rules} are applied to guarantee termination.

Algorithm \ref{algorithm:q-ans} consists of the following stages.
\begin{itemize}
\item Lines 1--7 transform a union of \textsf{BCQ}s, guarded quantification formulas into a set of \textsf{LG} clauses and \textsf{ICQ} clauses.
\item Lines 9--22 saturate the class of \textsf{LG} clauses and \textsf{ICQ} clauses.
\item Lines 24--25 output the answer to the given \textsf{BCQ}s.
\end{itemize}

Lines 1--3 initialise the $\mathit{workedOff}$ and $\mathit{usable}$ clausal sets. Lines \text{4--6} transform a union of \textsf{BCQ} into a set of \textsf{ICQ} and guarded clauses, and then add these clauses to the $\mathit{usable}$ clausal set. Line 7 performs the \emph{input reduction} that removes redundancy in $\mathit{usable}$. 

The while-loop in Lines 9--22 terminates if either $\mathit{usable}$ is empty or it contains an empty clause $\bot$. Lines 10--11 pick a clause, namely $\mathit{given}$, from the $\mathit{usable}$ causal set and then add $\mathit{given}$ to the $\mathit{workdedOffs}$ causal set. Lines 12--18 derive new conclusions. Lines 12--16 say that if $\mathit{given}$ is an \textsf{ICQ} clause, then the \textbf{Q-IC} procedure is applied to this \textsf{ICQ} clause and \textsf{LG} clauses in the $\mathit{workedOff}$ clausal set, deriving a set of \textsf{ICQ} clauses and \textsf{LG} clauses. These newly derived clauses are denoted as $\mathit{new}$. As \textsf{ICQ} clauses are negative clauses, the positive factoring rule \textbf{Factor} does not apply to them. Lines 17--18 say that if $\mathit{given}$ is an \textsf{LG} clause, then the \textbf{T-Res} or the \textbf{Factor} rules are applied to that clause, deriving new \textsf{LG} clauses, denoted as $\mathit{new}$. Finally Lines 19--22 are the \emph{inter-reduction steps} that removes redundancy in the $\mathit{new}$, the $\mathit{workdedOff}$ and the $\mathit{usable}$ clausal sets. 

Lines 24--25 output the answer to the given \textsf{BCQ}. Suppose $q = q_1 \lor \ldots \lor q_n$ is a union of \textsf{BCQ}s and $\Sigma$ is a set of the guarded quantification formulas. An empty $\mathit{usable}$ clausal set implies that $\{\lnot q_1, \ldots, \lnot q_n\} \cup \Sigma$ is satisfiable. Hence, the answer to~$q$ is `No'. If the $\mathit{usable}$ clausal set contains an empty clause, then $\{\lnot q_1, \ldots, \lnot q_n\} \cup \Sigma$ is unsatisfiable. In this case, the answer to $q$ is `Yes'.

Since new predicate symbols are iteratively introduced in the derivation, one needs to ensure that only finitely many new predicate symbols are required. The introduced new predicate symbol will be reused whenever one needs to define a clause that has been defined before. This approach is formally stated as: 

%

\begin{rem}
\label{rem:reuse}
In the \textbf{Q-Ans} procedure, suppose a predicate symbol $P$ is used to define an \textsf{LGQ} clause $C$ at one step in the derivation. Then, in any further step whenever a predicate symbol is needed for defining $C$, we reuse the symbol~$P$.	
\end{rem}
We show that for the fragments we consider the \textbf{Q-Ans} procedure requires a finite number of predicate symbols. 

\begin{lem}
\label{lem:definer}
In the application of the \textbf{Q-Ans} procedure to the \textsf{BCQ} answering problem for \textsf{GF}, \textsf{LGF} and \textsf{CGF}, only finitely many predicate symbols are introduced. 	
\end{lem}
\begin{proof}
In the \textbf{Q-Ans} procedure, new predicate symbols are introduced in Line 2, Lines 4--6 and Lines 14--15 in Algorithm \ref{algorithm:q-ans}. We distinguish these cases:

Line 2: Since the \textbf{Trans} process is applied to formulas before the saturation process, this introduces finitely many new predicate symbols. 

Lines 4--6: A union of \textsf{BCQ}s is transformed into a finite number of query clauses. By Lemma \ref{lem:sep_definer}, only finitely many new predicate symbols are needed for separating the input query clauses.

Lines 14--15: This step uses new predicate symbols to transform the \textbf{T-Res} resolvents $R$ of an \textsf{ICQ} clause and \textsf{LG} clauses by a set of \textsf{LGQ} clauses. Since we reuse the introduced predicate symbols (Remark \ref{rem:reuse}), we need to prove that given an \textsf{LGQ} clausal set, the number of different \textbf{T-Res} resolvents $R$ is finitely bounded, and therefore the number of predicate symbols needed to transform the \textbf{T-Res} resolvents $R$ to \textsf{LGQ} clauses is finitely bounded.

W.l.o.g.~suppose the \textbf{T-Res} rule is applied to an \textsf{ICQ} clause $Q_{icq} = \lnot A_1 \lor \ldots \lor \lnot A_{m} \lor \ldots \lor \lnot A_n$ as the main premise and \textsf{LG} clauses $C_1 = B_1 \lor D_1, \ldots, C_m = \ B_m \lor D_m$ as the side premises, deriving the \textbf{T-Res} resolvent  
\begin{align*}
R = D_1\sigma \lor \ldots \lor D_m\sigma \lor \lnot A_{m+1}\sigma \lor \ldots \lor \lnot A_{n}\sigma,
\end{align*}
where $\sigma = \mgu(A_1 \doteq B_1, \ldots, A_m \doteq B_m)$. By 1.~of Lemma \ref{lem:ttrans}, $D_1\sigma, \ldots, D_m\sigma$ are \textsf{LG} clauses and $\lnot A_{m+1}\sigma \lor \ldots \lor \lnot A_{n}\sigma$ is a query clause. By 1.~of Corollary \ref{col:tres_uni} and the fact that the variables in $\lnot A_{m+1}\sigma \lor \ldots \lor \lnot A_{n}\sigma$ are the non-top variables from~$Q_{icq}$, $\lnot A_{m+1}\sigma \lor \ldots \lor \lnot A_{n}\sigma$ is narrower than~$Q_{icq}$. By 3.~of Lemma \ref{lem:ttrans}, the clauses in $D_1\sigma, \ldots, D_m\sigma$ are no wider than the clauses in $C_1, \ldots, C_m$. Hence the \textbf{T-Res} resolvent~$R$ is indeed a disjunction of a query clause~(narrower than the query clause in the \textbf{T-Res} main premise) and \textsf{LG} clauses (that are no wider than the \textsf{LG} clauses in the \textbf{T-Res} side premises). We use the terminology $R$-\emph{type clauses} to refer to the \textbf{T-Res} resolvents of an \textsf{ICQ} clause and \textsf{LG} clauses. 

We first prove that in the application of the \textbf{Q-Abs} procedure to \textsf{LGQ} clauses, the number of $R$-type clauses is finite. Suppose $N$ is an \textsf{LGQ} clausal set. Then, by applying the \textbf{Q-Sep} procedure to the query clauses in $N$, as shown in Lines~4--6 of Algorithm~\ref{algorithm:q-ans}, $N$ is transformed into a set of \textsf{LG} clauses and a set of \textsf{ICQ} clauses. Suppose $N_1$ and~$N_2$ are sets of \textsf{LG} and \textsf{ICQ} clauses, respectively. W.l.o.g. suppose $N = N_1 \cup N_2$. We distinguish the inferences performed on $N_1$ and $N_2$.

i: Suppose $N_1^\prime$ is the union of $N_1$ and the \textsf{LG} clauses derived by applying the \textbf{T-Res}$^+$ system to $N$. By Lemma \ref{lem:bounded_width} and the property that \textsf{LG} clauses contain no nested compound terms, $N_1^\prime$ consists of finitely many clauses. Suppose~$N_1^{\prime\prime}$ is the set of \textsf{LG} clauses (after condensation and modulo variable renaming) built using the signature of $N_1^\prime$, and no clause $N_1^{\prime\prime}$ is wider than the maximal width of the clauses in $N_1^\prime$. By the fact that the clauses in $N_1^{\prime\prime}$ are of bounded depth and width, the number of clauses in~$N_1^{\prime\prime}$ is finitely bounded. Suppose $C$ is an \textsf{LG} clause that is a subclause in the $R$-type clause when applying the \textbf{T-Res} rule to $N$. By 3.~of Lemma \ref{lem:ttrans}, $C$ is no wider than the clauses in $N_1^\prime$, therefore $C$ belongs to $N_1^{\prime\prime}$. By the fact that the number of clauses in~$N_1^{\prime\prime}$~is bounded, the number of clauses that are built using \textsf{LG} subclauses is bounded, hence, using the signature in $N_1$, there are finitely many $D_1\sigma \lor \ldots \lor D_m\sigma$ clauses. 

ii: Suppose $N_2^\prime$ is the set of query clauses (after condensation and modulo variable renaming) built using the signature of $N_2$, and the clauses in $N_2^\prime$ are narrower than the maximal width of the clauses in $N_2$. Since clauses in $N_2^\prime$ are of bounded depth and width, there are finitely many clauses in $N_2^\prime$. Suppose $Q_r$ is the query clause occurring in the $R$-type clause in applying the \textbf{T-Res} rule to $N$. Then, $Q_r$ is narrower than the clauses in $N_2^\prime$, hence $Q_r$ belongs to $N_2^\prime$. Hence, using the signature in $N_2$, there are finitely many $\lnot A_{m+1}\sigma \lor \ldots \lor \lnot A_{n}\sigma$ clauses.

By the results established in i and ii, given an \textsf{LGQ} clausal set $N$, the number of $R$-type clauses that can be derived from $N$ is finitely bounded. Then, for each $R$-type clause, only a finite number of new predicate symbols is needed. Since we reuse the introduced predicate symbols as stated in Remark \ref{rem:reuse}, the total number of new predicate symbols for transforming $R$-type clauses is finitely bounded. Then, Lines~14--15 only require a finitely bounded number of new predicate symbols.  
\end{proof}

Next, we prove that the \textbf{Q-Ans} procedure guarantees termination.
\begin{thm}
\label{thm:ter}
The \textbf{Q-Ans} procedure guarantees termination of deciding satisfiability for the \textsf{LGQ} clausal class.	
\end{thm}
\begin{proof}
By Theorem \ref{thm:lgc}, the \textbf{Q-Ans} procedure is guaranteed to terminate on the \textsf{LG} clausal class. By Lemmas \ref{lem:sep_conclusion} and \ref{lem:q-co}, applying the \textbf{Q-Ans} procedure to query clauses and \textsf{LG} clauses derives \textsf{LGQ} clauses that are no wider and no deeper than the premises. By Lemma \ref{lem:definer}, applying the \textbf{Q-Ans} procedure to \textsf{LGQ} clauses requires finitely many new predicate symbols. Therefor, the \textbf{Q-Ans} procedure decides satisfiability of the \textsf{LGQ} clausal class.
\end{proof}

Finally, the next theorem positively answers Question \ref{ques:qans}.

\begin{thm}
\label{thm:ans}
The \textbf{Q-Ans} procedure is a decision procedure for answering \textsf{BCQ}s for \textsf{GF}, \textsf{LGF} and \textsf{CGF}.
\end{thm}
\begin{proof}
By Theorems \ref{thm:trans}, \ref{thm:tresp} and \ref{thm:ter}.
\end{proof}

%

\section{Saturation-based \textsf{BCQ} rewriting for the guarded quantification fragments}
\label{sec:qrew}

In this section, we turn our attention to investigating the \emph{saturation-based \textsf{BCQ} rewriting problem for the guarded quantification fragments}.

\setcounter{ques}{1}
\begin{ques}
Suppose $\Sigma$ is a set of formulas in \textsf{GF}, \textsf{LGF} and \textsf{CGF}, $D$ is a set of ground atoms and $q$ is a union of \textsf{BCQ}s. Further, suppose $N$ is the saturation obtained by applying the procedure devised for Question 1 to $\{\lnot q\} \cup \Sigma$. Can $N$ be back-translated to a (Skolem-symbol-free) first-order formula~$\Sigma_q$ such that $\Sigma \cup D \models q$ if and only if $D \models \Sigma_q$?
\end{ques}

Unlike the previous setting of \textsf{BCQ} answering, the \textsf{BCQ} rewriting problem depends only on the rules $\Sigma$ and the query $q$. As guarded quantification formulas are free of function symbols, the function symbols in the saturation of $\{\lnot q\} \cup \Sigma$ are Skolem symbols, hence the obtained formula $\Sigma_q$ should also be function-free.



\subsection*{\textbf{Basic notions and rules for back-translation}}
\label{sec:unsko_con}
That a clausal set $N$ can be back-translated into a first-order formula if $N$ is \emph{globally consistent}, \emph{globally linear}, \emph{normal} and \emph{unique} is shown in~\cite[chapter 5]{E96}. To avoid ambiguity, we replace the word \emph{consistency} with \emph{compatibility} in this paper.

Now we formally define the above notions, starting with \emph{global compatibility}. The \emph{compatibility property} of a clause in \textbf{Section \ref{sec:trans}} is extended to that of a clausal set. Recall that two compound terms $t$ and $s$ are \emph{compatible} if the argument sequences of $t$ and~$s$ are identical. A clause $C$ is \emph{compatible} if, in $C$, compound terms that are under the same function symbol are compatible.

%

\begin{defi}[Compatibility]
\label{def:lg_consis}


A clausal set $N$ is \emph{locally compatible} if all clauses in~$N$ are \emph{compatible}. A clausal set $N$ is \emph{globally compatible} if compound terms in $N$ that are under the same function symbol are compatible.	
\end{defi}


\begin{defi}[Linearity]
\label{def:lG_rinear}
A pair of compound terms $t$ and $s$ is \emph{linear} if the set of arguments in $t$ is a subset of that in $s$ or vice-versa. A clause $C$ is \emph{linear} if in $C$, each pair of compound terms that are under different function symbols, is linear.

A clausal set $N$ is \emph{locally linear} if all clauses in $N$ are linear. A clausal set $N$ is \emph{globally linear} if each pair of compound terms in $N$ that are under different function symbols is linear.
\end{defi}

\begin{defi}[Normality]
\label{def:normal}
A clause is \emph{normal} if the compound terms in it contain only variables as arguments. A clausal set is \emph{normal} if each clause in it is normal.
\end{defi}

\begin{defi}[Uniqueness] 
\label{def:unique} 
A compound term $f(t_1, \ldots, t_n)$ is \emph{unique} if $t_1, \ldots, t_n$ are distinct variables. A clausal set $N$ is \emph{unique} if every compound term in $N$ is unique.
\end{defi}

A first-order clausal set $N$ can be back-translated into a first-order formula if $N$ satisfies all the aforementioned properties.

\begin{thm}[{\cite[\normalfont{chapter 5}]{E96}}]
\label{thm:unsko}
Suppose $N$ is a normal, unique, globally linear and globally compatible first-order clausal set. Then, $N$ can be back-translated into a first-order formula without Skolem symbols.	
\end{thm}

Next, we introduce the basic rules for back-translation. We use the notation $C(t)$ to denote that $C(t)$ is a clause and $t$ is a term that possibly occurs in $C(t)$. We use $C_n(f(\overline {x^n_m}))$ to denote that $f(\overline {x^n_m})$ is a flat compound term and $\overline {x^n_m}$ is a variable sequence $x_1, \ldots, x_m$ occurring in the clause $C_n$.

%

A term is abstracted from a clause using:
\begin{mdframed}[linewidth=2pt]
\begin{displaymath}
 \prftree[l]{\textbf{Abs}: \ \ } 
   {N \cup \{C(t)\}}
   {N \cup \{C(y) \lor t \not \approx y\}}
\end{displaymath}
if $t$ is a term and the variable $y$ does not occur in $C(t)$. 
\end{mdframed}

Variables are renamed using:
\begin{mdframed}[linewidth=2pt]
\begin{displaymath}
 \prftree[l]{\textbf{Rena}: \ \ }
   {N \cup \{C(x)\}}
   {N \cup \{C(y)\}}
\end{displaymath}
if every occurrence of the variable $x$ in $C(x)$ is replaced by the variable $y$ and $y$ does not occur in $C(x)$.
\end{mdframed}

A clausal set is unskolemised to a first-order formula using:
\begin{mdframed}[linewidth=2pt]
Suppose $N^\prime$ is a first-order clausal set
\begin{displaymath}
\Bigg\{
\begin{array}{l}
C_1(\ldots, f(x_1, \ldots, x_n), \ldots, a, \ldots, z), \\
\hfil  \ldots, \\
C_m(\ldots, g(x_1, \ldots, x_n), \ldots, b, \ldots)	
\end{array}
\Bigg\},
\end{displaymath}
where $a$ and $b$ represent the Skolem and the non-Skolem constants in $N^\prime$, respectively, $f$ and $g$ represent the Skolem function symbols in $N^\prime$, and $z$ represents the variables that are not under Skolem functions in $N^\prime$.

Let $F$ be a Skolem-symbol-free first-order formula
\begin{displaymath}
\exists y\forall x_1 \ldots x_n \exists y_1\ldots y_k \forall z
\left\lbrack
\begin{array}{l}
C_1(\ldots, y_1, \ldots, y, \ldots, z) \land \\
\hfil  \ldots \\
C_m(\ldots, y_k, \ldots, b, \ldots)	
\end{array}
\right\rbrack,
\end{displaymath}
where the variables $y, y_1, \ldots, y_k$ do not occur in $N^\prime$. Then, $N^\prime$ is unskolemised by the following rule:
\begin{displaymath}
 \prftree[l]{\textbf{Unsko}: \ \ }
   {N \cup N^\prime}
   {N \cup \{F\}}
\end{displaymath}
if $N^\prime$ is normal, unique, globally linear and globally compatible.
\end{mdframed}


%
%
The challenge of applying the \textbf{Unsko} rule to a clausal set $N$ is not only about computing a correct result, but it is about ensuring that $N$ is normal, unique, globally linear and globally compatible. Given a clausal set~$N$ that is obtained by saturating a set of clausified formulas, the \textbf{Unsko} rule restores first-order quantifications for $N$ by eliminating the Skolem symbols in $N$. We refer the reader to~\cite[chapter 5]{E96} and \cite[pages 63--69]{GSS08} for more details on unskolemisation.

 \begin{lem}[{\cite[\normalfont{section 5}]{GSS08}}]
 \label{lem:sound_unsko}
The \textbf{Abs}, the \textbf{Rena} and the \textbf{Unsko} rules preserve logical equivalence.
 \end{lem}
Next, we devise a back-translation procedure for \textsf{LGQ} clausal sets. This procedure first transforms an \textsf{LGQ} clausal set $N$ to a normal, unique, globally linear and globally compatible clausal set~$N_1$, and then unskolemises $N_1$ into a Skolem-symbol-free first-order formula. The following \textsf{LGQ} clausal set
\begin{align*}
N = 
\left\{
\begin{array}{l}
\lnot G_1(x_1,a) \lor A_1(f(x_1,a),x_1) \lor A_2(g(x_1,a),x_1), \\
\lnot G_2(x_2,x_3) \lor A_3(f(x_2,x_3),x_2) \lor A_4(g(x_2,x_3),x_2),\\
\lnot G_3(b,x_4) \lor A_5(g(b,x_4),b)\\
\lnot G_4(x_5, c, c) \lor A_6(h(c,c,x_5)) \lor A_7(h(c,c,x_5))\\
\lnot B_1(x_8, x_6) \lor \lnot B_2(x_6, x_7) \lor \lnot B_3(x_7, x_8)
\end{array}
\right\}
\end{align*}
will be used as a running example, in which $a$ and $c$ are non-Skolem constants and $b$ is a Skolem constant.


\subsubsection*{\textnormal{\textbf{Transforming \textsf{LGQ} clausal sets to normal and unique clausal sets}}}
In this section, we transform an \textsf{LGQ} clausal set into a normal, unique, locally linear and locally compatible clausal set. First, we introduce two variations of the \textbf{Abs} rule. 


Constants in compound terms are abstracted using:
\begin{mdframed}[linewidth=2pt]
\begin{displaymath}
 \prftree[l]{\textbf{ConAbs:} \ }
   {N \cup \{C(f(\ldots,a,\ldots))\}}
   {N \cup \{C(f(\ldots,x,\ldots)) \lor x \not \approx a\}}
\end{displaymath}
if the following conditions are satisfied.
\begin{enumerate}[noitemsep]
\item $C(f(\ldots,a,\ldots))$ is a compound-term clause.
\item The variable $x$ does not occur in $C(f(\ldots,a,\ldots))$.
\item All occurrences of $a$ in $C(f(\ldots,a,\ldots))$ are simultaneously replaced by~$x$.
\end{enumerate}
\end{mdframed}

Duplicate variables in compound terms are abstracted using:
\begin{mdframed}[linewidth=2pt]
\begin{displaymath}
 \prftree[l]{\textbf{VarAbs:} \ }
   {N \cup \{C(f(\ldots,x,\ldots,x,\ldots)\}}
   {N \cup \{C(f(\ldots,x,\ldots,y,\ldots) \lor y \not \approx x\}}
\end{displaymath}
if the following conditions are satisfied.
\begin{enumerate}[noitemsep]
\item $C(f(\ldots,x,\ldots,x,\ldots))$ is a compound-term clause.
\item The variable $y$ does not occur in $C(f(\ldots,x,\ldots,x,\ldots))$.
\item Let the second variable $x$ in $f(\ldots,x,\ldots,x,\ldots)$ occur at the position $i$ in $f(\ldots,x,\ldots,x,\ldots)$. Then, all occurrence of $x$ in position $i$ in all compound terms in $C(f(\ldots,x,\ldots,x,\ldots))$ are simultaneously replaced by $y$.
\end{enumerate}
\end{mdframed}


We use \textbf{Q-Abs} to denote the procedure of applying the \textbf{ConAbs} and the \textbf{VarAbs} rules to an \textsf{LGQ} clausal set. The \textbf{Q-Abs} procedure ensures that an \textsf{LGQ} clausal set is transformed into a \emph{normal} and \emph{unique} clausal set. Using the \textsf{LGQ} clausal set $N$ as an example, the \textbf{Q-Abs} procedure is applied to $N$ by the following steps.

\begin{enumerate}
\item Recursively apply the \textbf{ConAbs} rule to each clause in an \textsf{LGQ} clausal set. From $N$ we obtain
\begin{align*}
N_1 = 
\left\{
\begin{array}{l}
\lnot G_1(x_1,y_1) \lor A_1(f(x_1,y_1),x_1) \lor A_2(g(x_1,y_1),x_1) \lor y_1 \not \approx a, \\
\lnot G_2(x_2,x_3) \lor A_3(f(x_2,x_3),x_2) \lor A_4(g(x_2,x_3),x_2),\\
\lnot G_3(y_2,x_4) \lor A_5(g(y_2,x_4),y_2) \lor y_2 \not \approx b,\\
\lnot G_4(x_5, y_3, y_3) \lor A_6(h(y_3,y_3,x_5)) \lor A_7(h(y_3,y_3,x_5)) \lor y_3 \not \approx c \\
\lnot B_1(x_8, x_6) \lor \lnot B_2(x_6, x_7) \lor \lnot B_3(x_7, x_8)
\end{array}
\right\}.
\end{align*}
\item For each clause in the clausal set obtained in 1., recursively apply the \textbf{VarAbs} rule to it. From $N_1$ we obtain
\begin{align*}
N_2 = 
\left\{
\begin{array}{l}
\lnot G_1(x_1,y_1) \lor A_1(f(x_1,y_1),x_1) \lor A_2(g(x_1,y_1),x_1) \lor y_1 \not \approx a, \\
\lnot G_2(x_2,x_3) \lor A_3(f(x_2,x_3),x_2) \lor A_4(g(x_2,x_3),x_2),\\
\lnot G_3(y_2,x_4) \lor A_5(g(y_2,x_4),y_2) \lor y_2 \not \approx b,\\
\lnot G_4(x_5, y_3, y_4) \lor A_6(h(y_3,y_4,x_5)) \lor A_7(h(y_3,y_4,x_5)) \lor y_3 \not \approx c \lor y_4 \not \approx y_3\\
\lnot B_1(x_8, x_6) \lor \lnot B_2(x_6, x_7) \lor \lnot B_3(x_7, x_8)
\end{array}
\right\}.	
\end{align*}
\end{enumerate}

We use the notation \textsf{LGQ}$_\textsf{nu}$ to denote the clausal set obtained by applying the \textbf{Q-Abs} procedure to an \textsf{LGQ} clausal set. 
\begin{lem}
\label{lem:un}
Let $N$ be a set of \textsf{LGQ}$_\textsf{nu}$ clauses. Then, i) all clauses in $N$ are strongly compatible, and ii) $N$ is normal, unique, locally compatible and locally linear. 	
\end{lem}
\begin{proof}
W.l.o.g.~suppose $N_1$ is an \textsf{LGQ} clausal set satisfying such that applying \textbf{Q-Abs} procedure to $N_1$ derives $N$. Further, suppose $C$ is a clause in $N_1$.

By the strong compatible property of \textsf{LGQ} clauses and the fact that the \textbf{ConAbs} and the \textbf{VarAbs} rules simultaneously abstract variables or constants from $C$, applying the \textbf{Q-Abs} procedure to $C$ derives a strongly compatible clause. Hence, the clauses in $N$ are strongly compatible, therefore $N$ is locally compatible and locally linear.

That $C$ is simple implies that the arguments in compound terms of $C$ are either variables or constants. Suppose $C^\prime$ is the clause obtained by recursively applying the \textbf{ConAbs} rule to $C$. Since each application of the \textbf{ConAbs} rule to $C$ abstracts a constant occurring in the compound terms of $C$, no constants occur in compound terms in $C^\prime$, hence $C^\prime$ is normal. Suppose~$C^{\prime\prime}$ is the clause obtained by recursively applying the \textbf{VarAbs} rule to $C^\prime$. Since each application of the \textbf{VarAbs} rule to $C^\prime$ abstracts a duplicate variable occurring in the compound terms of~$C^\prime$, no duplicate variables occur in compound terms in $C^{\prime\prime}$, therefore $C^{\prime\prime}$ is unique. The fact that $C^\prime$ is normal implies that~$C^{\prime\prime}$ is normal. Then, $N$ is normal and unique.
\end{proof} 
Note that an \textsf{LGQ}$_\textsf{nu}$ clause may not belong to the \textsf{LGQ} clausal class due to the presence of equality literals. 


\subsubsection*{\textnormal{\textbf{Renaming \textsf{LGQ}$_\textsf{nu}$ clausal sets for unskolemisation}}}
In this section, we transform an \textsf{LGQ$_\textsf{nu}$} clausal set into a normal, unique, globally compatible and globally linear clausal set, preparing the set for unskolemisation. 

Given an \textsf{LGQ$_\textsf{nu}$} clausal set $N$, one needs to locate the \textsf{LGQ$_\textsf{nu}$} clauses in $N$ that have common Skolem function symbols, so that we can simultaneously unskolemise these clauses. We introduce the notions of \emph{connectedness}, \emph{inter-connectedness} and \emph{closed clausal set} to define clauses that have identical function symbols.

\begin{defi}[Inter-connected clausal set]
\label{def:connect}
Two clauses are \emph{connected} if they contain at least one common function symbol. Two clausal sets are \emph{connected} if they contain at least one common function symbol, otherwise, they are \emph{unconnected}.

A clausal set $N$ is an \emph{inter-connected clausal set} if for any pair of clauses $C$ and~$C^\prime$ in $N$, there exists a sequence of clauses $C, C_1, \ldots, C_n, C^\prime$ in $N$ such that each pair of adjacent clauses in $C, C_1, \ldots, C_n, C^\prime$ is connected.
\end{defi}

Recall that a flat clause is a clause containing no function symbols. We say that a clausal set is \emph{flat} if the set contains only flat clauses. We partition an \textsf{LGQ$_\textsf{nu}$} clausal set~$N$ into clausal sets $N_1, \ldots, N_n$ such that i) each $N_i$ is either an \emph{inter-connected clausal set} or a \emph{flat clausal set}, and ii) each pair of clausal sets in $N_1, \ldots, N_n$ are \emph{unconnected}. Then, $N_1, \ldots, N_n$ are \emph{closed clausal sets} in~$N$.

An inter-connected \textsf{LGQ$_\textsf{nu}$} clausal set has the following useful property.

\begin{lem}
\label{lem:interc}
Let $N$ be an inter-connected \textsf{LGQ$_\textsf{nu}$} clausal set. Then, all compound terms in $N$ have the same arity.
\end{lem}
\begin{proof}
In a clausal set, compound terms that are under the same function symbol have the same arity. By i) of Lemma \ref{lem:un}, the compound terms in an \textsf{LGQ$_\textsf{nu}$} clause have the same arity. Hence, all compound terms in an inter-connected \textsf{LGQ$_\textsf{nu}$} clausal set have the same arity.
\end{proof}

Given a closed \textsf{LGQ$_\textsf{nu}$} clausal set $N$, the \textbf{Rena} rule does not apply to it if~$N$ is a flat clausal set. Variables in an inter-connected \textsf{LGQ$_\textsf{nu}$} clausal set are renamed using the following rule:

\begin{mdframed}[linewidth=2pt]
\begin{displaymath}
 \prftree[l]{\textbf{VarRe:} \ }
   {N \cup \{C_1(f(\overline {x^1_m})), \ldots, C_n(g(\overline {x^n_m})))\}}
   {N \cup \{C_1(f(\overline {y_m})), \ldots, C_n(g(\overline {y_m})))\}}
\end{displaymath}
if the following conditions are satisfied.
\begin{enumerate}[noitemsep]
\item $\{C_1(f(\overline {x^1_m})), \ldots, C_n(g(\overline {x^n_m})))\}$ is an inter-connected \textsf{LGQ}$_\textsf{nu}$ clausal set.
\item For variable sequences $\overline {x^1_m}, \ldots, \overline {x^n_m}$ occurring in all compound terms of \\ $\{C_1(f(\overline {x^1_m})), \ldots, C_n(g(\overline {x^n_m})))\}$, each of $\overline {x^1_m}, \ldots, \overline {x^n_m}$ is renamed with $\overline {y_m}$.
\item The variable sequence $\overline {y_m}$ does not occur in $\{C_1(f(\overline {x^1_m})), \ldots, C_n(g(\overline {x^n_m}))\}$.
\end{enumerate}
\end{mdframed}

We use \textbf{Q-Rena} to denote the procedure of applying the \textbf{VarRe} rule to an inter-connected \textsf{LGQ$_\textsf{nu}$} clausal set. The \textbf{Q-Rena} procedure transforms an \textsf{LGQ}$_\textsf{nu}$ clausal set to a normal, unique, globally compatible and globally linear clausal set, detailed below.
\begin{enumerate}
\item Partition an \textsf{LGQ$_\textsf{nu}$} clausal set to closed \textsf{LGQ$_\textsf{nu}$} clausal sets. We use the \textsf{LGQ$_\textsf{nu}$} clausal set $N_2$ from the previous section as an example. Partition $N_2$ into closed \textsf{LGQ$_\textsf{nu}$} clausal sets
\begin{align*}
& N_2^\prime = 
\left\{
\begin{array}{l}
\lnot G_1(x_1,y_1) \lor A_1(f(x_1,y_1),x_1) \lor A_2(g(x_1,y_1),x_1) \lor y_1 \not \approx a, \\
\lnot G_2(x_2,x_3) \lor A_3(f(x_2,x_3),x_2) \lor A_4(g(x_2,x_3),x_2),\\
\lnot G_3(y_2,x_4) \lor A_5(g(y_2,x_4),b) \lor y_2 \not \approx b
\end{array}
\right\},
\\
& N_2^{\prime\prime} = 
\left\{
\begin{array}{l}
\lnot G_4(x_5, y_3, y_4) \lor A_6(h(y_3,y_4,x_5)) \lor A_7(h(y_3,y_4,x_5)) \lor y_3 \not \approx c \lor y_4 \not \approx y_3
\end{array}
\right\}, \\
& \text{and} \ N_2^{\prime\prime\prime} = \{\lnot B_1(x_8, x_6) \lor \lnot B_2(x_6, x_7) \lor \lnot B_3(x_7, x_8)\}.
\end{align*}
\item Apply the \textbf{VarRe} rule to the inter-connected \textsf{LGQ$_\textsf{nu}$} clausal sets obtained in 1. Using a sequence of new variables $x$ and $y$, applying the \textbf{VarRe} rule to $N_2^\prime$ derives 
\begin{align*}
N_3^\prime = 
\left\{
\begin{array}{l}
\lnot G_1(x,y) \lor A_1(f(x,y),x) \lor A_2(g(x,y),x) \lor y \not \approx a, \\
\lnot G_2(x,y) \lor A_3(f(x,y),x) \lor A_4(g(x,y),x),\\
\lnot G_3(x,y) \lor A_5(g(x,y),x) \lor x \not \approx b
\end{array}
\right\}.
\end{align*}
Using new variables $x_1, y_1, z_1$, applying the \textbf{VarRe} rule to $N_2^{\prime\prime}$ transforms it into
\begin{align*}
N_3^{\prime\prime} = 
\left\{
\begin{array}{l}
\lnot G_4(x_1, y_1, z_1) \lor A_6(h(y_1,z_1,x_1)) \lor A_7(h(y_1,z_1,x_1)) \lor y_1 \not \approx c \lor z_1 \not \approx y_1
\end{array}
\right\}.
\end{align*}
Finally, from $N_2$ we obtain the clausal set $N_3^\prime \cup N_3^{\prime\prime} \cup N_2^{\prime\prime\prime}$.
\end{enumerate}

We use the notation of \textsf{LGQ}$_\textsf{nucl}$ to denote the clausal set obtained by applying the \textbf{Q-Rena} procedure to an \textsf{LGQ}$_\textsf{nu}$ clausal set. 

\begin{lem}
\label{lem:nucl}
Let $N$ be an \textsf{LGQ}$_\textsf{nucl}$ clausal set. Then, $N$ is normal, unique, globally compatible and globally linear.
\end{lem}
\begin{proof}
Suppose $N_1$ is an inter-connected \textsf{LGQ}$_\textsf{nu}$ clausal set, and $N_2$ is a flat \textsf{LGQ}$_\textsf{nu}$ clausal set. As $N_2$ is a flat clausal set, it is trivially is normal, unique, globally compatible and globally linear. 

We prove that applying the \textbf{Q-Rena} procedure to $N_1$ transforms it to a normal, unique, globally compatible and globally linear clausal set. Suppose $N_1^\prime$ is the clausal set obtained by applying the \textbf{Q-Rena} procedure to $N_1$. By Lemma \ref{lem:un}, $N_1^\prime$ is normal and unique. By Lemma \ref{lem:interc}, the \textbf{Q-Rena} procedure renames the variables in $N_1$ so that the variable arguments in all compound terms of $N_1$ are renamed with an identical variable sequence. Then,~$N_1^\prime$ is globally compatible and globally linear. Since~$N_2$ is normal, unique, globally compatible and globally linear, $N$ is normal, unique, globally compatible and globally linear.
\end{proof}

\subsubsection*{\textnormal{\textbf{Unskolemising \textsf{LGQ}$_\textsf{nucl}$ clausal sets}}}

In this section, we unskolemise an \textsf{LGQ}$_\textsf{nucl}$ clausal set to a first-order formula without Skolem symbols. Two variations of the \textbf{Unsko} rule, respectively, are devised for inter-connected \textsf{LGQ}$_\textsf{nucl}$ clausal sets and flat \textsf{LGQ}$_\textsf{nucl}$ clausal sets.

An inter-connected \textsf{LGQ}$_\textsf{nucl}$ clausal set is unskolemised using: 
\begin{mdframed}[linewidth=2pt]

Suppose $N^\prime$ is an inter-connected \textsf{LGQ}$_\textsf{nucl}$ clausal set 
\begin{displaymath}
\Bigg\{
\begin{array}{l}
C_1(x_1, \ldots, x_n, f(x_1, \ldots, x_n), z_1, a), \\
\hfil  \ldots \\
C_n(x_1, \ldots, x_n, g(x_1, \ldots, x_n), z_t, b)	
\end{array}
\Bigg\},
\end{displaymath}
where $a$, $b$, $x_1, \ldots, x_n$ and $z_1, \ldots, z_t$ represent the Skolem constants, the non-Skolem constants and the variables introduced by the \textbf{Q-Rena} and \textbf{Q-Abs} procedures, respectively. Suppose $F$ is the Skolem-symbol-free first-order formula
\begin{displaymath}
\exists y\forall x_1 \ldots x_n \exists y_1\ldots y_m \forall z_1, \ldots, z_t
\left\lbrack
\begin{array}{l}
C_1(x_1, \ldots, x_n, y_1, z_1, y) \land \\
\hfil  \ldots \\
C_n(x_1, \ldots, x_n, y_m, z_t, b)	
\end{array}
\right\rbrack,
\end{displaymath}
where the variables $y, y_1, \ldots, y_m$ do not occur in $N^\prime$. 

Then, $N^\prime$ is unskolemised by the following rule:
\begin{displaymath}
 \prftree[l]{\textbf{UnSkI}: \ }
   {N \cup N^\prime}
   {N \cup \{F\}}.
\end{displaymath}
\end{mdframed}

A flat \textsf{LGQ}$_\textsf{nucl}$ clausal set is unskolemised using:
\begin{mdframed}[linewidth=2pt]
\begin{displaymath}
 \prftree[l]{\textbf{UnSkF}: \ }
   {N \cup \{C_1(x,a), \ldots, C_n(y,b)\}
}
{N \cup \{\exists z \forall xy (C_1(x,z) \land \ldots \land C_n(y,b))\}}
\end{displaymath}
if the following conditions are satisfied.
\begin{enumerate}[noitemsep]
\item $\{C_1(x,a), \ldots, C_n(y,b)\}$ is a flat \textsf{LGQ}$_\textsf{nucl}$ clausal set.
\item $a$ and $b$, respectively, represent the Skolem and the non-Skolem constants in $\{C_1(x,a), \ldots, C_n(y,b)\}$.
\item The variable $z$ does not occur in $\{C_1(x,a), \ldots, C_n(y,b)\}$. 
\end{enumerate}
\end{mdframed}

We use \textbf{Q-Unsko} to denote the procedure of applying the \textbf{UnSkI} and the \textbf{UnSkF} rules to an \textsf{LGQ}$_\textsf{nucl}$ clausal set. Using the \textsf{LGQ}$_\textsf{nucl}$ clausal set $N_2^{\prime\prime\prime} \cup N_3^\prime \cup N_3^{\prime\prime}$ as an example, we show what the \textbf{Q-Unsko} procedure does.
\begin{enumerate}
\item For inter-connected \textsf{LGQ}$_\textsf{nucl}$ clausal sets, the \textbf{UnSkI} rule is applied to them. Applying the \textbf{UnSkI} rule to $N_3^\prime$ and $N_3^{\prime\prime}$, respectively, derives
\begin{align*}
F_1 = \exists z^\prime \forall x y \exists x^\prime y^\prime 
\left\lbrack
\begin{array}{ll}
(\lnot G_1(x,y) \lor A_1(x^\prime,x) \lor A_2(y^\prime,x) \lor y \not \approx a) & \land \\
(\lnot G_2(x,y) \lor A_3(x^\prime,x) \lor A_4(y^\prime,x)) & \land\\
(\lnot G_3(x,y) \lor A_5(y^\prime,x) \lor x \not \approx z^\prime) &
\end{array}
\right\rbrack \ \text{and}
\\
F_2 = \forall y_1 z_1 x_1 \exists x_1^\prime
\left\lbrack
\begin{array}{l}
\lnot G_4(x_1, y_1, z_1) \lor A_6(x_1^\prime) \lor A_7(x_1^\prime) \lor y_1 \not \approx c \lor z_1 \not \approx y_1
\end{array}
\right\rbrack.
\end{align*}
\item For flat \textsf{LGQ}$_\textsf{nucl}$ clausal sets, the \textbf{UnSkF} rule is applied to them. Applying the \textbf{UnSkF} rule to $N_2^{\prime\prime\prime}$ unskolemise it into
\begin{align*}
F_3 = \forall x_6 x_7 x_8
\left\lbrack
\begin{array}{l}
\lnot B_1(x_8, x_6) \lor \lnot B_2(x_6, x_7) \lor \lnot B_3(x_7, x_8)
\end{array}
\right\rbrack.
\end{align*}
\item Conjunctively connect the outputting formulas of 1.~and 2. The running sample~$N$ is hence back-translated to a Skolem-symbol-free first-order formula $F_1 \land F_2 \land F_3$.
\end{enumerate}

\begin{lem}
\label{lem:btrans}
The back-translation defined by applying the \textbf{Q-Unsko} procedure to an \textsf{LGQ}$_\textsf{nucl}$ clausal set is a Skolem-symbol-free first-order formula (with equality).
\end{lem}
\begin{proof}
By Lemma \ref{lem:nucl}, Theorem \ref{thm:unsko} and the definition of the \textbf{Q-Unsko} procedure.	
\end{proof}

The result of our back-translation procedure is summarised as follows.
\begin{lem}
\label{lem:qrew_sound}
Let $N$ be an \textsf{LGQ} clausal set. Then, i) successively applying the \textbf{Q-Abs}, the \textbf{Q-Rena} and the \textbf{Q-Unsko} procedures to $N$ back-translates it into a Skolem-symbol-free first-order formula $F$, and ii) $F$ is logically equivalent to $N$.
\end{lem}
\begin{proof}
By ii) of Lemma~\ref{lem:un}, Lemmas \ref{lem:nucl} and \ref{lem:btrans}, $N$ is ensured to be back-translated to a Skolem-symbol-free first-order formula. That the \textbf{ConAbs} and the \textbf{VarAbs} rules are special cases of the \textbf{Abs} rule, the \textbf{VarRe} rule is a special case of the \textbf{Rena} rule, the \textbf{UnSkI} and the \textbf{UnSkF} rules are special cases of the \textbf{Unsko} rule and Lemma \ref{lem:sound_unsko} imply that $F$ and $N$ are logically equivalent.  
\end{proof}
\textbf{Figure~\ref{fig:qar}} summarises our back-translation procedure for the \textsf{LGQ} clausal class.

\begin{figure}[t]
\center
\includegraphics[width=.9\textwidth]{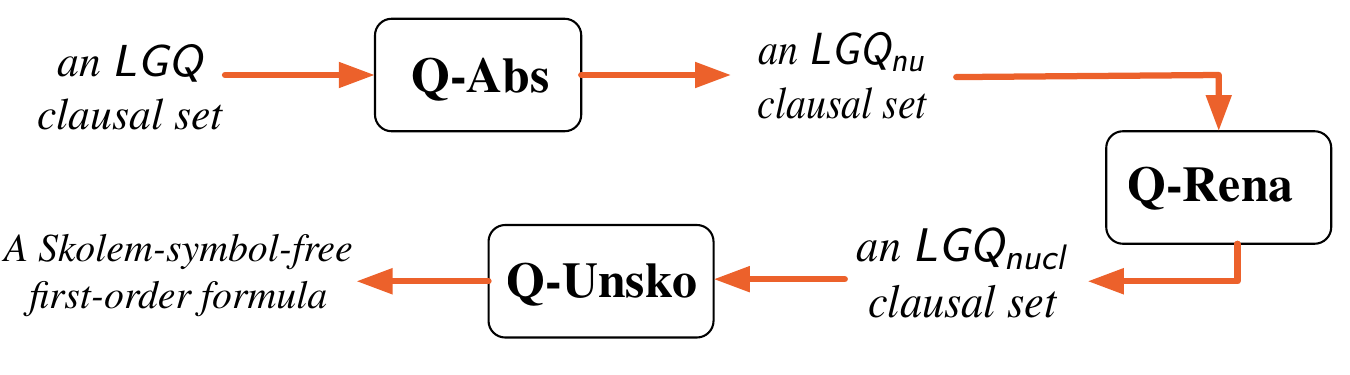}
\caption{The back-translation process for \textsf{LGQ} clausal sets}
\label{fig:qar}
\end{figure}

Returning to Question \ref{ques:qrew}, let a first-order formula $\Sigma_q$ be computed such that $D \models \Sigma_q$ if and only if $\Sigma \cup D \models q$. The final step in our procedure is to negate the first-order formula form of the saturation of $\Sigma \cup \{\lnot q\}$. In our example, we need negate $F_1 \land F_2 \land F_3$ to obtain as $\Sigma_q$:
\begin{align*}
&\forall z^\prime \exists x y \forall x^\prime y^\prime 
\left\lbrack
\begin{array}{ll}
(G_1(x,y) \land \lnot A_1(x^\prime,x) \land \lnot   A_2(y^\prime,x) \land y \approx a) & \lor \\
(G_2(x,y) \land \lnot A_3(x^\prime,x) \land \lnot A_4(y^\prime,x)) & \lor\\
(G_3(x,y) \land \lnot A_5(y^\prime,x) \land x \approx z^\prime) &
\end{array}
\right\rbrack \bigvee \\
&\exists y_1 z_1 x_1 \forall x_1^\prime
\left\lbrack
\begin{array}{l}
G_4(x_1, y_1, z_1) \land \lnot A_6(x_1^\prime) \land \lnot A_7(x_1^\prime) \land y_1 \approx c \land z_1 \approx y_1
\end{array}
\right\rbrack \lor \\
&\exists x_6x_7x_8 [
B_1(x_8, x_6) \land B_2(x_6, x_7) \land B_3(x_7, x_8)].
\end{align*}

Let $N$ be an \textsf{LGQ} clausal set. We use \textbf{Q-Rew} to denote the procedure of successively applying the \textbf{Q-Abs}, the \textbf{Q-Rena} and the \textbf{Q-Unsko} procedures to $N$, deriving a first-order formula $F$, and then negating $F$.  

Finally, we positively answer Question \ref{ques:qrew}.
\begin{thm}
\label{thm:rew-gf}
Suppose $\Sigma$ is a set of guarded quantification formulas, $D$ is a set of ground atoms and $q$ is a union of \textsf{BCQ}s. Further, suppose $N$ is a saturation obtained by applying \textbf{Q-Ans} to $\{\lnot q\} \cup \Sigma$. Then, applying the \textbf{Q-Rew} procedure to $N$ produces a Skolem-symbol-free first-order formula $\Sigma_q$ such that $\Sigma \cup D \models q$ if and only if $D \models \Sigma_q$.
\end{thm}
\begin{proof}
By Lemma \ref{lem:qrew_sound}.
\end{proof}
Comparing the signature in~$\Sigma_q$ and that in $\Sigma$ and $q$, $\Sigma_q$ may contain predicate and equality symbols not occurring in $q$ and $\Sigma$, since these symbols may have been introduced by the \textbf{Q-Ans} and the \textbf{Q-Abs} procedures, respectively.

\section{Related work}
\label{sec:relate}

\subsection*{\textbf{Resolution-based decision procedures}}
The basis of our \textsf{BCQ} answering and rewriting approaches is saturation-based resolution, which provides a practical and powerful method for developing decision procedures, as is evidenced in~\cite{dNdR03,GHMS98,K06,BGW93,H99,FATZ93,SH00,HS99}. 

The \textbf{P-Res} rule is inspired by the `partial replacement' strategy in \cite{BG97,BG01} and the `partial conclusion' of the `Ordered Hyper-Resolution with Selection' rule in~\cite{GdN99}. Even though \cite{GdN99} claims that the idea of `partial conclusion' can be easily generalised in the framework of~\cite{BG01}, it does not show how and no proof is provided. In~\cite{BG97} and \cite{BG01}, the `partial replacement' strategy seems to be what is behind `partial conclusions', and it is proved that for ground clauses the `partial replacement' strategy makes the application of a selection-based resolution rule, viz., the \textbf{S-Res} rule, redundant. In this paper, we formalise `partial replacement' in the \textbf{P-Res} system with the \textbf{P-Res} rule as the core rule. We have proved the system is generally sound and refutationally complete for full first-order clausal logic.


The \textbf{P-Res} rule adds high-level flexibility to the approach of an \textbf{S-Res} inference step, as one can choose any sub-multiset of the \textbf{S-Res} side premises as the \textbf{P-Res} side premises. This means that the \textbf{P-Res} rule gives us the option to choose a desirable resolvent from the possible `partial resolvents'. This technique is critical in our methods to querying for the guarded quantification fragments, allowing a choice of the `partial resolvent' that can be expressed in the same clausal class as the \textbf{P-Res} premises.

Motivated by the `MAXVAR' technique in \cite{dNdR03}, we devised the top-variable technique. The `MAXVAR' technique and the top-variable technique are also used in~\cite{GdN99} and~\cite{ZS20a}, respectively. A detailed example to demonstrate how the `MAXVAR' technique works is given in \cite{GdN99}, and the reader is referred to the manuscript \cite{dNdR03} for the formal definitions and proofs. \cite{dNdR03} uses the `MAXVAR' technique to avoid term depth increase in the resolvents of the loosely guarded clauses with nested compound terms. The presentation of the `MAXVAR' technique in \cite{dNdR03} is complicated: one needs to identify the depth of a sequence of variables, and then apply a specially devised unification algorithm to find `MAXVAR'. Moreover, the `MAXVAR' technique requires the use of non-liftable orderings, which are not compatible with the framework of~\cite{BG01}. 

We introduce the top-variable technique as a variation and simplification of the `MAXVAR' technique in the conference paper \cite{ZS20a}, which considers the \textsf{LG} clausal class \emph{with no nested compound terms}. The top-variable technique is generalised to apply to query clauses and already uses liftable orderings, so that it fits into the framework of \cite{BG01}. However, in \cite{ZS20a}, the pre-conditions of the top-variable technique, so-called query pairs, cannot be immediately applied in our general querying setting. 

Improving on \cite{ZS20a,GdN99,dNdR03}, in the present paper, we first give a clean approach to compute top variables, viz., the $\ComT$ function, and we then encode the top-variable technique in the $\TRes$ function, as given in Algorithm~\ref{algorithm:top}. We formally prove that the \textbf{T-Res} rule can be used in \emph{any} saturation-based resolution inference system following principles of the framework of~\cite{BG01}. We further generalise the premises of the \textbf{T-Res} rule to non-ground flat clauses and \textsf{LG} clauses, with detailed formal proofs given in Lemma \ref{lem:pres_gnd_pairing}, Corollary \ref{col:tres_uni} and Lemma~\ref{lem:tres_cons}. 

The \textbf{T-Res} system extends the resolution systems for the guarded fragment in~\cite{dNdR03,K06,GdN99} and the loosely guarded fragment in \cite{dNdR03,GdN99,ZS20a}. Although \cite{K06} is not interested in the loosely guarded fragment, it points out that the guarded clauses have the property that all compound terms have the same sequence of variables, i.e., the \emph{strongly compatible property}, which is an essential observation for our saturation-based rewriting procedure. Nonetheless, in \cite{K06}, this property is only used in analysing the complexity of its resolution decision procedure for the guarded fragment. \cite{GdN99} includes a discussion of refinement for the loosely guarded fragment, but does not give a formal description of the refinement or relevant proofs. A detailed refinement for the loosely guarded fragment is given in~\cite{dNdR03} with proofs, but \cite{dNdR03} uses non-liftable orderings, which are not compatible with the framework of \cite{BG01}. The resolution framework in \cite{BG01} provides a powerful system unifying many different resolution refinement strategies that exist in different forms, such as standard resolution, ordered resolution, hyper-resolution and selection-based resolution, and it provides vigorous simplification rules and redundancy elimination techniques, and forms the basis of the most state-of-the-art first-order theorem provers, such as SPASS \cite{WDFKSW09}, Vampire~\cite{RV01b}, E~\cite{S13}, and Zipperposition \cite{C15}. Our initial work in \cite{ZS20a} gives a resolution-based procedure in line with the resolution framework of \cite{BG01} for deciding satisfiability of \textsf{LGF} and querying for \textsf{LGF}, but only solves the \textsf{BCQ} answering problem for the Horn fragment of \textsf{LGF}.

In this paper, we formally define and thoroughly investigate partial resolution and the top-variable resolution techniques and develop detailed proofs. We then show that these techniques can be used and extended to decide satisfiability, \textsf{BCQ} answering and saturation-based \textsf{BCQ} rewriting for the guarded quantification fragments.

These are significant improvements and extensions over \cite{dNdR03,K06,GdN99,ZS20a}. Moreover, our methods provide the basis for \textsf{BCQ} answering and new saturation-based \textsf{BCQ} rewriting procedures for all the guarded quantification fragments.


\subsection*{\textbf{\textsf{BCQ} answering problem}}
The \emph{chase algorithms} \cite{CGP15}, which can be viewed as a form of \emph{forward chaining} \cite{RN20} or \emph{semantic tableau} \cite{H01}, is the state-of-the-art methods in solving \textsf{BCQ} answering problems in database and knowledge representation. These methods are applied on the ground data and $\Sigma$-rules in implication normal form. Unlike chase, our saturation-based query answering procedure does not require the grounding of clauses, which significantly reduces the number of clauses that need to be generated and handled. In our procedures, the inferences are performed differently, in particular, we are not limited to forward chaining and instead the $\Sigma$-clauses can be saturated first and then data can be added. Not only do our procedures avoid grounding, but they can simulate grounding by performing inferences on data first. 

The following \emph{ontology-based data access} \cite{ADDGMR08,CGL07,DFKMSSS08,HMAMSPMFKKSSFHWK08} scenario further motivates the \emph{saturation-based methods} to address query answering problems: given a set~$\Sigma$ of guarded quantification formulas, a \textsf{BCQ} $q$ and datasets~$D$, checking whether $\Sigma \cup D \models q$ is equivalent to checking unsatisfiability of $\{\lnot q\} \cup \Sigma \cup D$.
\begin{quote}
Suppose both $q$ and $\Sigma$ are fixed. We pre-saturate $\{\lnot q\} \cup \Sigma$ and use~$N$ to denote this pre-saturation. Then, independent of the datasets~$D$, the saturation~$N$ can be reused in checking satisfiability of $N \cup D$. This prevents having to recompute numerous  inferences of $\{\lnot q\} \cup \Sigma$ unnecessarily.
\end{quote}

Previous works investigate the \textsf{BCQ} answering problem for Datalog$^\pm$ \cite{CGL09} and description logics, such as guarded Datalog$^\pm$ rules \cite{CGP15,CGL12,CGF13} and fragments of the description logic~$\mathcal{ALCHOI}$ \cite{KKZ12,CGL07,MRC14,RA10}. Constraints in relational databases and ontological languages in knowledge bases are widely formalised in rules of Datalog$^\pm$, therefore devising automated querying procedures for Datalog$^\pm$ is important. 

A Datalog$^\pm$ rule is a first-order formula in the form 
\begin{align*}
F = \forall \overline x \overline y (\varphi(\overline x, \overline y) \to \exists \overline z \phi(\overline x, \overline z)),	
\end{align*}
where $\varphi(\overline x, \overline y)$ and $\phi(\overline x, \overline z)$ are conjunctions of atoms. Although answering \textsf{BCQ}s for Datalog$^\pm$ rules is undecidable \cite{BV81}, answering \textsf{BCQ}s for the guarded fragment of Datalog$^\pm$, viz., guarded Datalog$^\pm$ rules, is \textsc{2ExpTime}-complete \cite{CGF13}. The above Datalog$^\pm$ rule~$F$ is a guarded Datalog$^\pm$ rule if there exists an atom in $\varphi(\overline x, \overline y)$ that contains all free variables of $\exists \overline z \phi(\overline x, \overline z)$. Guarded Datalog$^\pm$ can be extended to the so-called loosely guarded and clique-guarded Datalog$^\pm$ by adopting the definition of the loosely guarded and the clique-guarded fragments, respectively. For example, 
\begin{align*}
\forall xyz (\textsf{Siblings}(x,y) \land \textsf{Siblings}(y,z) \land \textsf{Siblings}(z,x) \to \exists u (\textsf{Mother}(u,x,y,z)))	
\end{align*}
is a loosely guarded Datalog$^\pm$ rule. Guarded, loosely guarded and clique-guarded Datalog$^\pm$ rules can be seen as belonging to the Horn fragments of \textsf{GF}, \textsf{LGF} and \textsf{CGF}, respectively. Therefore our methods apply and lay the theoretical foundation for the first practical decision procedure of answering \textsf{BCQ}s for guarded, loosely guarded and clique-guarded Datalog$^\pm$ rules. Note that there are guarded Datalog$^\pm$ rules that are not expressible in \textsf{GF} \cite[page 103]{BBtC13}, however, the \textbf{Trans} process transforms these Datalog$^\pm$ rules into Horn guarded clauses.

The fragments of expressive description logic $\mathcal{ALCHOI}$ \cite{BHCS17} are prominent ontological languages in semantic web \cite{H08}. Query answering approaches for fragments of $\mathcal{ALCHOI}$ have been extensively studied in the literature \cite{KKZ12,CGL07,MRC14,RA10,B07}. A key technique in this area is transforming \textsf{BCQ}s into knowledge bases; see the rolling-up technique in \cite{T01} and the tuple graph technique in \cite{CGL98}. Interestingly, our \textbf{Q-Sep} procedure also achieves encoding of a query clause into the knowledge base of \textsf{LG} clauses. By the standard translation \cite[chapter 2]{BRV01}, axioms in the description logic $\mathcal{ALCHOI}$ can be translated into guarded formulas needing only unary and binary predicate symbols. Hence, our \textbf{Q-Ans} procedure can also be used as a practical decision procedure for \textsf{BCQ} answering for the expressive description logic $\mathcal{ALCHOI}$.

%

The \emph{squid decomposition technique} analyses the complexity for answering \textsf{BCQ}s over weakly guarded Datalog$^\pm$ \cite{CGF13}. In squid decompositions, a \textsf{BCQ} is regarded as a squid-like graph in which branches are `tentacles' and variable cycles are `heads'. Squid decomposition finds ground atoms that are complementary in the squid head, and then uses ground unit resolution to eliminate the heads. In contrast, our approach uses the separation rules to first cut `tentacles' and then uses the \textbf{T-Res} rule to resolve cycles in `heads'. Our approach produces compact saturations of \textsf{BCQ}s and the guarded quantification formulas, thus avoiding the significant overhead of grounding. 

\subsection*{\textbf{\textsf{BCQ} rewriting problem}}
Standard \textsf{BCQ} rewriting settings consider the following problem: given a union $q$ of \textsf{BCQ}s, a set $\Sigma$ of first-order formulas and a dataset $D$, can we produce (function-free) first-order formulas $\Sigma_q$, so that the entailment checking problem of $D \cup \Sigma \models q$ is reduced to the model checking problem of $D \models \Sigma_q$. If there exists such a $\Sigma_q$, $\Sigma$ and~$q$ are said to be \emph{first-order rewritable} \cite{CGL07}. Problems on the first-order rewritability property have been extensively studied in~\cite{CGL07,HLPW18,BBFSSTW10,TW20,TSCS15} for different description logics, and in \cite{GOP14,CGL12,HLPW18,BBLP18} for fragments of Datalog$^\pm$ rules. However, it is known that \textsf{BCQ} answering for none of the guarded quantification fragments are first-order rewritable. Another interesting saturation-based rewriting approach is~\cite{UBU07}, in which one first saturates axioms of the description logic $\mathcal{SHIQ}$, presenting the saturation as a set of disjunctive Datalog rules, and then  deductive databases are used to check entailment of \textsf{BCQ}s over the disjunctive Datalog rules. 

Unlike the idea of the first-order rewritability, saturation-based \textsf{BCQ} rewriting regards $D \models \Sigma_q$ as an entailment checking problem. Unlike \cite{UBU07}, in our query rewriting, queries are included in the reasoning process to obtain a saturation. Our saturation-based query rewriting is advantageous in ontology-based data access scenarios: Having a function-free first-order formula $\Sigma_q$ such that $D \cup \Sigma \models q$ if and only if $D \models \Sigma_q$, we can check~$\Sigma_q$ over different datasets $D_1, \ldots, D_n$. More importantly, to check whether $D_i \models \Sigma_q$, we can use reasoning methods other than resolution, e.g., the chase algorithm, as~$\Sigma_q$ is free of Skolem symbols. This combines different reasoning tools can potentially accelerate query answering processes. Moreover, devising this rewriting procedure is interesting and challenging in its own right, as it required a new investigation and new techniques to back-translate a first-order clausal set into a function-free first-order formula, which in general is an undecidable problem.
 

\section{Conclusion and Discussion}
\label{sec:conclu}
Considering the problem of query answering for the guarded quantification fragments, we present three sound and refutationally complete saturation-based resolution inference systems for general first-order clausal logic. Based on the top-variable inference system and customised separation rules, we establish the theoretical foundation for the first practical decision procedures of \textsf{BCQ} answering for the guarded, the loosely guarded, and the clique-guarded fragments. By extending the \textsf{BCQ} answering procedures with the back-translation techniques, we have devised a novel saturation-based \textsf{BCQ} rewriting procedure for these fragments.  


We are confident that our procedures provide a solid foundation for practical implementations. We claim the procedures can be implemented in any saturation-based theorem prover, as they are devised in line with the resolution framework in \cite{BG01}. Compared to the framework in \cite{BG01}, novel techniques are i) the \textbf{SepDeQ} and the \textbf{SepIndeQ} rules, ii) the \textbf{P-Res} and the \textbf{T-Res} rules and iii) the rules in the \textbf{Q-Rew} procedure.

i) Given a query clause $Q$, the application of the \textbf{SepDeQ} or the \textbf{SepIndeQ} rules to $Q$ can be implemented by the following steps.
\begin{enumerate}
	\item Find the surface literals in $Q$. By regarding each literal $L$ in $Q$ as a multiset in which the elements are the variable arguments of $L$, one can implement a multiset ordering $\succ_m$ for the literals in $Q$. The $\succ_m$-maximal literals in $Q$ are the surface literals in $Q$.
	\item Identify the separable surface literals in $Q$. Check whether two surface literals in~$Q$ have overlapping variables.
	\item Identify the separable subclauses in $Q$. Suppose $L_1$ and $L_2$ are two separable surface literals in $Q$. To separate $L_1$ from $Q$, one needs to find the literals in $Q$ that are  $\succ_m$-smaller than $L_1$, namely the literals guarded by $L_1$. The literals guarded by $L_1$ are a separable subclause in $Q$.
	\item Separate the subclause guarded by $L_1$ from $Q$. Following the conditions defined in the \textbf{SepDeQ} or the \textbf{SepIndeQ} rule, apply \emph{formula renaming with negative literals} to replace the literals guarded by $L_1$ by a fresh predicate symbol containing the only overlapping variables of $L_1$ and $L_2$. 
\end{enumerate}

ii) A possible implementation of the \textbf{P-Res} or the \textbf{T-Res} rule is: Suppose in a selection-based resolution (\textbf{S-Res}) inference, $C_1, \ldots, C_n$ are the side premises, and $C$ is the main premise with the negative literals $\lnot A_1, \ldots, \lnot A_n$ selected. Then, one can use the selection-based resolution (\textbf{S-Res}) to implement a \textbf{P-Res} or a \textbf{T-Res} resolvent of $C$ and $C_1, \ldots, C_n$ as follows.
\begin{enumerate}
\item Without deriving any resolvent, compute an mgu $\sigma^\prime$ between $C$ and $C_1, \ldots, C_n$. 
\item Unselect the literals $\lnot A_1, \ldots, \lnot A_n$ in $C$, and then select a sub-multiset $\lnot A_1, \ldots, \lnot A_m$ of $\lnot A_1, \ldots, \lnot A_n$ where $1 \leq m \leq n$, performing the \textbf{P-Res} rule on $C_1, \ldots, C_m$ and $C$ with $\lnot A_1, \ldots, \lnot A_m$ selected. For the case of the \textbf{T-Res} inference, $\lnot A_1, \ldots, \lnot A_m$ are the top-variable literals computed using the variable ordering $\succ_v$ and $\sigma^\prime$.
\item When the \textbf{P-Res} or the \textbf{T-Res} resolvent is derived, unselect $\lnot A_1, \ldots, \lnot A_m$.
\end{enumerate}

iii) The \textbf{Abs}, the \textbf{Rena} and the \textbf{Unsko} rules have been used in eliminating second-order quantifiers tasks, as implemented in the SCAN system \cite{O96}.

One next step is implementing the \textbf{Q-Ans} and the \textbf{Q-Rew} procedures and evaluating them on real-world ontologies. For example, we could focus on ontologies that are composed by the fragments of the description logic $\mathcal{ALCHOI}$ and guarded, loosely guarded and clique-guarded Datalog$^\pm$, since the number of \textsf{GF} problems in the TPTP first-order theorem proving benchmark \cite{S16} is rather small. 

Two other interesting questions for future work are: 1) Extend our saturation-based procedures to support the tasks of \textsf{BCQ} answering and saturation-based \textsf{BCQ} rewriting for the \emph{guarded negation} and the \emph{clique-guarded negation fragments} \cite{BtCS15}. This will require equality reasoning which we conjecture can be handled by extensions of the procedures presented in this paper with paramodulation or superposition. Whether our saturation-based methods can be refined to decide satisfiability of other variations of the guarded fragment such as the \emph{guarded fragment with transitive guards} \cite{ST04}, the \emph{triguarded fragment} \cite{RS18,KR21}, the \emph{two-variable guarded fragment with counting quantifiers} \cite{P07} and the \emph{forward guarded fragment} \cite{B21}, and querying for other guard-related fragments such as the \emph{monadic fragment of the two-variable guarded fragment with transitive guards} \cite{GPT13} and the \emph{forward guarded fragment} \cite{B21} remains to be investigated.

2) In our \textbf{Q-Rew} procedure, the rewritten queries are expressible in \textsf{LGF} and \textsf{BCQ}s, but with equality. It would be interesting to know whether in the setting of the saturation-based \textsf{BCQ} rewriting problem for the guarded quantification fragments with equality, one can translate the saturated clausal set back into \textsf{BCQ}s and formulas in these guarded quantification fragments with equality. The answer is not straightforward, as we first need to develop a decision procedure for the problem of the \textsf{BCQ} answering for these equality-occurring fragments.

\begin{acknowledgements}
We would like to thank the editor and the reviewers for the useful comments. Sen Zheng’s work is partially sponsored by the Great Britain-China Educational Trust.
\end{acknowledgements}

%
%

\bibliographystyle{spmpsci}      
\bibliography{reference}   


\end{document}